\documentclass{LMCS}

\def\doi{9(3:01)2013}
\lmcsheading%
{\doi}
{1--44}
{}
{}
{Nov.~15, 2012}
{Jul.~\phantom04, 2013}
{}

\pdfoutput=1
\usepackage{xspace,graphicx,hyperref,enumerate}
\usepackage{amssymb,amsmath,amsthm}

\newcommand{\OMIT}[1]{}

\newcommand{\N}{\mathbb{N}}
\newcommand{\cal}[1]{\mathcal{#1}}
\newcommand{\A}{{\cal A}}

\newcommand{\R}{{\cal R}}
\newcommand{\LL}{{\cal L}}
\newcommand{\CC}{{\cal C}}

\renewcommand{\SS}{{\cal S}}

\newcommand{\K}{{\cal K}}

\renewcommand{\l}{\ell} 
\renewcommand{\phi}{\varphi}
\newcommand{\fth}{\hfill $\Box$}

\newtheorem{theorem}{Theorem}[section]

\newtheorem{claim}{Claim}[theorem]
\newtheorem{observation}[theorem]{Observation}

\newcommand{\PSPACE}{{\sc PSpace}}
\newcommand{\pspace}{\PSPACE}
\newcommand{\ptime}{{\sc Ptime}}
\newcommand{\nlog}{{\sc NLogSpace}}
\newcommand{\dlog}{{\sc DLogSpace}}
\newcommand{\nexp}{{\sc NExptime}}

\newcommand{\REC}{{\sf REC}}
\newcommand{\REG}{{\sf REG}}
\newcommand{\NFA}{NFA\xspace}
\newcommand{\lang}[1]{\mathcal{L}(#1)}
\newcommand{\RAT}{{\sf RAT}}
\newcommand{\SCR}{{\sf SCR}}
\newcommand{\Reg}{\REG} %regular relations
\newcommand{\Regk}[1]{\REG_{#1}} %k-ary regular relations
\newcommand{\Rat}{\RAT} %rational relations
\newcommand{\sigmas}{\Sigma^*}
\newcommand{\sigmabot}{\Sigma_\bot}
\newcommand{\gammabot}{\Gamma_\bot}
\newcommand{\sigbot}{\sigmabot}
\newcommand{\e}{\varepsilon}
\newcommand{\pref}{\preceq_{{\rm pref}}}

\newcommand{\suff}{\preceq_{{\rm suff}}}
\newcommand{\subw}{\preceq}
\newcommand{\subseq}{\sqsubseteq}
\newcommand{\set}[1]{\{#1\}}

\newcommand{\cum}{\textit{cum}\xspace}

\newcommand{\INTPI}[3]{({#1} \mathrel{\cap_{#3}} {#2}) \stackrel{\text{\tiny?}}{=}\emptyset}
\newcommand{\INTP}[2]{({#1} \cap {#2})\ensuremath{\stackrel{\text{\tiny ?}}{=}}\emptyset}

\newcommand{\INT}[2]{({#1} \cap {#2}) \stackrel{\text{\tiny ?}}{=}\emptyset}

\newcommand{\GENINT}{\text{\sc GenInt}}
\newcommand{\ecrpq}{{\sc ECRPQ}} 
\newcommand{\crpq}{{\sc CRPQ}}

\begin{document}

\title[Graph Logics with Rational Relations]{Graph Logics with Rational Relations\rsuper*}

\author[P.~Barcel{\'o}]{Pablo Barcel{\'o}\rsuper a}	%required
\address{{\lsuper a}Department of Computer Science, 
University of Chile}	%required
\email{pbarcelo@dcc.uchile.cl}  %optional

\author[D.~Figueira]{Diego Figueira\rsuper b}	%optional
\address{{\lsuper{b,c}}Laboratory for Foundations of Computer Science, 
University of Edinburgh}	%optional
\email{\{dfigueir, libkin\}@inf.ed.ac.uk}  %optional

\author[L.~Libkin]{Leonid Libkin\rsuper c}	%optional
%\address{Laboratory for Foundations of Computer Science, 
%University of Edinburgh}	%optional
%\email{libkin@inf.ed.ac.uk}  %optional

\thanks{{\lsuper{a,b,c}}Partial support provided by Fondecyt grant
  1110171 for Barcel\'o and EPSRC grants G049165 and J015377 for
  Figueira and Libkin.} %optional

%% mandatory lists of keywords and classifications:

\ACMCCS{[{\bf Theory of computation}]: Formal languages and automata
  theory; Theory and algorithms for application domains---Database
  theory---Database query languages (principles); Design and analysis
  of algorithms}

\titlecomment{{\lsuper*}This is the full version of the conference paper \cite{bfl2012}.}
%%%%%%%%%%%%%%%%%%%%%%%%%%%%%%%%%%%%%%%%%%%%%%%%%%%%%%%%%%%%%%%%%%%%%%%%%%%

%% the abstract has to PRECEED the command \maketitle:
%% be sure not to issue the \maketitle command twice!
\begin{abstract}
We investigate some basic questions about the interaction of regular
and rational relations on words.  The primary motivation comes from
the study of logics for querying graph topology, which have recently
found numerous applications. Such logics use conditions on paths
expressed by regular languages and relations, but they often need to
be extended by rational relations such as subword or
subsequence. Evaluating formulae in such extended graph logics boils
down to checking nonemptiness of the intersection of rational
relations with regular or recognizable relations (or, more generally,
to the generalized intersection problem, asking whether some
projections of a regular relation have a nonempty intersection with a
given rational relation).

We prove that for several basic and commonly used rational relations,
the intersection problem with regular relations is either undecidable
(e.g., for subword or suffix, and some generalizations), or decidable
with non-primitive-recursive complexity (e.g., for subsequence
and its generalizations). These results are used to rule out many
classes of graph logics that freely combine regular and rational
relations, as well as to provide the simplest problem related to
verifying lossy channel systems that has non-primitive-recursive
complexity.  We then prove a dichotomy result for logics combining
regular conditions on individual paths and rational relations on
paths, by showing that the syntactic form of formulae classifies them
into either efficiently checkable or undecidable cases. We also give
examples of rational relations for which such logics are decidable
even without syntactic restrictions. 
\end{abstract}

\maketitle

\section{Introduction}
\label{sec:intro}
The motivation for the problems investigated in this paper comes from the
study of logics for querying graphs. Such logics form the basis of
query languages for graph databases, that have recently found
numerous applications in areas including biological networks, social
networks, Semantic Web, crime detection, etc.\ (see \cite{AG-survey}
for a survey) and led to multiple systems and prototypes. In such
applications, data is usually represented as a labeled graph. For
instance, in social networks, people are nodes, and labeled edges
represent different types of relationship between them;  in RDF -- the
underlying data model of the Semantic Web -- data is modeled as a graph,
with RDF triples naturally representing labeled edges.

The questions that we address are related to the interaction of
various classes of relations on words, for instance, rational
relations (examples of those include subword and subsequence) or
regular relations (such as prefix, or equality of words). An example
of a question we are interested in is as follows: is it decidable
whether a given regular relation contains a pair $(w,w')$ so that
$w$ is a subword/subsequence of $w'$?  Problems like this are very
basic and deserve a study on their own, but they are also
necessary to answer questions on the power and complexity of querying
graph databases. We now explain how they arise in that setting.

Logical languages for querying graph data have been developed since the
late 1980s (and some of them became precursors of languages later used
for XML).  They query the {topology}
of the graph, often leaving querying data that might be stored in the
nodes to a standard database engine. Such logics are quite
different in their nature and applications from another class of
graph logics based on spatial calculi  \cite{CGG02,DGG07}.
Their formulae combine various
reachability patterns. The simplest form is known as   
{\em regular path queries (RPQs)} \cite{CMW87,CM90}; they check the
existence of a path whose label belongs to a regular language.
Those are typically used as atoms and then closed under
conjunction and existential quantification, resulting in the class of 
{\em conjunctive regular path queries (\crpq s)}, which have
been the subject of much investigation~\cite{CGLV00,DT01,FLS98}. For
instance, a \crpq\ may ask for a node $v$ such that there 
exist nodes $v_1$ and $v_2$ and paths from $v$ to $v_i$ with the label
in  a regular language $L_i$, for $i=1,2$.

The expressiveness of these queries, however, became insufficient in
applications such as the Semantic Web or biological
networks due to their inability to {\em compare} paths. For instance,
it is a common requirement in RDF languages to 
 compare paths based on specific semantic
associations \cite{AS03}; biological sequences often need to
be compared for similarity, based, for example, on the edit distance.
%\cite{Gusfield}. 

To address this, an extension of \crpq s with relations on paths was
proposed \cite{pods10}. It used {\em regular} relations on
paths, i.e., relations given by synchronized automata
\cite{oldstuff,frenchreinvention}. Equivalently, these are the
relations definable in automatic structures on words
\cite{jacm2003,graedel2000,bruyere}. They include
prefix, equality, equal length of words, or fixed edit
distance between words. The extension of \crpq s with them, called \ecrpq s, was
shown to have acceptable complexity (\nlog\ with respect to
data, \pspace\ with respect to query).

However, the expressive power of \ecrpq s is 
still short of the expressiveness needed in many
applications. For instance, semantic associations between paths used in
RDF applications often deal with {subwords} or {subsequences},
but these relations are {\em not} regular. They are {\em rational}:
they are still accepted by automata, but those whose 
heads move asynchronously. Adding them to a query language must be
done with extreme care: simply replacing regular relations with
rational in the definition of \ecrpq s makes query evaluation
undecidable!

So we set out to investigate the following problem: given a class of
graph queries, e.g., \crpq s or \ecrpq s, what happens if one adds the
ability to test whether pairs of paths belong to a rational relation $S$,
such as subword or subsequence? We start by observing that this
problem is a generalization of the {\em intersection problem}:
given a regular relation
$R$, and a rational relation $S$, is $R\cap S\neq\emptyset$?  It is
well known that there exist rational relations $S$ for which it is
undecidable \cite{berstel}; however, we are not interested 
in artificial relations obtained by encoding PCP instances, but rather
in very concrete relations used in querying graph data.

The intersection problem captures the essence of graph logics 
\ecrpq s and \crpq s (for the latter, when restricted to the class of
recognizable relations \cite{berstel,choffrut-bulletin}). In fact,
query evaluation can be cast as the {\em generalized intersection
  problem}. Its input includes an $m$-ary %recognizable or 
regular
relation $R$, a binary rational relation $S$, and a set $I$ of pairs
from $\{1,\ldots,m\}$. It asks whether there is a tuple
$(w_1,\ldots,w_m)\in R$ so that $(w_i,w_j)\in S$ whenever $(i,j)\in
I$. For $m=2$ and $I=\{(1,2)\}$, this is the usual intersection
problem.

Another motivation for looking at these basic problems
comes from verification of lossy channel systems (finite-state
processes that communicate over unbounded, but lossy, FIFO
channels). Their reachability problem is known to be
decidable, although the complexity is not bounded by any
multiply-recursive function \cite{CS-lics08}. In fact, a
``canonical'' problem used in reductions showing this  enormous
complexity \cite{CS-fsttcs07,CS-lics08} can be restated
as follows: given a binary rational relation $R$, does it have a pair
$(w,w')$ so that $w$ is a subsequence of $w'$? This naturally leads to
the question whether the same bounds hold for the simpler instance of
the intersection problem when we use {regular} relations instead of
rational ones.  We actually show that this is true.

\smallskip
\subsubsection*{\underline{Summary of results}}
We start by showing that evaluating \crpq s and \ecrpq s extended with
a rational relation $S$ can be cast as the generalized intersection
problem for $S$ with recognizable and regular relations
respectively. Moreover, the complexity of the basic intersection
problem is a lower bound for the complexity of query
evaluation.

We then study the complexity of the intersection problem for fixed
relations $S$. For recognizable relations, it is well known to
be efficiently decidable for every rational $S$. For regular
relations, we show that if $S$ 
is the subword, or the suffix relation, then the problem is
undecidable. That is, it is undecidable to check, given a binary
regular relation $R$, whether it contains a pair $(w,w')$ so that $w$
is a subword of $w'$, or even a suffix of $w'$. We also present a 
generalization of this result. 

The analogous problem for the subsequence relation is known to be
decidable, and, if the input is a rational relation $R$, then the
complexity is non-multiply-recursive \cite{CS-fsttcs07}. We extend
this in two ways. First, we show that the lower bound remains true even for
regular relations $R$. Second, we extend decidability to the class of
all rational relations for which one projection is closed under
subsequence (the subsequence relation itself is trivially such,
obtained by closing the first projection of the equality relation).

In addition to establishing some basic facts about classes of
relations on words, these results tell us about the infeasibility of
adding rational relations to \ecrpq s: in fact adding subword makes
query evaluation undecidable, and while it remains decidable with
subsequence, the complexity is prohibitively high. 

So we then turn to the generalized intersection problem with
recognizable relations, corresponding to the evaluation of \crpq s
with an extra relation $S$. We show that the shape of the relation $I$
holds the key to decidability. If its underlying undirected graph is
acyclic, then the problem is decidable in \pspace\ for every rational
relation $S$ (and for a fixed formula the complexity drops to \nlog).
In the cyclic case, the problem is undecidable 
for some rational relation $S$. For relations generalizing
subsequence, we have decidability when $I$ is a DAG, and for
subsequence itself, as well as for suffix, query evaluation is
decidable regardless of the shape of \crpq s. 

Thus, under the mild syntactic restriction of
acyclicity of comparisons with respect to rational relations,
such relations can be added to the common class \crpq\ of
graph queries, without incurring a high complexity cost. 

\medskip
\subsubsection*{\underline{Organization}}
We give basic definitions in Section \ref{sec:prelim} and define the
main problems we study in Section \ref{sec:gip}. Section \ref{gl:sec}
introduces graph logics and establishes their connection with the
(generalized) intersection problem. Section \ref{intp:sec} studies
decidable and undecidable cases of the intersection problem. Section
\ref{restr:sec} looks at the case of recognizable relations and \crpq
s and establishes decidability results based on the intersection
pattern. 

%Complete proofs of all results are available in the
%full version of the paper. 

%%% Local Variables: 
%%% mode: latex
%%% TeX-master: "lics12"
%%% End: 

\section{Preliminaries} 
\label{sec:prelim}
\newcommand{\wproj}{\pi}
Let $\N = \set{1,2,\dotsc}$, $[i..j] = \set{i, i+1, \dotsc, j}$ (if $i > j$, $[i..j]=\emptyset$), $[i] = [1..i]$. Given, $A,B \subseteq \N$, an \emph{increasing} function  $f : A \to B$ is one such that $f(i) \geq f(j)$ whenever $i>j$. If $f(i) > f(j)$ we call it \emph{strictly} increasing.

\subsubsection*{\underline{Alphabets, languages, and morphisms}} 

We shall use letters $\Sigma$, $\Gamma$
%, $\Delta$, with sub- and superscripts, 
to denote finite alphabets. The set of all finite words
over an alphabet $\Sigma$ is denoted by $\sigmas$. We write $\e$ for
the empty word, $w\cdot w'$ for the concatenation of two words, and
$|w|$ for the length of a  word $w$. Given a word $w \in \Sigma^*$, $w[i..j]$ stands for the substring in positions $[i..j]$, $w[i]$ for $w[i..i]$, and $w[i..]$ for $w[i..|w|]$. Positions in the word start with $1$.

If $w=w' \cdot u \cdot w''$, then 
\begin{iteMize}{$\bullet$}
\item $u$ is a {\em subword} of $w$
(also called \emph{factor} in the literature, written as $u \subw w$),  
\item $w'$ is a {\em prefix} of $w$ (written as $w'
\pref w$), and 
\item $w''$ is a {\em suffix} of $w$ (written as $w'' \suff w$). 
\end{iteMize}

We say that $w'$ is a {\em subsequence} of $w$ (also called
\emph{subword embedding} or \emph{scattered subword} in the literature, written as $w' \subseq 
w$) if $w'$ is obtained by removing some letters (perhaps none) from
$w$, i.e., $w=a_1\ldots a_{n}$, and $w'=a_{i_1}a_{i_2}\ldots
a_{i_k}$, where $1 \leq i_1 < i_2 < \ldots < i_k\leq n$. 

If $\Sigma \subset \Gamma$ and $w\in\Gamma^*$, then by $w_\Sigma$ we
denote the projection of $w$ on $\Sigma$. That is, if $w=a_1\ldots
a_n$ and $a_{i_1},\ldots, a_{i_k}$ are precisely the letters from
$\Sigma$, with $i_1 < \ldots < i_k$, then $w_\Sigma=a_{i_1}\ldots
a_{i_k}$. 
%Using this concept, one can characterize rational relations as relations of the form $\{(w_\Sigma, u_\Sigma) \ | \ (w,u) \in R\}$, where $R$ ranges over $\REG_2$ over the alphabet $\sigmabot$. 

\smallskip

Recall that a {\em monoid} $M=\langle U,\cdot,1\rangle$ has an
associative 
binary operation $\cdot$ and a neutral element $1$ satisfying
$1x=x1=x$ for all $x$ (we often write $xy$ for $x\cdot y$).
The set $\Sigma^*$ with the operation of concatenation and 
the neutral element $\e$ forms a {monoid} 
$\langle \Sigma^*,\cdot,\e\rangle$, 
the free monoid generated by $\Sigma$.
A  function $f: M\to M'$ between two monoids 
is a {\em morphism} if it sends the neutral element of $M$ to the
neutral element of $M'$, and if $f(xy)=f(x)f(y)$ for all
$x,y\in M$. Every morphism
$f: \langle \Sigma^*,\cdot,\e\rangle \to M$ is uniquely determined by
the values $f(a)$, for $a\in \Sigma$, as $f(a_1\ldots a_n)=f(a_1)\cdots
f(a_n)$. A morphism $f: \langle \Sigma^*,\cdot,\e\rangle \to \langle
\Gamma^*,\cdot,\e\rangle$ is called {\em alphabetic} if $f(a)\in
\Gamma\cup\{\e\}$, and {\em strictly alphabetic} if $f(a)\in \Gamma$
for each $a\in \Sigma$, see  \cite{berstel}. 

A language $L$ is a subset of $\Sigma^*$, 
for some finite alphabet $\Sigma$. 
It is {\em recognizable} if there is a
finite monoid $M$, a morphism $f: \langle \Sigma^*,\cdot,\e\rangle \to
M$, and a subset $M_0$ of $M$ such that $L=f^{-1}(M_0)$. 

A language $L$ is {\em regular} if there exists an NFA
(non-deterministic finite 
automaton) $\A=\langle Q,\Sigma,q_0,\delta,F\rangle$ such that
$L=\lang \A$, the
language of words accepted by $\A$. We use the standard notation for
NFAs, where $Q$ is the set of states, $q_0$ is the initial state, $F$
is the set of final states, and $\delta\subseteq Q\times\Sigma\times
Q$ is the transition relation.  

A language is {\em rational} if it is denoted by a regular expression; such
expressions are built from $\emptyset$, $\e$, and alphabet letters by
using operations of concatenation ($e\cdot e'$), union ($e \cup e'$),
and Kleene star ($e^*$). It is of course a classical result of
formal language theory that the classes of recognizable, regular, and
rational languages coincide.

\medskip
\subsubsection*{\underline{Recognizable, regular, and rational relations}}

While the notions of recognizability, 
regularity, and rationality
coincide over
languages $L\subseteq \Sigma^*$, they differ over 
{relations} over $\Sigma$, i.e., subsets of
$\Sigma^* \times \ldots \times \Sigma^*$. We now define
those (see
\cite{berstel,carton,choffrut-bulletin,oldstuff,frenchreinvention,thomas92}). 

Since $\langle \Sigma^*,\cdot,\e\rangle$ is a 
%the free $\Sigma$-generated 
monoid, 
the product $(\Sigma^*)^n$ has the structure of a monoid too. We can
thus define {\em recognizable $n$-ary relations} over $\Sigma$ as
subsets $R\subseteq (\Sigma^*)^n$ so that there exists a finite monoid
$M$ and a morphism $f: (\sigmas)^n\to M$ such that $R=f^{-1}(M_0)$ for
some $M_0\subseteq M$.  The class of $n$-ary recognizable relations
will be denoted by $\REC_n$; when $n$ is clear or irrelevant, we write
just $\REC$.

It is well-known that a relation $R \subseteq (\sigmas)^n$ is in
$\REC_n$ iff it is a finite union of the sets of the form
$L_1\times\ldots\times L_n$, where each $L_i$ is a regular language
over $\Sigma$, see \cite{berstel,oldstuff}. 

Next, we define the class of regular relations. Let
$\bot\not\in\Sigma$ be a new alphabet letter, and let $\sigmabot$ be 
$\Sigma\cup\{\bot\}$. Each tuple $\bar w=(w_1,\ldots,w_n)$ of words from
$\sigmas$ can be viewed as a word over $\sigmabot^n$ as follows: pad
words $w_i$ with $\bot$ so that they all are of the same length, and
use as the $k$th symbol of the new word the $n$-tuple of the $k$th
symbols of the padded words. Formally, let $\l=\max_i |w_i|$. Then
$w_1\otimes \ldots \otimes w_n$ is a word of length $\l$ whose $k$th
symbol is $(a_1,\ldots,a_n)\in \sigmabot^n$ such that 
$$a_i = \begin{cases} \text{the }k\text{th letter of }w_i & \text{ if
  }|w_i| \geq k \\ \bot & \text{ otherwise.}\end{cases}
$$ 
We shall also write $\otimes\bar w$ for $w_1\otimes \ldots \otimes
w_n$. We define  $\wproj_i (u_1 \otimes \dotsb \otimes u_k) = u_i$ for all $i \in [k]$.
A relation $R\subseteq (\Sigma^*)^n$ is called a {\em regular $n$-ary
  relation} over $\Sigma$ if there is a finite automaton $\A$ over
$\sigbot^n$ that accepts $\{\otimes\bar w\ | \ \bar w\in R\}$. 
 The class of $n$-ary regular relations
is denoted by $\REG_n$; as before, we write $\REG$ when $n$ is clear
or irrelevant. 

Finally, we define rational relations. There are two equivalent ways
of doing it. One uses regular expressions, which are now built from
tuples $\bar a \in (\Sigma\cup\{\e\})^n$ using the same operations of
union, concatenation, and Kleene star. Binary relations $\suff$,
$\subw$, and $\subseq$ are all rational: the expression 
$\big(\bigcup_{a\in\Sigma}(\e,a)\big)^*\cdot
\big(\bigcup_{a\in\Sigma}(a,a)\big)^*$ defines 
$\suff$, the expression $\big(\bigcup_{a\in\Sigma}(\e,a)\big)^*\cdot
\big(\bigcup_{a\in\Sigma}(a,a)\big)^* \cdot
\big(\bigcup_{a\in\Sigma}(\e,a)\big)^*$ defines $\subw$, and the
expression $\big(\bigcup_{a\in\Sigma}(\e,a) \cup (a,a)\big)^*$ defines
$\subseq$. 
%put back later
%The ternary concatenation relation is rational too, defined
%by $\bigcup_{a\in\Sigma}(a,\e,a)\cdot
%\bigcup_{a\in\Sigma}(\e,a,a)$.

Alternatively, $n$-ary rational relations can be defined by means of $n$-tape
automata, that 
have $n$ heads for the tapes and one additional control; at every
step, based on the state and the letters it is reading, the automaton
can enter a new state and move some (but not necessarily all) tape
heads. The classes of $n$-ary relations so defined are called {\em
  rational $n$-ary relations}; we use the notation $\RAT_n$ or just
$\RAT$, as before. 

\medskip
\subsubsection*{\underline{Relationships between classes of relations}}
While it is well known that $\REC_1=\REG_1=\RAT_1$, we have strict
inclusions $$\REC_k \ \subsetneq\ \REG_k\ 
\subsetneq\ \RAT_k$$ for every $k>1$ (see for example \cite{berstel}). For instance, ${\pref} \in \REG_2 -
\REC_2$ and ${\suff}\in\RAT_2-\REG_2$. 

The classes of recognizable and regular relations are closed under
intersection; however the class of rational relations is not. In fact,
one can find $R\in\REG_2$ and $S\in\RAT_2$ so that $R\cap
S\not\in\RAT_2$. However, if $R\in\REC_m$ and $S\in\RAT_m$, then
$R\cap S\in\RAT_m$.

Binary rational relations can be characterized as follows
\cite{berstel,nivat68}. A relation $R\subseteq \Sigma^*\times\Sigma^*$
is rational iff there is a  
finite alphabet $\Gamma$, a regular language $L\subseteq \Gamma^*$ and
two alphabetic
morphisms $f,g: \Gamma^*\to\Sigma^*$ such that $R=\{(f(w),g(w)) \ |
\ w\in L\}$. If we require $f$ and $g$ to be strictly alphabetic
morphisms,  we get the class of {\em length-preserving} regular
relations, i.e., $R\in\REG_2$ so that $(w,w')\in R$ implies
$|w|=|w'|$. Regular binary relations are then finite
unions of relations of the form $\{(w\cdot u,w')\ | \ (w,w')\in R,
\ u\in L\}$ and $\{(w,w'\cdot u)\ | \ (w,w')\in R,
\ u\in L\}$, where $R$ ranges over length-preserving regular
relations, and $L$ over regular languages. 

%%% PUT IT BACK FOR FULL VERSION
\OMIT{
If $\Sigma \subset \Gamma$ and $w\in\Gamma^*$, then by $w_\Sigma$ we
denote the projection of $w$ on $\Sigma$. That is, if $w=a_1\ldots
a_n$ and $a_{i_1},\ldots, a_{i_k}$ are precisely the letters from
$\Sigma$, with $i_1 < \ldots < i_k$, then $w_\Sigma=a_{i_1}\ldots
a_{i_k}$. Using this, one can characterize rational relations as
relations of the form $\{(w_\Sigma, u_\Sigma) \ | \ (w,u) \in R\}$,
where $R$ ranges over $\REG_2$ over the alphabet $\sigmabot$. }

\medskip
\subsubsection*{\underline{Properties of classes of relations}} Since
relations in 
$\REC$ and $\REG$ are given by NFAs, they inherit all the
closure/decidability properties of regular languages. If $R\in\RAT$,
then each of its projections is a regular language, and can be
effectively constructed (e.g., from the description of $R$ as an
$n$-tape automaton). Hence, the nonemptiness problem is decidable for
rational relations. However, testing nonemptiness of the intersection
of two rational relations is undecidable \cite{berstel}. Also, for 
$R, R' \in \RAT$, the following are undecidable:
checking whether $R\subseteq R'$ or $R=R'$, universality
($R=\sigmas\times\sigmas$), and checking whether $R\in\REG$ or
$R\in\REC$ \cite{berstel,carton,lisovik}. 

\smallskip
\subsubsection*{Remark} We defined recognizable,
regular, and rational relations over the same alphabet, i.e., as
subsets of $(\sigmas)^n$. Of course it is possible to define them as
subsets of $\Sigma_1 \times \ldots \times \Sigma_n$, with the
$\Sigma_i$'s not necessarily distinct. Technically, there are no
differences and all the results will continue to hold. Indeed, one can
simply consider a new alphabet $\Sigma$ as the disjoint union of 
$\Sigma_i$'s, and enforce the condition that the $i$th projection only
use the letters from $\Sigma_i$ (this is possible for all the classes
of relations we consider). In fact, in the proofs we shall be using both
types of relations.

\medskip

\subsubsection*{\underline{Well-quasi-orders}}
  A well-quasi-order ${\leq} \subseteq A \times A$ is a reflexive and transitive relation such that for every infinite sequence $(a_i)_{i \in \N}$ over $A$ there are $i<j$ with $a_i \leq a_j$. We will make use of the following two lemmas.

\begin{lem}[Higman's Lemma \cite{higmanslem}]
For every alphabet $\Sigma$, the subsequence relation ${\subseq} \subseteq \Sigma^* \times \Sigma^*$ is a well quasi-order.
\end{lem}
\begin{lem}[Dickson's Lemma \cite{dicksonslem}]
  For every well-quasi-order ${\leq} \subseteq A \times A$, the product order ${\leq^k} \subseteq A^k \times A^k$ (where $(a_1, \dotsc, a_k) \leq^k (a'_1, \dotsc, a'_k)$ if{f} $a_i \leq a'_i$ for all $i\in [k]$) is a well-quasi-order.
\end{lem}

%%% Local Variables: 
%%% mode: latex
%%% TeX-master: "lrr"
%%% End: 

\section{Generalized intersection problem}
\label{sec:gip}

We now formalize the main technical problem we study.
Let $\R$ be a class of relations over $\Sigma$, and $\SS$ a
class of binary relations over $\Sigma$. We use the notation
$[m]$ for $\{1,\ldots,m\}$. 
If $R$ is an $m$-ary relation, $S$ is a binary relation, and $I
\subseteq [m]^2$, we write $R \cap_I S$ for the set of tuples
$(w_1,\ldots, w_m)$ in $R$ such that $(w_i,w_j)\in S$ whenever
$(i,j)\in I$.

The {\em generalized
intersection problem} $\INTPI \R\SS I$ is defined as:

\medskip

\centerline{
\fbox{\!\!\!
\begin{tabular}{ll}
{\sc Problem:} & $\INTPI \R\SS I$\\
{\sc Input:} & an $m$-ary relation $R \in \R$, \\
            & a relation $S\in \SS$, and $I\subseteq [m]^2$\\
{\sc Question:} & is $R \cap_I S \neq \emptyset$?
\end{tabular}\!\!\!
}
}

\medskip

If $\SS=\{S\}$, we write $S$ instead of $\{S\}$. 
We write $\GENINT_S(\R)$ for the class of all problems 
$\INTPI \R S I$ where $S$ is fixed, i.e., the input consists of
$R\in\R$ and $I$. 
\OMIT{
We write $\GENINT(\R,\SS)$ 
for the class of all problems of the form 
 $\INTPI {\R_m}  S I$ where $I\subseteq [m]^2$ and  $\R_m$ is the
class of $m$-ary relations in $\R$. Most often we refer to such
problems when $\R$ is $\REC$ or $\REG$ or $\RAT$. For a fixed relation
$S$, we shall write $\GENINT_S(\R)$ for $\GENINT(\R,S)$; in that case,
the only input to the problem is $R\in\R$. 
}
As was explained in the introduction, this problem captures the
essence of evaluating queries in various graph logics, e.g., \crpq s
or \ecrpq s extended with rational relations $S$. The classes $\R$ will
typically be $\REC$ and $\REG$.

If $m=2$ and $I=\{(1,2)\}$, the generalized intersection problem
becomes simply  the {\em intersection problem} for
the classes $\R$ and $\SS$ of binary relations:

\medskip
\centerline{
\fbox{
\begin{tabular}{ll}
{\sc Problem:} & $\INT \R \SS$\\
{\sc Input:} & $R \in \R$ and $S\in \SS$\\
{\sc Question:} & is $R\cap S\neq \emptyset$?
\end{tabular}\!\!\!
}
}

\medskip

The problem $\INTP \REC S$ is decidable for every rational relation
$S$, simply by constructing $R\cap S$, which is a rational relation,
and testing its nonemptiness.  However, 
$\INTP \REG S$ could already be undecidable (we shall give one
particularly simple example later).

%%% Local Variables: 
%%% mode: latex
%%% TeX-master: "lrr"
%%% End: 

\section{Graph logics and the generalized intersection problem}
\label{gl:sec}

\newcommand{\Lto}[1]{\stackrel{#1}{\to}}
\newcommand{\Ltoo}[1]{\stackrel{#1}{\longrightarrow}}

\newcommand{\ggets}{\mbox{:--}}

In this section we show how the (generalized) intersection problems
provide us with upper and lower bounds on the complexity of
evaluating a variety of logical queries over graphs.  We start by
recalling the basic classes of logics used in querying graph data,
and show  that extending them with rational relations allows us to
cast the query evaluation problem as an instance of the generalized
intersection problem. The key observations are that:
\begin{iteMize}{$\bullet$}\itemsep=0pt
\item the complexity of $\GENINT_S(\REC)$ and  $\INTP \REC S$
  provide an upper and a lower bound for the complexity of evaluating
  \crpq($S$) queries; and
\item for \ecrpq($S$), these bounds are provided by the 
complexity of $\GENINT_S(\REG)$ and of $\INTP \REG S$.
\OMIT{
\item the complexity of $\GENINT_S(\REC)$ or  $\GENINT_S(\REG)$
  provides an upper bound on the 
complexity of evaluation \crpq($S$) or \ecrpq($S$) queries; and 
\item the complexity of $\INTP \REC S$ or $\INTP \REG S$ a lower bound
  for evaluating such queries. 
}
\end{iteMize}

\noindent The standard abstraction of graph databases \cite{AG-survey} is  
finite $\Sigma$-labeled graphs $G=\langle V,E\rangle$, where $V$ is a
finite set of nodes, or vertices, and $E\subseteq V \times \Sigma
\times V$ is a set of labeled edges. A {\em path} $\rho$ from $v_0$ to
$v_m$ in $G$ is a sequence of edges $(v_0,a_0,v_1)$, $(v_1,a_1,v_2), \cdots,
(v_{m-1},a_{m-1},v_m)$ from $E$, for some $m \geq 0$.
The {\em label} of
$\rho$, denoted by $\lambda(\rho)$, is the word $a_0 \cdots a_{m-1} \in
\Sigma^*$.

The main building
blocks for graph queries  are {\em regular path
  queries}, or {\em RPQ}s \cite{CMW87}; they are expressions of the form $x \Lto{L} y$,
where $L$ is a regular language. We normally assume that $L$ is
represented by a regular expression or an NFA. Given a $\Sigma$-labeled graph
$G=\langle V,E\rangle$, the answer to an RPQ above is the set of pairs
of nodes $(v,v')$ such that there is a path $\rho$ from $v$ to $v'$
with $\lambda(\rho)\in L$.

{\em Conjunctive RPQs}, or {\em \crpq s} \cite{CGLV00,CGLV00b,CM90} are
the closure of RPQs under conjunction and existential
quantification. Formally, they are expressions of the form 
\begin{equation}
\label{crpq-eq}
\phi(\bar x) \ \ =\ \ \exists \bar y\ \bigwedge_{i=1}^m (u_i
\Ltoo{L_i} u_i')
\end{equation}
where variables $u_i,u_i'$s  come from $\bar x, \bar y$. The
semantics naturally extends the semantics of RPQs: $\phi(\bar a)$ is
true in $G$ iff there is a tuple $\bar b$ of nodes such that for every
$i\leq m$ and every $v_i,v_i'$ interpreting $u_i$ and $u_i'$,
respectively, we have a path $\rho_i$ between $v_i$ and $v_i'$ whose
label $\lambda(\rho_i)$ is in $L_i$. 

\crpq s can further be extended to {\em compare}
paths. For that, we need to name path variables, and choose a class of
allowed relations on paths. The simplest such extension is 
the class of \crpq$(S)$ queries, where 
$S$ is a binary relation over $\Sigma^*$. Its formulae are of the form
\begin{equation}
\label{crpqs-eq}
\phi(\bar x) \ \ =\ \ \exists \bar y\ \Big( \bigwedge_{i=1}^m (u_i
\Ltoo{\chi_i:L_i} u_i') \ \  \wedge \ \  \bigwedge_{(i,j)\in I}
S(\chi_i,\chi_j)\Big) 
\end{equation}
where $I\subseteq [m]^2$. We use variables
$\chi_1,\ldots,\chi_m$ to denote paths; these are quantified
existentially. That is, the semantics of $G\models \phi(\bar a)$ is
that there is a tuple $\bar b$ of nodes and paths $\rho_k$, for $k\leq
m$, between
$v_k$ and $v_k'$ (where, as before,   $v_k,v_k'$ are elements of $\bar
a, \bar b$ interpreting $u_k,u_k'$) 
such that 
$(\lambda(\rho_i),\lambda(\rho_j))\in S$ whenever $(i,j)\in I$. 
For instance, the query $$\exists y, y'\ \big( (x \Ltoo{\chi:\Sigma^*a} y)
\wedge (x\Ltoo{\chi':\Sigma^*b} y') \wedge \chi \subseq \chi'\big)$$
finds nodes $v$ so that there are two paths starting from $v$, one
ending with an $a$-edge, whose label is a subsequence of the other one,
that ends with a $b$-edge.

The input to the {\em query evaluation problem} consists of a graph
$G$, a tuple $\bar v$ of nodes, and a query $\phi(\bar x)$; the
question is whether $G\models\phi(\bar v)$. This corresponds to the
{\em combined complexity} of query evaluation. In the context of query
evaluation, one is often interested in {\em data complexity}, when the
typically small formula $\phi$
is fixed, and the input consists of the typically large graph $(G,\bar
v)$. We now relate it to the complexity of $\GENINT_S(\REC)$.

\begin{lem}
\label{crpqs-lemma-one}
Fix a \crpq($S$) query $\phi$ as in (\ref{crpqs-eq}). Then there is a
\dlog\ algorithm that, given a 
graph $G$ and a tuple $\bar v$ of nodes, constructs an $m$-ary
relation $R \in 
\REC$ so that the answer to the generalized
intersection problem $\INTPI R S I$ is `yes' iff $G\models\phi(\bar v)$. 
\end{lem}

\OMIT{
\begin{proof}[Proof idea]
 Given a $\Sigma$-labeled graph
  $G=\langle V, E\rangle$ and two nodes $v,v'$, we write $\A(G,v,v')$
  for $G$ viewed as an NFA with the initial state $v$ and the final
  state $v'$.  Now consider a \crpq($S$) query $\phi(\bar x)$ given by
  (\ref{crpqs-eq}). Let $\bar v$ be a tuple of nodes of $G$, of the
  same length as $\bar x$.

  The algorithm first enumerates all tuples $\bar b$ of nodes of $G$
  of the same length as $\bar y$.
  % We construct $R_{\bar b}\in\REC_m$ as follows.
  Let $n_i$ and $n_i'$ be the interpretations of $u_i$ and $u_i'$,
  when $\bar x$ is interpreted as $\bar v$ and $\bar y$ as $\bar
  b$. Define $R_{\bar b} \ = \ \prod_{i=1}^m (\LL(\A(G,n_i,n_i')) \cap
  L_i)$; this is a relation in $\REC_m$. Hence, $R=\bigcup_{\bar b}
  R_{\bar b}$ in $\REC_m$ too. It is now easy to see that $R\cap_I
  S\neq \emptyset$ iff $G\models \phi(\bar v)$.  
\end{proof}} 

\begin{proof}
  Given a $\Sigma$-labeled graph $G=\langle V, E\rangle$ and two nodes
$v,v'$, we write $\A(G,v,v')$ for $G$ viewed as an NFA with the
initial state $v$ and the final state $v'$ (that is, the set of states
is $V$, the transition relation is $E$, and the alphabet is
$\Sigma$). The language of such an automaton, $\LL(\A(G,v,v'))$, is
the set of labels of all paths between $v$ and $v'$.

Now consider a \crpq($S$) query $\phi(\bar x)$ given by $$\exists \bar
y\ \Big( \bigwedge_{i=1}^m (u_i \Ltoo{\chi_i:L_i} u_i') \ \ \wedge
\ \ \bigwedge_{(i,j)\in I} S(\chi_i,\chi_j)\Big),$$ as in
(\ref{crpqs-eq}). Suppose we are given a graph $G$ as above and a
tuple of nodes $\bar v$, of the same length as the length of $\bar
x$. The \dlog\ algorithm works as follows.

First we enumerate all tuples $\bar b$ of nodes of $G$ of the same
length as $\bar y$; since $\phi$ is fixed, this can be done in
\dlog. 
For each $\bar b$, we construct an $m$-ary relation $R_{\bar b}$ in $\REC$ as
follows. 
Let $n_i$ and $n_i'$ be the interpretations of $u_i$ and $u_i'$, when
$\bar x$ is interpreted as $\bar v$ and $\bar y$ as $\bar b$. Then
$$R_{\bar b} \ = \ \prod_{i=1}^m (\LL(\A(G,n_i,n_i')) \cap L_i).$$ 
Note that it can be constructed in \dlog; indeed each coordinate of
$R_{\bar b}$ is simply a product of the automaton $\A(G,n_i,n_i')$ and
a fixed automaton defining $L_i$. Next, let $R=\bigcup_{\bar b}
R_{\bar b}$. This is constructed in \dlog\ too. Now it follows
immediately from the construction that $R\cap_I S\neq \emptyset$
iff for some $\bar b$, there exist paths $\rho_i$ between $n_i,n_i'$,
for $i\leq m$, such that $(\lambda(\rho_l),\lambda(\rho_j))\in S$
whenever $(l,j)\in I$, i.e., iff $G\models \phi(\bar v)$. 
\end{proof}
\medskip

Conversely, the intersection problem for recognizable relations
and $S$ can be encoded as answering \crpq($S$) queries. 

\begin{lem}
\label{crpqs-lemma-two}
For any given binary relation $S$, there is a \crpq($S$) query $\phi(x,x')$
and a \dlog\ algorithm that,
given a relation $R\in\REC_2$, constructs a graph $G$ and two nodes
$v,v'$ so that $G\models\phi(v,v')$ iff $R\cap S\neq \emptyset$.
\end{lem}

\begin{proof}
  Let $R$ be in $\REC_2$. It is given as $\bigcup_{i=1}^n (L_i \times
  K_i)$, where the $L_i$s and the $K_i$s are regular languages over
  $\Sigma$. These languages are given by their NFAs which we can view
  as $\Sigma$-labeled graphs. Let $\langle V_i, E_i\rangle$ be the
  underlying graph of the NFA defining $L_i$, such that $v^i_0$ is the
  initial state, and $F_i$ is the set of final states. Likewise, let
  $\langle W_i, H_i\rangle$ be the underlying graph of the NFA
  defining $K_i$, such that $w^i_0$ is the initial state, and $C_i$ is
  the set of final states.

  We now construct the graph $G$. Its labeling alphabet is the union
  of $\Sigma$ and $\{\#, \$, !\}$. Its set of vertices is the disjoint
  union of all the $V_i$s, $W_i$s, as well as two distinguished nodes
  {\em start} and {\em end}. Its edges include all the edges from
  $E_i$s and $H_i$s, and the following:
  \begin{iteMize}{$\bullet$}
  \item $\#$-labeled edges from {\em start} to each initial state,
    i.e., to each $v_i^0$ and $w_i^0$ for all $i \leq n$.
  \item $\$$-labeled edges between the initial states of automata with
    the same index, i.e., edges $(v^i_0,\$,w^i_0)$ for all $i \leq n$.
  \item $!$-labeled edges from final states to {\em end}, i.e., edges
    $(v,!,\text{\em end})$, where $v \in \bigcup_{i\leq n} F_i \cup
    \bigcup_{i \leq n} C_i$.
  \end{iteMize}

  We now define a \crpq($S$) query $\phi(x,y)$ (omitting path
  variables for paths that are not used in comparisons):
$$\exists x_1, x_2, z_1, z_2\ \left(
  \begin{array}{cccc} & x \Lto{\#} x_1 & \wedge & x \Lto{\#} x_2\\
    \wedge & x_1 \Lto{\$} x_2 & & \\
    \wedge & x_1 \Lto{\chi:\sigmas} z_1 & \wedge & x_2
    \Lto{\chi':\sigmas}
    z_2\\
    \wedge & z_1 \Lto{!} y & \wedge &  z_2 \Lto{!} y\\
    \wedge & S(\chi,\chi')
  \end{array}
\right)
$$ 

The query says that from {\em start}, we have $\#$-edges to the
initial states $v^i_0$ and $w^i_0$: they must have the same index
since there is a $\$$-edge between them. From there we have two paths,
$\rho$ and $\rho'$, corresponding to the variables $\chi$ and $\chi'$,
which are $\Sigma$-labeled, and thus are paths in the automata for
$L_i$ and $K_i$, respectively. From the end nodes of those paths we
have $!$-edges to {\em end}, so they must be final states; in
particular, $\lambda(\rho)\in L_i$ and $\lambda(\rho')\in K_i$. We
finally require $(\lambda(\rho),\lambda(\rho'))\in S$, i.e.,
$(\lambda(\rho),\lambda(\rho'))\in (L_i\times K_i)\cap S$. Hence, if
$G\models\phi(\text{\em start},\text{\em end})$ then for some $i\leq
n$ we have two words $(w,w')$ that belong to $(L_i\times K_i)\cap S$,
i.e., $R\cap S\neq \emptyset$. Conversely, if $R\cap S\neq \emptyset$,
then $(L_i\times K_i)\cap S\neq \emptyset$ for some $i\leq n$, and the
witnessing paths of the nonemptiness of $(L_i\times K_i)\cap S$ will
witness the formula $\phi(\text{\em start},\text{\em end})$ (together
with initial states of the automata of $L_i$ and $K_i$ and some of
their final states). 
\end{proof}

\medskip

Combining the lemmas, we obtain:

\begin{thm}\hfill
\label{crpq-thm}
Let $\K$ be a complexity class closed under \dlog\ reductions. Then:
\begin{enumerate}[\em(1)]\itemsep=0pt
\item If the problem  $\GENINT_S(\REC)$ is in $\K$, then data
  complexity of \crpq($S$) queries is in $\K$; and
\item If the problem $\INT \REC S$ is hard for $\K$, then so is data
  complexity of \crpq($S$) queries.
\end{enumerate}
\end{thm} 

\noindent We now consider {\em extended CRPQs}, or {\em ECRPQs}, which enhance
CRPQs with regular relations \cite{pods10}, and prove a similar result
for them, with the role of $\REC$ now played by $\REG$.  Formally,
\ecrpq s are 
expressions of the form 
\begin{equation}
\label{ecrpq-eq}
\phi(\bar x) \ \ =\ \ \exists \bar y\ \Big(\bigwedge_{i=1}^m (u_i
\Ltoo{\chi_i:L_i} u_i') \ \ \wedge\ \  \bigwedge_{j=1}^k R_j(\bar\chi_j)\Big)
\end{equation}
where each $R_j$ is a relation from $\REG$, and $\bar\chi_j$ a tuple
from $\chi_1,\ldots,\chi_m$ of the same arity as $R_j$. The semantics
of course extends the semantics of \crpq s: the witnessing paths
$\rho_1,\ldots,\rho_m$ should also satisfy the condition that for every
atom $R(\rho_{i_1},\ldots,\rho_{i_l})$ in (\ref{ecrpq-eq}), the tuple 
$(\lambda(\rho_{i_1}),\ldots,\lambda(\rho_{i_l}))$ is in $R$. 

Finally, we obtain \ecrpq($S$) queries by adding comparisons with
respect to a relation 
$S\in\RAT$, getting a class of queries $\phi(\bar x)$ of the
form 
\begin{equation}
\label{ecrpqs-eq}
\exists \bar y\ \Big(\!\bigwedge_{i=1}^m (u_i
\Ltoo{\chi_i:L_i} u_i')   \wedge  \bigwedge_{j=1}^k R_j(\bar\chi_j)
 \wedge\! \bigwedge_{(i,j)\in I} S(\chi_i,\chi_j)\Big)
\end{equation}
%Exact analogs of Lemmas \ref{crpqs-lemma-one} and
%\ref{crpqs-lemma-two} hold, with \ecrpq s replacing \crpq s and $\REG$
%replacing $\REC$. 
Similarly to the case of \crpq s, we can establish a
connection between data complexity of \ecrpq($S$) queries and the
complexity of the generalized intersection problem: 
%by means of two
%lemmas.

\begin{thm}
\label{ecrpq-thm}
Let $\K$ be a complexity class closed under \dlog\ reductions. Then:
\begin{enumerate}[\em(1)]\itemsep=0pt
\item If the problem  $\GENINT_S(\REG)$ is in $\K$, then data
  complexity of \ecrpq($S$) queries is in $\K$; and
\item If the problem $\INT \REG S$ is hard for $\K$, then so is data
  complexity of \ecrpq($S$) queries.
\end{enumerate}
\end{thm} 

\noindent Similarly to the proof of Theorem \ref{crpq-thm}, the result
will be an immediate consequence of two lemmas. First, evaluation of
\ecrpq($S$) queries is reducible to the generalized intersection
problem for regular relations.

\begin{lem}
\label{ecrpqs-lemma-one}
Fix an \ecrpq($S$) query $\phi$ as in (\ref{ecrpqs-eq}). Then there is a
\dlog\ algorithm that, given a 
graph $G$ and a tuple $\bar v$ of nodes, constructs an $m$-ary relation $R \in
\REG$ so that the answer to the generalized
intersection problem $\INTPI R S I$ is `yes' iff $G\models\phi(\bar
v)$. 
\end{lem}

Conversely, the intersection problem for regular relations
and $S$ can be encoded as answering \ecrpq($S$) queries. 

\begin{lem}
\label{ecrpqs-lemma-two}
For each binary relation $S$, there is an \ecrpq($S$) query $\phi(x,x')$
and a \dlog\ algorithm that,
given a relation $R\in\REG_2$, constructs a graph $G$ and two nodes
$v,v'$ so that $G\models\phi(v,v')$ iff $(R\cap S)\neq \emptyset$.
\end{lem}

The proof of Lemma \ref{ecrpqs-lemma-one} is almost the same as the
proof of  Lemma \ref{crpqs-lemma-one}: as before, we enumerate tuples
$\bar b$, construct relations $R_{\bar b}$ and $R=\bigcup_{\bar b}
R_{\bar b}$, but this time we take the product of this recognizable
relation with regular relations mentioned in the query. Since the
query is fixed, and hence we take a product with a fixed number of
fixed automata, such a product construction can be done in \dlog. The
result is now a regular $m$-ary relation. The rest of the proof is
exactly the same as in Lemma \ref{crpqs-lemma-one}. 

We now prove Lemma \ref{ecrpqs-lemma-two}. Let $R\in\REG_2$ be given
by an NFA over $\sigbot\times\sigbot$ whose underlying graph is
$G_R=\langle V_R, E_R\rangle$, where $E_R\subseteq V_R \times
(\sigbot\times\sigbot) \times V_R$. Let $v_0$ be its initial state,
and let $F$ be the set of final states. 

We now define the graph $G$. Its labeling alphabet $\Gamma$ is the
disjoint union of $\sigbot\times\sigbot$, the alphabet $\Sigma$
itself, and a new symbol $\#$. Its nodes $V$ include all nodes in
$V_R$ and two extra nodes, $v_f$ and $v'$. The edges are:
\begin{iteMize}{$\bullet$}
\item all the edges in $E_R$;
\item edges $(v,\#,v_f)$ for every $v\in F$;
\item edges $(v',a,v')$ for every $a\in\Sigma$. 
\end{iteMize}
 We now define two regular relations over $\Gamma$. The first, $R_1$,
 consists of pairs $(w, w')$, where $w \in
 (\sigbot\times\sigbot)^*$ and $w'\in \sigmas$. Furthermore, $w$ is of
 the form $w' \otimes w''$ for some $w''\in\sigmas$. It is straightforward to
 check that this relation is regular. The second one, $R_2$, is the same
 except $w$ is of the form $w''\otimes w'$.  In other words, the first
 component is $w_1\otimes w_2$, and the second is either $w_1$
 or $w_2$, for $R_1$ or $R_2$, respectively. 

Next, we define the \ecrpq($S$)  $\phi(x,y)$:
$$\exists x_1, y_1, x_2, y_2, z\ 
\left(
 \begin{array}{cccccc}
& x \Lto{\chi: \sigbot\times\sigbot} z & \wedge & z \Lto{\#} y && \\
\wedge & x_1 \Lto{\chi_1:\sigmas} y_1 & \wedge & x_2
\Lto{\chi_2:\sigmas} y_2 &&\\
\wedge & R_1(\chi,\chi_1) & \wedge & R_2(\chi,\chi_2) & \wedge &
S(\chi_1,\chi_2) 
\end{array}
\right)
$$ Note that when this formula is evaluated over $G$, with $x$
interpreted as $v_0$ and $y$ interpreted as $v_f$, the paths $\chi_1$
and $\chi_2$ can have arbitrary labels from $\Sigma^*$. Paths $\chi$
can have arbitrary labels over $\sigbot\times\sigbot$; however, since
they start in $v_0$ and 
must be followed by an $\#$-edge, they end in a final state of
the automaton for $R$, and hence labels of these paths are precisely
words in $\sigbot\times\sigbot$ of the form
$w_1\otimes w_2$, where $(w_1,w_2)\in R$. Now $R_1$ ensures that the
label of $\chi_1$ is $w_1$ and that the label of $\chi_2$ is $w_2$.
Hence the labels of $\chi_1$ and $\chi_2$ are precisely the pairs of
words in $R$, and the query asks whether such a pair belongs to
$S$. Hence, $G\models \phi(v_0,v_f)$ iff $R\cap S\neq \emptyset$. It
is straightforward to check that the construction of $G$ can be
carried out in \dlog. This proves the lemma and the theorem.

\bigskip

Thus, our next goal is to understand the behaviors of the generalized
intersection problem for various rational relations $S$ which are of
interest in graph logics; those include subword, suffix,
subsequence. In fact to rule out many undecidable or infeasible cases
it is often sufficient to analyze the intersection problem. We do this
in the next section, and then analyze the decidable cases to come up
with graph logics that can be extended with rational relations.

%%% Local Variables: 
%%% mode: latex
%%% TeX-master: "lrr"
%%% End: 

\section{The intersection problem: decidable and undecidable cases} 
\label{intp:sec}
\newcommand{\rpep}{\textup{PEP}^{\textit{reg}}}

We now study the problem $\INTP \REG S$ for binary rational relations
$S$ such as subword and subsequence, and for classes of relations
generalizing them. The input is a binary regular relation $R$ over
$\Sigma$, given by an \NFA over $\sigbot\times\sigbot$.  The question
is whether $R\cap S\neq\emptyset$.  We also derive results about the 
complexity of \ecrpq($S$) queries. 
For all
lower-bound results in this section, we assume that the alphabet
contains at least two symbols.

As already mentioned, there exist rational relations $S$ such that
$\INTP \REG S$ is undecidable. However, we are interested in relations
that are useful in graph querying, and that are among the most
commonly used rational relations, and for them the status of the
problem was unknown. 

Note that the problem $\INTP\REC S$ is tractable:
given $R\in\REC$, the relation $R\cap S$ is rational, can be
efficiently constructed, and checked for nonemptiness.

\subsection{Undecidable cases: subword and relatives}
\label{subword-subsec}

We now show that even for such simple relations as subword
and suffix, the intersection problem is undecidable. That is, given an
NFA over $\sigbot\times\sigbot$ defining a regular relation $R$, the
problem of checking for the existence of a pair $(w,w')\in R$ with
$w \suff w'$ or $w\subw w'$ is undecidable.

\begin{thm}
\label{suffix-thm}
The problems \mbox{$\INTP \REG \suff$} and  \mbox{$\INTP \REG \subw$}
are undecidable.
\end{thm}

As an immediate consequence of this, we obtain:

\begin{cor}
The query evaluation problem for \ecrpq($\suff$) and \ecrpq($\subw$)
is undecidable.
\end{cor}

Thus, some of the most commonly used rational relations cannot be
added to \ecrpq s without imposing further restrictions. 

%\medskip
\OMIT{
\noindent
{\em Proof idea}. 
We sketch the idea of the proof for $\suff$. We encode nonemptiness
for linearly bounded automata (LBA). The alphabet $\Sigma$ is the
disjoint union of the tape alphabet of the LBA, its states, and the 
designated symbol $\$$. 
Each configuration $C$ with the tape content $a_0\ldots a_n$, where $a_0$
and $a_n$ are the left and 
right markers, the state is $q$, and the head points at $a_i$ is
encoded as a word $w_C=\$a_0\ldots a_{i-1}qa_i\ldots a_n\$$. Note that
the relation $\{(w_C,w_{C'})\ | \ C' \text{ is an immediate successor of
}C\}$ is regular and hence so is the relation $R=\{(w_{C_0}w_{C_1}\ldots
w_{C_m},w_{C_1'}\ldots w_{C_m'}) \ | \ C_{i+1}'$ is an immediate successor of
$C_i\text{ for } i<m\}$, since all configuration encodings are of the
same length. In fact, taking product with a regular language, we can
also assume that $R$ 
enforces $C_0$ to be an initial configuration, and $C_m$ to be a
final configuration. If $R\cap\mbox{$\suff$}$ is nonempty, it contains
a pair $ (w_{C_0}w_{C_1}\ldots w_{C_m},w_{C_1}\ldots
w_{C_m})$ such that $C_{i+1}$ is an immediate successor of 
$C_i$  for  all $i<m$, i.e., iff there is an accepting computation of the
LBA. This proves undecidability. The proof for $\subw$ is very
similar. \fth

\medskip

Note that the relation $R$ constructed in the proof is definable in
first-order logic, so the intersection problem for suffix and subword
is undecidable even if the input relation comes from the class of
star-free regular relations. 

\smallskip} 

\smallskip 

We skip the proof of Theorem \ref{suffix-thm} for the time being and
concentrate first on how to obtain a more general undecidability
result out of it.  As we will see below, the essence of the
undecidability result is that relations such as $\suff$ and $\subw$
can be decomposed in a way that one of the components of the
decomposition is a graph of a nontrivial strictly alphabetic
morphism. More precisely, let $R\cdot R'$ be the binary relation
$\{(w\cdot w', u\cdot u') \ | \ (w,u)\in R \text{ and }(w',u')\in
R'\}$. Let $\text{Graph}(f)$ be the graph of a function
$f:\Sigma^*\to\Sigma^*$, i.e., $\{(w,f(w))\ |\ w\in\sigmas\}$.

\begin{prop}
\label{gen-subw-cor}
Let $R_0,R_1$ be binary relations on $\Sigma$ such that $R_0$ is
recognizable and its second
projection is $\Sigma^*$. Let $f$ be a strictly alphabetic
morphism that is not constant (i.e. 
the image of $f$ contains at least two letters). Then, for $S=R_0\cdot\text{\rm Graph}(f)\cdot
R_1$, the problem $\INTP \REG S$ is undecidable. 
\end{prop}

Note that 
both $\suff$ and $\subw$ are of the required shape:
suffix is $(\{\e\}\times\Sigma^*) \cdot \text{Graph(id)} \cdot
(\{\e\}\times\{\e\})$, and 
subword is $(\{\e\}\times\Sigma^*) \cdot \text{Graph(id)} \cdot
(\{\e\}\times\Sigma^*)$, where id is the identity alphabetic morphism.

\begin{proof}[Proofs of Theorem \ref{suffix-thm} and Proposition
  \ref{gen-subw-cor}]
  We present the proof for the suffix relation $\suff$. The proofs for
  the subword relation, and more generally, for the relations
  containing the graph of an alphabetic morphism follow the same idea
  and will be explained after the proof for $\suff$. The proof is by
  encoding nonemptiness for linearly bounded automata (LBA). Recall
  that an LBA $\A$ has a tape alphabet $\Gamma$ that contains two
  distinguished symbols, $\alpha$ and $\beta$, which are the left and
  the right marker. The input word $w\in (\Gamma-\{\alpha,\beta\})^*$
  is written between them, i.e., the content of the input tape is
  $\alpha \cdot w\cdot \beta$. The LBA behaves just like a Turing
  machine, except that when it is reading $\alpha$ or $\beta$, it
  cannot rewrite them, and it cannot move left of $\alpha$ or right of
  $\beta$. The problem of checking whether the language of a given LBA
  is nonempty is undecidable.

  We encode this as follows.  The alphabet $\Sigma$ is the disjoint
  union of the tape alphabet $\Gamma$ of the LBA $\A$, the set of its
  states $Q$, and the designated symbol $\$$ (we assume, of course,
  that these are disjoint). A configuration $C$ of the LBA consists of
  the tape content $a_0\ldots a_n$, where $a_0=\alpha$ and
  $a_n=\beta$, and all the $a_i$s, for $0< i < n$, are letters from
  $\Gamma-\{\alpha,\beta\}$, the state $q$, and the position $i$, for
  $0 \leq i \leq n$, that the head is pointing to. We encode this as a
  word
$$w_C\ =\ \$a_0\ldots a_{i-1}qa_i\ldots a_n\$ \in \Sigma^*$$
of length $n+4$. Of course if the head is pointing to $\alpha$, the
configuration is $\$qa_0\ldots a_n\$$. Note that if we have a run of
the LBA with configurations $C_0, C_1, \ldots$, then the lengths of
all the $w_{C_i}$s are the same.

Next, note that the relation $$R^\A_{{\rm imm}}\ = \ \{(w_C,w_{C'})\ |
\ C' \text{ is an immediate successor of }C\}$$ is regular (in fact
such a relation is well-known to be regular even for arbitrary Turing
machines \cite{jacm2003,graedel2000,bruyere}). Since all
configurations are of the same length, we obtain that the relation
$$R_{\A}' \ =\ \{(w_{C_0}w_{C_1}\ldots
w_{C_m},w_{C_1'}\ldots w_{C_m'}) \ | \ C_{i+1}'\ \text{is an immediate
  successor of }C_i\text{ for } i<m\}$$ is regular too (since only one
configuration in the first projection does not correspond to a
configuration in the second projection). By taking the product with a
regular language that ensures that the first symbol from $Q$ in a word
is $q_0$, and the last such symbol is from $F$, we have a regular
relation
$$R_{\A} \ =\ \biggl\{(w_{C_0}w_{C_1}\ldots
w_{C_m},w_{C_1'}\ldots w_{C_m'}) \ \biggl|
\begin{array}{l} C_{i+1}'\ \text{is an immediate successor of
  }C_i\text{ for } i<m;\\
  C_0 \text{ is an initial configuration };\\
  C_m \text{ is a final configuration }
\end{array}
\biggr\}$$ which can be effectively constructed from the description
of the LBA.

Now assume that $R_{\A}\cap\mbox{$\suff$}$ is nonempty. Then, since
all encodings of configurations are of the same length, it must
contain a pair $ (w_{C_0}w_{C_1}\ldots w_{C_m},w_{C_1}\ldots w_{C_m})$
such that $C_{i+1}$ is an immediate successor of $C_i$ for all
$i<m$. Since $C_0$ is an initial configuration and $C_m$ is a final
configuration, this implies that the LBA has an accepting
computation. Conversely, if there is an accepting computation with a
sequence of configurations $C_0,C_1,\ldots,C_m$ of the LBA, then the
pair $ (w_{C_0}w_{C_1}\ldots w_{C_m},w_{C_1}\ldots w_{C_m})$ is both
in $R_{\A}$ and in the suffix relation.  Hence,
$R_{\A}\cap\mbox{$\suff$}$ is nonempty iff there is an accepting
computation of the LBA, proving undecidability.

The proof for the subword relation is practically the same. We change
the definition of relation $R_{\A}$ so that there is an extra \$
symbol inserted between $w_{C_0}$ and $w_{C_1}$, and two extra \$
symbols after $w_{C_m}$ in the first projection; in the second
projection we insert extra two \$ symbols before $w_{C_1'}$ and after
$w_{C_m'}$. Note that the relation remains regular: even if the
components are not fully synchronized, at every point there is a
constant delay between them (either 2 or 1), and this can be captured
by simply encoding one or two alphabet symbols into the state. Since
in each word there are precisely two places where the subword \$\$\$
appears, the subword relation in this case becomes the suffix
relation, and the previous proof applies.

The same proof can be applied to deduce Proposition
\ref{gen-subw-cor}. Note that we can encode letters of alphabet
$\Sigma$ within the alphabet $\{0,1\}$ so that the encodings of each
letter of $\Sigma$ will have the same length, namely $\lceil \log_2
(|\Gamma|+|Q|+1)\rceil$.  Then the same proof as before will apply to
show undecidability over the alphabet $\{0,1\}$, since the encodings
of configurations still have the same length.

Since $R_0$ is regular, it is of the form $\bigcup_i L_i \times K_i$,
and by the assumption, $\bigcup_i K_i=\sigmas$. Thus, the encoding of
the initial configuration will belong to one of the $K_i$s, say
$K_j$. We then take a fixed word $w_0\in L_j$ and assume that the
second component of the relation starts with $w_0$ (which can be
enforced by the regular relation). Likewise, we take a fixed pair
$(w_1,w_2)\in R_1$, and assume that $w_1$ is the suffix of the first
component of the relation, and $w_2$ is the suffix of the second. This
too can be enforced by the regular relation.

Now if we have a non-constant alphabetic morphism $f$, we have two
letters, say $a$ and $b$, so that $f(a)\neq f(b)$. We now simply use
these letters, with $a$ playing the role of $0$, and $b$ playing the
role of $1$ in the first projection of relation $R$, and $f(a), f(b)$
playing the roles of $0$ and $1$ in the second projection, to encode
the run of an LBA as we did before. The only difference is that
instead of a sequence of \$ symbols to specify the positions of the
encoding we use a (fixed-length) sequence that is different from
$w_0,w_1,w_2$ above, to identify its position uniquely. Then the proof
we have presented above applies verbatim.
\end{proof}

\subsection{Decidable cases: subsequence and relatives}
\label{subseq-subsec}

We now show that the intersection problem is decidable for the
subsequence relation $\subseq$ and, much more generally, for a class
of relations that do not, like the relations considered in the
previous section, have a ``rigid'' part. More precisely, the problem is also decidable for any relation so that its projection on the first component is closed under subsequence. However, the complexity bounds are extremely
high. In fact we show that the complexity of checking whether
$(R \cap \mbox{$\subseq$}) \neq \emptyset$, when $R$ ranges over
$\REG_2$, is not bounded by any multiply-recursive function. This was
previously known for $R$ ranging over $\RAT_2$, and was viewed as the
simplest problem with non-multiply-recursive
complexity \cite{CS-fsttcs07}. We now push 
it further and show that this high complexity is already achieved with
regular relations. 

Some of the ideas for showing this come from a decidable relaxation of
the Post Correspondence Problem (PCP), namely
the \emph{regular Post Embedding Problem}, or $\rpep$, introduced
in \cite{CS-fsttcs07}.
An instance of this problem consists of two morphisms $\sigma,\sigma': \Sigma^* \to \Gamma^*$ and a regular language $L \subseteq \Sigma^*$; the question is whether there is some $w \in L$
such that 
$\sigma(w) \subseq \sigma'(w)$ (recall that in the case of the PCP the
question is whether $\sigma(w) = \sigma'(w)$ with $L = \Sigma^+$). 
We call $w$ a \emph{solution} to the instance $(\sigma,\sigma',L)$. 
The $\rpep$ problem is known to be decidable, and as hard as the reachability
problem for lossy channel systems \cite{CS-fsttcs07} which
cannot be bounded by any primitive-recursive function ---in fact, by any
multiply-recursive function (a generalization of primitive recursive functions with hyper-Ackermannian complexity, see \cite{rose}). More precisely, it is shown in \cite{schsch} to be precisely at the level $\textup F_{\omega^\omega}$ of the fast-growing hierarchy of recursive functions \cite{fast,rose}.\footnote{In this hierarchy---also known as the Extended Grzegorczyk Hierarchy---, the classes of functions $\textup F_\alpha$ are closed under elementary-recursive reductions, and are indexed by ordinals. Ackermannian complexity corresponds to level $\alpha = \omega$, and level $\alpha = \omega^\omega$ corresponds to some hyper-Ackermannian complexity. }
 
The problem $\rpep$ is just a reformulation of the problem $\INTP \RAT
{\mbox{$\subseq$}}$. Indeed, relations of the form 
$\{(f(w), g(w))\ | \ w\in L\}$, where $L\subseteq \sigmas$ ranges over
regular 
languages and $f,g$ over morphisms
$\Sigma^*\to\Gamma^*$ are precisely the relations in $\RAT_2$ \cite{berstel,nivat68}. 
Hence, $\INTP \RAT {\mbox{$\subseq$}}$ is decidable, with
non-multiply-recursive complexity.  
\begin{prop}[\cite{CS-fsttcs07}]\label{prop:subseq-rat-dec-nmr}
$\INTP \RAT {\mbox{$\subseq$}}$ is decidable, non-multiply-recursive.
\end{prop}

We show that the lower bound
already applies to regular relations. 

\begin{thm}
\label{subseq-dec}
The problem $\INTP \REG
{\mbox{$\subseq$}}$ is decidable, and its complexity is not bounded by
any multiply-recursive function.
\end{thm}

The proof of the theorem above will be shown further down, after some preparatory definitions and lemmas are introduced.

It is worth noticing that one cannot solve the problem $\INTP \REG
{\mbox{$\subseq$}}$ by simply reducing to nonemptiness of rational
relations due to the following.

\begin{prop}\label{prop:reg-subseq-nonrat} There is a binary
regular relation $R$ such that $(R \cap {\subseq})$ is not rational.
\end{prop} \begin{proof} Let $\Sigma = \set{ a, b}$, and
consider the following regular relation, \[ R = \set{(a^m,b^m \cdot
a^{m'}) \mid m, m' \in \N} .  \]%
\newcommand{\Aut}{\+A}%
\newcommand{\run}{\rho_\Aut}%
Note that the relation $R \cap {\subseq}$
is then $\set{(a^m,b^m \cdot a^{m'} ) \mid m, m' \in \N, m' \geq
m}$. We show that $R \cap {\subseq}$ is not rational by means of
contradiction.  Suppose that it is, and let $\Aut$ be an NFA over
$\set{a,b,\e} \times \set{a,b,\e}$ that recognizes $R \cap
{\subseq}$. Suppose $Q$ is the set of states of $\Aut$, and
$|Q|=n$. 

Consider the following pair \[ (a^{n+1}, b^{n+1} \cdot
a^{n+1}) \quad\in\quad R \cap {\subseq}. \] Then there must be some
$u \in (\set{a,b,\e} \times \set{a,b,\e})^*$ such that
\[(\wproj_1(u),\wproj_2(u)) = (a^{n+1}, b^{n+1} \cdot
a^{n+1})\] and $u \in \lang\Aut$.  Let $\run : [0..|u|] \to Q$ be the
accepting run of $\Aut$ on $u$, and let $1 \leq i_1< \dotsb < i_{n+1}
\leq |u|$ be such that $\wproj_2(u[i_j]) = a$ for all $j \in
[n+1]$. Clearly, among $\run(i_1), \dotsc, \run(i_{n+1})$ there must
be two repeating elements by the pigeonhole principle. Let $1\leq
j_1<j_2\leq n+1$ be such elements, where
$\run(i_{j_1})=\run(i_{j_2})$. Hence $u' = u[1..i_{j_1}-1] \cdot
u[i_{j_2}..] \in \lang \Aut$, and therefore \[
\big(\wproj_1(u'), \wproj_2(u')\big) \quad\in\quad R
\cap {\subseq}.  \] Notice that $\wproj_2(u') = b^{n+1}
\cdot a^{n+1 - (j_2 - j_1)}$. But by definition of $R \cap {\subseq}$ we have that $\wproj_1(u') = a^{n+1}$ with $n+1 - (j_2 - j_1) \geq n+1$, which is clearly false. The contradiction comes from the assumption that $R \cap
{\subseq}$ is rational.  
\end{proof}

\smallskip

As already mentioned, the decidability part of Theorem~\ref{subseq-dec} follows from Proposition~\ref{prop:subseq-rat-dec-nmr}.  We prove the lower bound by reducing $\rpep$ into $\INTP \REG {\mbox{$\subseq$}}$.

This reduction is done in two phases. 
First, we show that there is a reduction from $\rpep$ into the problem of finding solutions of $\rpep$ with a certain shape, which we call a \emph{strict codirect solutions} (Lemma~\ref{lemma:instrumental}). Second, we show that there is a reduction from the problem of finding strict codirect solutions of a $\rpep$ instance into $\INTP \REG {\mbox{$\subseq$}}$ (Proposition~\ref{prop:reducRPEP-inter-REG-subseq}).
Both reductions are elementary and thus the hardness result of Theorem~\ref{subseq-dec} follows.

In the next section we define the strict codirect solutions for $\rpep$, showing that we can restrict to this kind of solutions. In the succeeding section we show how to reduce the problem into $\INTP \REG {\mbox{$\subseq$}}$.

%To prove Theorem~\ref{subseq-dec}, we first have to show a lemma for the Post Embedding Pro%blem.

\subsubsection{Codirect solutions of $\rpep$}
%\newcommand{\rpep}{\textup{PEP}^{\textit{reg}}}
% The \emph{regular Post Embedding Problem} is introduced in \cite{CS-fsttcs07}, and it is a problem which is in the level $\frak F_{\omega^\omega}$ of the Fast-Growing Hierarchy of recursive functions \cite{fast}. This problem, henceforth denoted by $\rpep$, is the problem of, given two morphisms $\sigma,\sigma' : \Sigma^* \to \Gamma^*$ and a regular language $L \subseteq \Sigma^*$, whether there is some $w \in L$ such that $\sigma(w) \subseq \sigma'(w)$. In this case we say that $w$ is a \emph{solution} to the problem instance $(\sigma,\sigma',L)$. This problem is known to be decidable, as hard as the reachability problem for lossy channel systems, which cannot be bounded by any multiple-recursive function.
\newcommand{\interemptypb}[2]{\INTP{#1}{#2}}
\newcommand{\interemptypbdir}[2]{{({#1} \cap {#2})^{\text{dir}}} \stackrel{?}{=} \emptyset}
\newcommand{\interemptypbsyn}[2]{{({#1} \cap {#2})^{\text{syn}}} \stackrel{?}{=} \emptyset}
\newcommand{\interemptypbcodir}[2]{{({#1} \cap {#2})^{\text{codir}}} \stackrel{?}{=} \emptyset}
\newcommand{\interemptypbgen}[1]{{({#1})} \stackrel{?}{=} \emptyset}

There are some variations of the $\rpep$ problem that result being equivalent problems. These variations restrict the solutions to have certain properties.
Given a $\rpep$ instance $(\sigma,\sigma',L)$, we say that $w \in L$ with $|w|=m$ is a \emph{codirect solution} if there are (possibly empty) words $v_1, \dotsc, v_m$  such that
  \begin{enumerate}[1.]
    \item \label{item:codir:1}
    $v_k \subseq \sigma'(w[k])$ for all $1 \leq k \leq m$,
    \item \label{item:codir:2}
    $\sigma(w[1..m]) = v_1 \dotsb v_m$, and
    \item \label{item:codir:3}
    $|\sigma(w[1..k])| \geq |v_1 \dotsb v_k|$ for all $1 \leq k \leq m$.
    \newcounter{enumi_saved}
    \setcounter{enumi_saved}{\value{enumi}}
  \end{enumerate}
If furthermore
  \begin{enumerate}[1.]
      \setcounter{enumi}{\value{enumi_saved}}
    \item \label{item:codir:4}
    $|\sigma(w[1..k])| > |v_1 \dotsb v_k|$ for all $1 \leq k < m$,
  \end{enumerate}
we say that it is a \emph{strict codirect solution}. In this case we say that the solution $w$ is \emph{witnessed by} $v_1, \dotsc, v_m$.
In \cite{CS-fsttcs07} it has been shown that the problem of whether an instance of the $\rpep$ problem has a codirect solution is equivalent to the problem of whether it has a solution. Moreover, it can be shown that this also holds for strict codirect solutions.
\begin{lem} \label{lemma:instrumental} 
  The problem of whether a $\rpep$ instance has a strict codirect solution is as hard as whether a $\rpep$ instance has a solution.
\end{lem}
\begin{proof}
We only show how to reduce from finding a codirect solution problem to finding a strict codirect solution problem. The other direction is trivial, since a strict codirect solution is in particular a  solution.
  Let $(\sigma, \sigma', L)$ be a $\rpep$ instance, and $w \in L$  be a  codirect solution with $|w|=m$, minimal in size, and witnessed by $v_1, \dotsc, v_k$. Let $\+A=(Q,\Sigma,q_0,\delta,F)$ be an \NFA representing $L$, where $|Q|=n$. Let $\rho : [0..m] \to Q$ be an accepting run of $\+A$ on $w$. Let $0 \leq k_1 < \dotsb < k_{t} \leq m$ be all the elements of $\set{ s \geq 0 : |\sigma(w[1..s])| = |v_1 \dotsb v_{s}|}$. Observe that $k_1=0$, and $k_t = m$ by condition~\ref{item:codir:2}. It is not difficult to show that by minimality of $m$ there cannot be more than $n$ indices.
\begin{claim}\label{claim:bound-codir-sol}
   $t \leq n$.
\end{claim}
\begin{proof}
  Suppose ad absurdum that $t \geq n+1$. Then, there must be two $k_l < k_{l'}$ such that   $\rho(k_l)=\rho(k_{l'})$. Hence,  $w' = w[1..k_l] \cdot w[k_{l'}+1..]  \in L$ is also a codirect solution, contradicting  that $w$ is a minimal size solution.
\end{proof}

Let $L[q,q']$ be  the regular language denoted by the \NFA $(Q,\Sigma,q,\delta,\set{q'})$.

\begin{claim}
For every $i<t$,  $(\sigma,
    \sigma', L[\rho(k_i),\rho(k_{i+1})])$ has a strict codirect
    solution.
\end{claim}
\begin{proof}
We show that for every $i<t$, $w[k_i+1 .. k_{i+1}]$ is a solution for $(\sigma,  \sigma', L[\rho(k_i),\rho(k_{i+1})])$, witnessed by $v_{k_i+1}, \dotsc, v_{k_{i+1}}$. 

Clearly, condition \ref{item:codir:1} still holds. 
Further, since
\[
  |\sigma(w[1..k_i])| = |v_1 \dotsb v_{k_i}| 
\qquad\text{and}\qquad
  |\sigma(w[1..k_{i+1}])| = |v_1 \dotsb v_{k_{i+1}}|,
\]
we have that $|\sigma(w[k_i+1..k_{i+1}])| = |v_{k_i+1} \dotsb v_{k_{i+1}}|$ and then 
\[\sigma(w[k_i+1..k_{i+1}]) = v_{k_i+1} \dotsb v_{k_{i+1}},\] 
verifying condition \ref{item:codir:2}. 

Finally, by the fact that $k_i$ and $k_{i+1}$ are consecutive indices we cannot have some $k'$ with $k_i+1 < k' < k_{i+1}$ so that $|\sigma(w[k_i +1..k'])| = |v_{k_i+1} \dotsb v_{k'}|$ since it would imply  $|\sigma(w[1..k'])| = |v_1 \dotsb v_{k'}|$ and in this case $k' \geq k_{i+1}$. Then, conditions \ref{item:codir:3} and \ref{item:codir:4} hold.
\end{proof}

Therefore, we obtain the following reduction.
\begin{claim}
  $(\sigma,  \sigma', L)$ has a codirect solution if, and only if, there exist $\set{q_1, \dotsc, q_t} \subseteq Q$ with $q_1 = q_0$ and $q_t \in F$, such that for every $i$, $(\sigma,\sigma',L[q_i,q_{i+1}])$ has a strict codirect solution.
\end{claim}
This reduction being exponential is outweighed by the fact that we are dealing with a much harder problem.
\end{proof}

With the help of Lemma \ref{lemma:instrumental} 
we prove Theorem \ref{subseq-dec} in the next section.  

\subsubsection{Proof of Theorem \ref{subseq-dec}} 
Since decidability follows from Proposition~\ref{prop:subseq-rat-dec-nmr}, 
we only show the lower bound.
To this end, we show how to
code the existence of a strict codirect solution as an instance of
$\INTP \REG {\mbox{$\subseq$}}$.

\begin{prop}\label{prop:reducRPEP-inter-REG-subseq}
There is an elementary reduction from the existence of strict codirect solutions of $\rpep$ into $\INTP \REG {\mbox{$\subseq$}}$.
\end{prop}

Given a $\rpep$ instance $(\sigma,\sigma',L)$, remember that the presence of a strict codirect solution  enforces that if there is a pair $(u,v)=(\sigma(w),\sigma'(w))$ with $w \in L$ and $u \subseq v$, it is such that for every proper prefix $u'$ of $u$ the smallest prefix $v'$ of $v$ such that $u' \subseq v'$ must be so that $|v'| > |u'|$. 
In the proof, we convert the \emph{rational} relation $R = \set{(\sigma(w),\sigma'(w)) \mid w \in L}$ into a length-preserving \emph{regular} relation $R'$ over an extended alphabet $\Gamma\cup \set \#$, defined as the set of all pairs $(u,v) \in (\Gamma\cup \set \#)^* \times (\Gamma\cup \set \#)^*$ so that $|u|=|v|$ and $(u_\Gamma, v_\Gamma) \in R$. If we now let $R''$ to be the regular relation
$R' \cdot \set{(\e,v) \mid v \in \set{\#}^*}$, we obtain that:
\begin{enumerate}[(i)]
\item if $w \in {R'' \cap {\subseq}}$ then $w' \in {R \cap {\subseq}}$,
where $w'$ is the projection of $w$ onto $\Gamma^* \times \Gamma^*$;
and 
\item if there is some strict codirect solution $w' \in R \cap {\subseq}$, then there is some $w \in R'' \cap {\subseq}$ such that $w'$ is the
projection of $w$ onto $\Gamma^* \times \Gamma^*$. 
\end{enumerate}
Whereas (i) is trivial, (ii) follows from the fact that $w'$ is a
strict codirect solution. If $w' = (u,v) \in R''$, where $f(w)
=(u)_\Gamma$, $g(w)= (v)_\Gamma$, the complication is now that, since
$u \in \Gamma\cup \set \#$, it could be that $u \not\subseq v$ just because
there is some $\#$ in $u$ that does not appear in $v$. But we show how to build $(u,v)$ such that whenever $u[i] = \#$ forces $v[j] = \#$ with $j
> i$ then we also have that $u[j] = \#$. This repeats, forcing $v[k]
= \#$ for some $k>j$ and so on, until we reach the tail of $v$ that
has sufficiently many $\#$'s to satisfy all the accumulated demands
for occurrences of $\#$.

\begin{proof}[Proof of Proposition~\ref{prop:reducRPEP-inter-REG-subseq}]
\newcommand{\rpepc}{\rpep_\textit{\!codir}}
Let $(\sigma,\sigma',L)$ be a $\rpep$ instance. For every $a \in \Sigma$, consider the binary relation $R_a$ consisting of all pairs $(u,u') \in (\Gamma\cup\set\#)^*\times (\Gamma\cup\set\#)^*$ such that $u_\Gamma = \sigma(a)$, $u'_\Gamma = \sigma'(a)$ and $|u|=|u'|$. Note that $R_a$ is a length-preserving regular relation. Let $R'$ be the set of pairs $(u_1 \dotsb u_m, u'_1 \dotsb u'_m)$ such that there exists $w \in L$ where $|w|=m$ and $(u_i,u'_i) \in  R_{w[i]}$ for all $i$. Note that $R'$ is still a length-preserving regular relation. Finally, we define $R$ as the set of pairs $(u,u' \cdot u'')$ such that $(u,u') \in R'$ and $u'' \in \set{\#}^*$. $R$ is no longer a length-preserving relation, but it is regular. Observe that if $R \cap {\subseq} \neq \emptyset$, then $(\sigma,\sigma',L)$ has a solution. Conversely, we show that if $(\sigma,\sigma',L)$ has a strict codirect solution, then $R \cap {\subseq} \neq \emptyset$.

Suppose that the $\rpep$ instance $(\sigma, \sigma', L)$ has a strict codirect solution $w \in L$ with $|w|=m$, witnessed by $v_1, \dotsc, v_m$. Assume, without any loss of generality, that $\sigma$ and $\sigma'$ are alphabetic morphisms and that $m>1$. We exhibit a pair $(u,u') \in R$ such that $u \subseq u'$.
We define $(u,u') = (u_1 \dotsb u_m, u'_1 \dotsb u'_m \cdot u'_{m+1})$, where $(u_i,u'_i) \in R_{w[i]}$ for every $i \leq m$, and $u'_{m+1} \in \set{\#}^*$. In order to give the precise definition of $(u,u')$, we need to introduce some concepts first. 

Let $\sigma_\#(a) \in \Gamma\cup\set\#$ be $\#$ if $\sigma(a) = \epsilon$, or $\sigma(a)$ otherwise; likewise for $\sigma'_\#$. 
By definition of strict codirect solution, we have the following.
\begin{claim}\label{claim:fst-elem-B}
  $\sigma(w[1]) \in \Gamma$.
\end{claim}
\begin{proof}
Indeed, if $\sigma(w[1]) \neq \Gamma$, then $\sigma(w[1])=\e$ and $|\sigma(w[1])|=0$, and then condition \ref{item:codir:4} of strict codirectness stating that $|\sigma(w[1])| > |v_1|$, would be falsified.
\end{proof}

Let us define the function $g: [m] \to [m]$ so that $g(i)$ is the minimum $j$ such that $v_1 \dotsb v_j = \sigma(w[1..i])$.  Note that there is always such a $j$, since $|\sigma(w[1..i])| > 0$ by Claim~\ref{claim:fst-elem-B}. Now we show some easy properties of $g$, necessary to correctly define the witnessing pair $(u,u') \in R$ such that $u \subseq u'$.
\begin{claim}\label{claim:g}
  $g(i) > i$ for all $1 \leq i<m$, and $g(m)=m$.
\end{claim}
\begin{proof}
Let $g(i)=j$ and hence $|\sigma(w[1..i])| = |v_1 \dotsb v_j|$. First, notice that $|v_1 \dotsb v_j| = |\sigma(w[1..i])| \geq |v_1 \dotsb v_i|$ by condition \ref{item:codir:3} of codirectness, and then that $j \geq i$.  If $i < m$,  $|v_1 \dotsb v_i| < |\sigma(w[1..i])|$ by condition \ref{item:codir:4}, and thus $|v_1 \dotsb v_i| < |v_1 \dotsb v_j|$ which implies $i < j$. If $i = m$, then $j=i$ by the fact that $j \geq i=m$.
\end{proof}
\begin{claim}\label{claim:g-monotone}
$g$ is increasing:  $g(i) \geq g(j)$ if $i \geq j$.
\end{claim}
\begin{proof}
Given $m \geq i \geq j \geq 1$, we have that 
\begin{align*}
 |v_1 \dotsb v_{g(i)}| & =|\sigma(w[1..i])| \tag{\text{by definition of $g$}}\\
                              & \geq |\sigma(w[1..j])| \tag{\text{since $i\geq j$}}\\
                              & =|v_1 \dotsb v_{g(j)}| \tag{\text{by definition of $g$}}
\end{align*}
which implies that $g(i) \geq g(j)$.
\end{proof}

\begin{observation} \label{rem:1}
For all $i \leq m$, if $\sigma(w[i]) \in \Gamma$ then $\sigma(w[i]) =
  \sigma'(w[g(i)])$.
\end{observation}

\newcommand{\Gpair}{G}
The most important pairs of positions $(i,j) \in [m] \times [m]$ that witness $u \subseq u'$, are those so that $j = g(i)$ and $\sigma(w[i]) \neq \e$. Once those are fixed, the remaining elements in the definition of $g$ are also fixed. Let us call $\Gpair$ to this set, and let us state some simple facts for later use.
\[
\Gpair = \set{ (i,g(i)) \in [m] \times [m] \mid \sigma(w[i])\in \Gamma }
\] 
\begin{observation} \label{rem:g-injective}
  For every $(i,j), (i',j') \in \Gpair$, if $i\neq i'$ then $j\neq j'$. In other words, $g$ restricted to $\set{i \mid \sigma(w[i]) \in \Gamma}$ is injective.
\end{observation}
\begin{claim}\label{claim:g-two-elem}
  Given $i,j$ with $(i,j) \in \Gpair$ and $i<m$, then $|\sigma(w[i..j])|\geq 2$.
\end{claim}
\begin{proof}
This is because $i < j$ by Claim~\ref{claim:g}, $\sigma(w[i]) \in \Gamma$ by definition of $\Gpair$, and $\sigma(w[j])=\sigma(w[g(i)]) \in \Gamma$ by definition of $g$.
\end{proof}

\smallskip

Since our coding uses the letter $\#$ as some sort of blank symbol, it will be useful to define the factors $\tilde u_1, \tilde u_2, \dotsc$ of $u$ that contain exactly one letter from $\Gamma$. We then define $\tilde u_i$ as the maximal prefix of $u_i \dotsb u_m$ belonging to the following regular expression: $\Gamma \cdot \set\#^*$.

We are now in good shape to define precisely $u_j, u'_j$ for every $j \in [m]$.
For every  $j < m$,
\begin{iteMize}{$\bullet$}
\item if $(i,j) \in \Gpair$ for some $i$, then 
\[u'_j = \tilde u_i \quad \text{and} \quad u_j = \sigma_\#(w[j]) \cdot u'_j[2..]; \text{ and} \]
\item if there is no $i$ so that $(i,j) \in G$, then
  \[(u_j,u'_j) = (\sigma_\#(w[j]), \sigma'_\#(w[j]) ).\] 
\end{iteMize}
And on the other hand, 
  $(u_m,u'_m) = (\sigma_\#(w[m]), \sigma'_\#(w[m]))$ and $u'_{m+1} = \#^{|u_1 \dotsb u_m|}$. Figure~\ref{fig:example-reg-npr} contains an example with all the previous definitions. Notice that the definition of $u_j$ makes use of $\tilde u _j$ and the definition of $\tilde u_j$ seems to make use of $u_j$. We next show that in fact $\tilde u_j$ does not depend on $u_j$, and that the  strings above are well defined.
  \begin{figure}
    \centering
 \includegraphics[width=\textwidth]{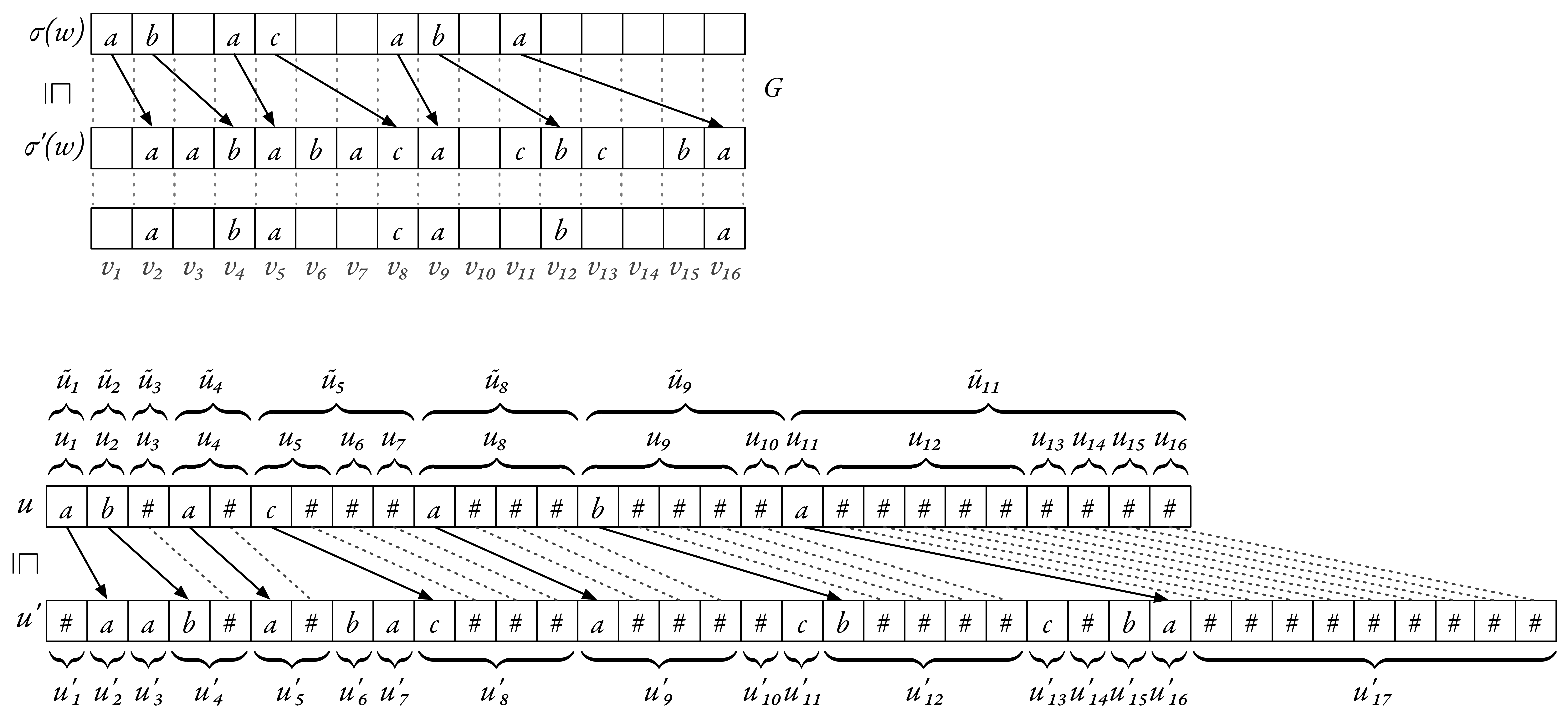}
    \caption{\normalsize Exemplary reduction from $\rpep$ to $\interemptypb{\Reg}{{\subseq}}$, for the case $\sigma(w) = abacaba$, $\sigma'(w) = aababacacbcba$.
}
    \label{fig:example-reg-npr}
  \end{figure}

\begin{observation}\label{rem:tilde-prefix}
For $i<m$, $\tilde u_i$ is a prefix of $u_i \dotsb u_{g(i)-1}$.
\end{observation}
\begin{proof}
By Claim~\ref{claim:g} and Claim~\ref{claim:g-two-elem},  $\sigma(w[i..g(i)])$ contains at least two elements and hence $u_i \dotsb u_{g(i)}$ contains at least two elements from $\Gamma$, namely $u_i[1]$ and $u_{g(i)}[1]$. Then, $\tilde u_i$ cannot contain $u_i \dotsb u_{g(i)-1} \cdot (u_{g(i)}[1])$ as a prefix.
\end{proof}
By the above Observation~\ref{rem:tilde-prefix}, to compute $\tilde u_i$ we only need $u_j$'s and $u'_j$'s with $j<i$, and hence $(u,u')$ is well defined.

\begin{observation}\label{rem:form-ui}
  All the $u_i$'s, $u'_i$'s and $\tilde u_i$'s are of the form $a
  \cdot \# \dotsb \#$ or $\# \dotsb \#$, for $a \in \Gamma$.
\end{observation}

From the definition of $(u,u')$ we obtain the following.
\begin{observation}\label{rem:size-of-ui's-projected}
  For every $n \leq m$,
  \begin{enumerate}[(1)]
  \item\label{item:size-of-ui's-projected:1} $|(u_1 \dotsb u_n)_\Gamma| = \set{i \in [n] \mid \exists j . (i,j) \in \Gpair} = |\sigma(w[1..n])|$, and
  \item\label{item:size-of-ui's-projected:2} $|(u'_1 \dotsb u'_n)_\Gamma| = \set{j \in [n] \mid \exists i . (i,j) \in \Gpair} = |\sigma'(w[1..n])|$.
  \end{enumerate}
\end{observation}

We now show that $(u, u') \in R$ and that $u \subseq u'$.

\begin{claim}\label{claim:u-u'-in-R}
$(u,u') \in R$.
\end{claim}
\begin{proof}
Note that $u_i = \sigma_\#(w[i])$ for all $i$ and then $(u_i)_\Gamma = \sigma(w[i])$. 

We also show that $(u'_i)_\Gamma = \sigma'(w[i])$.
If $u'_j$ is such that there is no $(i,j) \in \Gpair$, or $j=m$, then it is plain that $(u'_j)_\Gamma = \sigma'(w[j])$ by definition of $u'_j$. On the other hand, if $u'_j = \tilde u_i$ for $(i,j) \in \Gpair$, then 
\begin{align*}
(u'_j)_\Gamma &= (\tilde u_i)_\Gamma = (u_i)_\Gamma = (u_i[1])_\Gamma  \tag{by Observation~\ref{rem:form-ui}}\\
&=(\sigma(w[i]))_\Gamma \tag{by def.\ of $u_i$}\\
&=\sigma(w[i]) \tag{since $\sigma(w[i]) \in \Gamma$ by def.\ of $\Gpair$} \\
&= \sigma'(w[g(i)]) = \sigma'(w[j]). \tag{by Observation~\ref{rem:1}}
\end{align*}
Thus, every $(u_i,v_i)$ with $i \leq m$ is such that $(u_i)_\Gamma = \sigma(w[i])$ and $(u'_i)_\Gamma = \sigma'(w[i])$, meaning that  $(u_i,v_i) \in R_{w[i]}$ for every $i \leq m$. Hence, we have that $(u_1 \dotsb u_m, u'_1 \dotsb u'_m) \in R'$ and since $u'_{m+1} \in \set \#^*$, $(u,u') \in R$.
\end{proof}

Next, we prove that $u \subseq u'$, but before doing so, we need an additional straightforward claim.
Let $\set{i_1 < \dotsb < i_{|\Gpair|}} = \set{ i \mid (i,g(i)) \in \Gpair}$. Note that $i_1 = 1$  by Claim~\ref{claim:fst-elem-B}.
\begin{claim}\label{claim:ij-gij}
  $i_{j+1} \leq g(i_j)$
\end{claim}
\begin{proof}
By means of contradiction, suppose $g(i_j) < i_{j+1}$. Then, 
\begin{align*}
  |\sigma(w[1..g(i_j)])|&= |\set{i \in [g(i_j)] \mid \exists j . (i,j) \in \Gpair}| \tag{by Observation~\ref{rem:size-of-ui's-projected}.\ref{item:size-of-ui's-projected:1}}\\
  &= |\set{i \in [g(i_j)] \mid \exists j . (i,j) \in \Gpair}| \tag{since $g(i_j) < i_{j+1}$} \\
  &=|\sigma'(w[1..g(i_j)])|. \tag{by Observation~\ref{rem:size-of-ui's-projected}.\ref{item:size-of-ui's-projected:2}}
\end{align*}
In other words, there is some $k < m$ such that $|\sigma(w[1..k])| = |\sigma'(w[1..k])|$. This is in contradiction with condition \ref{item:codir:4} of strict codirectness. Hence, $g(i_j) \geq i_{j+1}$.
\end{proof}

\begin{claim} \label{claim:u-subseq-u'}
$u \subseq u'$.
\end{claim}
\begin{proof}
We factorize $u = \hat u_1 \dotsb \hat u_{|\Gpair|}$ and we show  that each $\hat u_i$ is a substring of $u'$ that appears in an increasing order.

 We define $\hat u_j = u_{i_j} \dotsb u_{i_{(j+1)}-1}$ for every $j<|\Gpair|$, and $\hat u_{|\Gpair|} = u_{i_{|\Gpair|}} \dotsb u_m$. Hence, the $\hat u_i$'s form a factorization of $u$. Indeed, this is the unique factorization in which each $\hat u_i$ is of the form $b \cdot \# \dotsb \#$ for $b \in \Gamma$.

For every $j<|\Gpair|$, we show that  $\hat u_j \subseq u'_{g(i_j)}$. 
\begin{align*}
  \hat u_j  &= u_{i_j} \dotsb u_{i_{(j+1)}-1} \\
        &\subseq u_{i_j} \dotsb u_{g(i_j)-1} \tag{by Claim~\ref{claim:ij-gij}}\\
        &\subseq \tilde{u}_{i_j} \tag{by Observation~\ref{rem:tilde-prefix}}\\
        &= \tilde u_{g^{-1}(g(i_j))} \tag{by Observation~\ref{rem:g-injective}} \\
        &= u'_{g(i_j)} \tag{by def.\ of $u'$}
\end{align*}
On the other hand, $\hat u_{|\Gpair|} \subseq u'_{g(i_{|\Gpair|})} \cdot u'_{m+1} = u'_m \cdot u'_{m+1}$. By Claim~\ref{claim:g-monotone}, $g$ is increasing. Hence, $u \subseq u'$.
\end{proof}
By Claims~\ref{claim:u-u'-in-R} and \ref{claim:u-subseq-u'}, we conclude that $R \cap {\subseq} \neq \emptyset$.
\end{proof}

\subsubsection{Subsequence-closed relations}
The next question is how far we can extend 
the decidability of 
$\INTP \RAT {\mbox{$\subseq$}}$.
It turns out
that if we allow one projection of a rational relation to be closed
under taking subsequences, then we retain decidability.

Let $R\subseteq\Sigma^*\times\Gamma^*$ be a binary relation. Define
another binary relation
\[R_{\subseq} = \{(u,w) \ | \ u \subseq u' \text{ and }(u',w)\in
R \text{ for some }u'\}\]
Then the class of {\em subsequence-closed relations}, or $\SCR$, is
the class $\{R_{\subseq} \ | \ R\in\RAT\}$. Note that the subsequence
relation itself is in $\SCR$, since it is obtained by closing the
(regular) equality relation under subsequence. That is, $\subseq \ \ =
\ \{(w,w)\ |\ w\in\sigmas\}_{\subseq}$. Not all rational relations are
subsequence-closed (for instance, subword is not). 

The following summarizes properties of subsequence-closed relations.

\begin{prop}
\label{scr-prop}
\hfill
\begin{enumerate}[\em(1)]\itemsep=0pt
\item $\SCR \subsetneq \RAT$. 
\item $\SCR \not\subseteq \REG$ and $\REG \not\subseteq \SCR$.
\item A relation $R$ is in $\SCR$ iff $\{ w \otimes w' \ | \ 
(w,w')\in R\}$ is accepted by an \NFA
$\A=\langle Q,\sigbot\times\sigbot,q_0,\delta,F\rangle$ such that 
$(q,(a,b),q')\in\delta$ implies  $(q,(\bot,b),q')\in\delta$ for all
$q,q'\in Q$ and $a,b\in\sigbot$. We call an automaton with such property a \emph{subsequence-closed automaton}.
\end{enumerate}
\end{prop}

\noindent Note that $(3)$ is immediate by definition of $R_\subseq$,
$(1)$ is a consequence of $(3)$, and $(2)$ is due to the fact that
$\subseq$ is not regular and that, for example, the identity
$\set{(u,u) \mid u \in \Sigma^*}$ is not a subsequence-closed
relation.

When an $\SCR$ relation is given as an input to a problem, we assume
that it is represented as a subsequence-closed automaton as defined in  item (3) in the above proposition. 

Note also that $\INTP \SCR \SCR$ is decidable in polynomial time: if
$R,R'\in\SCR$ and $R\cap R'\neq\emptyset$, then $(\e,w)\in R\cap R'$
for some $w$, and hence the problem reduces to simple \NFA nonemptiness
checking. 

The main result about $\SCR$ relations generalizes decidability of $\INTP \RAT
{\mbox{$\subseq$}}$. 

\begin{thm}
\label{scr-thm}
The problem $\INTP \RAT \SCR$ is decidable, with non-mutiply recursive complexity.
\end{thm}

\newcommand{\OM}{\ensuremath{\mathsf{SCR}}}
\newcommand{\Auta}{\+A_1}
\newcommand{\Autb}{\+A_0}
\newcommand{\Auti}{\+A_i}
\newcommand{\onemany}{subsequence-closed\xspace}

In order to prove Theorem \ref{scr-thm} we use Lemmas
\ref{lem:interpb->synchronizedinterpb} and
\ref{lem:LDC-syn-decidable}, as shown below.  But first we need to
introduce some additional terminology.
  We say that $(\Autb,\Auta)$ is an {\em instance} of
  $\interemptypb{\Rat}{\OM}$ over $\Sigma,\Gamma$ if $\Auta$ is a
  \onemany automaton over $\sigmabot \times \gammabot$, and $\Autb$ is
  a \NFA over $\sigmabot \times \gammabot$.  Given a
  $\interemptypb{\Rat}{\OM}$ instance $(\Autb,\Auta)$ over $\Sigma,
  \Gamma$, we say that $(w_1,w_2)$ is a {\em solution} if $w_1, w_2
  \in (\sigmabot \times \gammabot)^*$, $w_1 \in \lang{\Auta}, w_2 \in
  \lang{\Autb}$.
  We say that a solution $(w_0, w_1)$ of an instance $(\Autb,\Auta)$ over $\Sigma,\Gamma$ is {\em synchronized} if $\wproj_2(w_0)=\wproj_2(w_1)$.
We write $\interemptypbsyn{\Rat}{\OM}$ for the problem of whether there is a synchronized solution.

\begin{lem}\label{lem:interpb->synchronizedinterpb} There is a
polynomial-time reduction from the problem $\interemptypb{\Rat}{\OM}$
into $\interemptypbsyn{\Rat}{\OM}$.  \end{lem}

\begin{proof} We show that $\interemptypb{\Rat}{\OM}$ is reducible to
the problem of whether there exists a synchronized solution of
$\interemptypb{\Rat}{\OM}$. Suppose that $(\Autb,\Auta)$ is an
instance of $\interemptypb{\Rat}{\OM}$ over the alphabets $\Sigma,
\Gamma$. Consider the automata $\+A'_0, \+A'_1$ as the result of
adding all transitions $(q,(\bot, \bot), q)$ for every possible state
$q$ to both automata. It is clear that the relations recognized by
these remain unchanged, and that $\+A'_0$ is still a \onemany
automaton. Moreover, this new instance has a synchronized solution if
there is any, as stated in the following claim.  

\begin{claim} There
is a synchronized solution for $(\+A'_0, \+A'_1)$ if, and only if,
there is a solution for $(\Autb,\Auta)$.  \end{claim} 

The `only if'
part is immediate. For the `if' part, let $(w_0,w_1)$ be a
solution for $(\Autb,\Auta)$. Let $w_0 = w_{0,1} \dotsb w_{0,n}$,
$w_1 = w_{1,1} \dotsb w_{1,n}$ be factorizations of $w_0$ and $w_1$
such that for every $i \in \set{0,1}$, $\wproj_2(w_{i,1})$ is in
$\set\bot^*$; and for each $j>1, i \in {0,1}$, $\wproj_2(w_{i,j})$ is
 in $\Gamma \cdot  \set\bot^*$. It is plain that there is
always such factorization and that it is unique.

For every $j \in [n]$, we define $w'_{0,j} = w_{0,j} \cdot
(\bot,\bot)^{k}$ and $w'_{1,j} = w_{1,j} \cdot (\bot,\bot)^{-k}$, with
$k = |w_{1,j}| - |w_{0,j}|$, where we assume that $(\bot,\bot)^m$ with
$m \leq 0$ is the empty string. We define $w'_0 = w'_{0,1}
\dotsb w'_{0,n}$, $w'_1 = w'_{1,1} \dotsb w'_{1,n}$. Note that $(w'_0, w'_1)$ is a solution of $(\+A'_0,\+A'_1)$ since it is the result of adding letters $(\bot, \bot)$ to $(w_0,w_1)$, which is also a solution of $(\+A'_0,\+A'_1)$.
We have that
$\wproj_2(w'_0) = \wproj_2(w'_1)$, and therefore that $(w'_0,w'_1)$ is
a synchronized solution for $(\+A'_0, \+A'_1)$.  \end{proof}

\begin{lem} \label{lem:LDC-syn-decidable} There is a polynomial-time
reduction from $\interemptypbsyn{\Rat}{\OM}$ into
$\interemptypb{\Rat}{\subseq}$.  \end{lem} 

\begin{proof}
  The problem of finding a synchronized solution for $\+A_0,\+A_1$ can
  be then formulated as the problem of finding words $v, u_0, u_1 \in
  \sigmabot^*$ with $|v|=|u_0|=|u_1|$, so that $(u_0 \otimes v, u_1
  \otimes v)$ is a solution.  We can compute an \NFA $\+A$ over
  $\sigmabot^2 \times \gammabot$ from $\+A_0,\+A_1$, such that $(u_0,u_1,v) \in
  \lang{\+A}$ if, and only if, $u_0 \otimes v \in \lang{\+A_1}$ and
  $u_1 \otimes v\in \lang{\+A_0}$.  Consider now an automaton $\+A'$
  over $\sigmabot^2$ such that $\lang{\+A'} = \set{(u_0,u_1) \mid
    \exists v \ (u_0,u_1,v) \in \lang{\+A} }$. It corresponds to the
  rational automaton of the projection onto the first and second
  components of the ternary relation of $\+A$, and it can be computed
  from $\+A$ in polynomial time. We then deduce that there exists $u_0
  \otimes u_1 \in \lang{\+A'}$ so that $(u_0)_\Sigma \subseq
  (u_1)_\Sigma$ if, and only if, there is $v \in \gammabot^*$ with
  $|v|=|u_0|=|u_1|$ so that $u_0 \otimes v \in \lang{\A_0}$ and $u_1
  \otimes v \in \lang{\A_1}$, where $(u_0)_\Sigma \subseq
  (u_1)_\Sigma$. But this condition is in fact equivalent to $R_0 \cap R_1 \neq
  \emptyset$ (where $R_i = \set{ ((u)_\Sigma, (v)_\Sigma) \mid u \otimes v \in \lang{\A_i}}$), since
  \begin{iteMize}{$\bullet$}
  \item if $((u_1)_\Sigma, (v)_\Sigma) \in R_1$ and $(u_0)_\Sigma
    \subseq (u_1)_\Sigma$, then $((u_0)_\Sigma, (v)_\Sigma) \in R_1$
    (since $R_1 \in \SCR$) and hence $((u_0)_\Sigma, (v)_\Sigma) \in
    R_0 \cap R_1$; and
  \item if $R_0 \cap R_1 \neq \emptyset$, then there exists a
    synchronized solution $(u_0 \otimes v, u_1 \otimes v)$ of
    $\+A_0,\+A_1$; in other words, there are $|v|=|u_0|=|u_1|$ so that $u_0 \otimes v \in \lang{\A_0}$, $u_1
    \otimes v \in \lang{\A_1}$, and $(u_0)_\Sigma = (u_1)_\Sigma$.
  \end{iteMize}
  We have thus reduced the problem to an instance of
 $\INTP \RAT {\mbox{$\subseq$}}$: whether there is $(u,v)$ in the relation denoted by $\+A'$ so that $u \subseq v$.
\end{proof}

\begin{proof}[Proof of Theorem~\ref{scr-thm}] The decidability part of
Theorem~\ref{scr-thm} follows as a corollary of Lemmas
\ref{lem:interpb->synchronizedinterpb} and
\ref{lem:LDC-syn-decidable}, and
Proposition~\ref{prop:subseq-rat-dec-nmr}. Of course the complexity is
non-multiply-recursive, since the problem subsumes $\INTP \REG
{\mbox{$\subseq$}}$ of Theorem \ref{subseq-dec}.  \end{proof}

Coming back to graph logics, we obtain:

\begin{cor}
\label{ecrpq-subseq}
The complexity of evaluation of \ecrpq($\subseq$) queries is not
bounded by any multiply-recursive function.
\end{cor}

Another corollary can be stated in purely language-theoretic terms.

\begin{cor}
\label{nonempty-rel-cor}
Let $\CC$ be a class of binary relations on $\sigmas$ that is closed
under intersection and contains $\REG$. Then the nonemptiness problem
for $\CC$ is:
\begin{iteMize}{$\bullet$}\itemsep=0pt
\item undecidable if $\subw$ or $\suff$ is in $\CC$;
\item non-multiply-recursive if $\subseq$ is in $\CC$.
\end{iteMize}
\end{cor}

\subsection{Discussion} In addition to answering some basic
language-theoretic questions about the interaction of regular and
rational relations, and to providing the simplest yet problem with
non-multiply-recursive complexity, our results also rule out logical
languages for graph databases that freely combine regular relations
and some of the most commonly used rational relations, such as subword
and subsequence. With them, query evaluation becomes either
undecidable or non-multiply-recursive (which means that no realistic
algorithm will be able to solve the hard instances of this problem). 

This does not yet fully answer our questions about the evaluation of
queries in graph logics.
First, in the case of subsequence (or, more generally, $\SCR$
relations) we still do not know if query 
evaluation of \ecrpq s with such relations  is decidable (i.e., what
happens with $\GENINT_S(\REG)$ for such relations $S$). 

Even more importantly, we do not yet know what happens with the
complexity of \crpq s (i.e.,  $\GENINT_S(\REC)$) for various relations
$S$. These questions are answered in the next section.

%%% Local Variables: 
%%% mode: latex
%%% TeX-master: "lrr.tex"
%%% End: 

\section{Restricted logics and the generalized intersection problem}
\label{restr:sec}
The previous section already ruled out some graph logics with rational
relations as either undecidable or decidable with extremely high
complexity. This was done merely by analyzing the intersection problem
for binary rational and regular relations. We now move to the study of 
the generalized intersection problem, and use it to analyze the
complexity of graph 
logics in full generality. We first deal with the generalization of
the decidable case ($\SCR$ relations), and then consider the problem
$\GENINT_S(\REC)$, corresponding to \crpq s extended with relations
$S$ on paths.

\subsection{Generalized intersection problem and subsequence}

We know that $\INTP {\REG} {\mbox{$\subseq$}}$ is decidable, although
not multiply-recursive. What about its generalized version? It turns
out it remains decidable.

\begin{thm}
\label{genintp-subseq-thm}
The problem $\GENINT_{\subseq}(\REG)$ is decidable. That is, there is
an algorithm that decides, for a given 
$m$-ary regular relation $R$ and $I \subseteq [m]^2$, whether $R
\cap_I \mbox{$\subseq$}\neq \emptyset$. 
\end{thm}  

\begin{proof} Let $k \in \N$, $I \subseteq [k] \times [k]$ and $R \in
\Regk{k}$ be an instance of the problem. Let us define $G = \set{(w_1,
\dotsc, w_k) \mid \forall (i,j) \in I, w_i \subseq w_j}$. We show how
to compute if $R \cap G$ is empty or not.
 Let $\+A = (Q,(\sigmabot)^k,q_0,\delta,F)$ be a NFA over
 $(\Sigma_\bot)^k$ corresponding to $R$, for simplicity we assume that
 it is complete. Remember that every $w \in \lang{\+A}$ is such that
 $\wproj_i(w)$ is in $\Sigma^*; \set \bot ^*$ for every $i \in [k]$.

Given $u, v \in \Sigma^*$, we define $u \setminus v$ as $u[i..]$,
where $i$ is the maximal index such that $u[1..i-1] \subseq v$.
In other words, $u \setminus v$ is the result of removing from $u$ the maximal prefix that is a subsequence of $v$.

\newcommand{\tree}{\mathbf{t}}
We define a finite tree $\tree$ whose every node is labeled with
\begin{iteMize}{$\bullet$}
  \item a depth $n \geq 0$,
  \item $k$ words $w_1, \dotsc, w_k \in\sigmabot^n$,
  \item for every $(i,j) \in I$, a word $\alpha_{ij} \in \Sigma^*$, and
  \item a state $q \in Q$.
\end{iteMize}

For a node $x$ we denote these labels by $x.n$, $x.w_1, \dotsc,
x.w_k$, $x.\alpha_{ij}$ for every $(i,j) \in I$ and $x.q$
respectively. The tree is such that the following conditions are met.
\begin{iteMize}{$\bullet$} \item The root is labeled by $x.n=0$, $x.w_1 = \dotsb
= x.w_k = \e$, for very $(i,j) \in I$, $x.\alpha_{ij} =
\e$, and $x.q = q_0$.  \item A node $x$ has a child $y$ in
$\tree$ if and only if
  \begin{iteMize}{$-$}
  \item $y.n = x.n+1$,
  \item $x.w_i = y.w_i[1..y.x-1]$ for every $i \in [k]$,
  \item there is a transition $(x.q, \bar a , y.q) \in \delta$ with
  $\bar a = (y.w_i[y.n])_{i \in [k]}$, and
  \item $y.\alpha_{ij} = (w_i)_\Sigma \setminus (w_j)_\Sigma$ for every $(i,j) \in I$.
  \end{iteMize} \item A node $x$ is a leaf in $\tree$ if and only if
  is final or saturated (as defined below).  \end{iteMize}

A node $x$ is \textbf{final} if $x.q \in F$ and $x.\alpha_{ij} =
\e$ for all $(i,j) \in I$.  It is \textbf{saturated} if it is
not final and there is an ancestor $y \neq x$ such that $y.q = x.q$
and $y.\alpha_{ij} \subseq x.\alpha_{ij}$ for all $(i,j) \in I$.

\begin{lem} \label{lem:kary-tree-finite} The tree $\tree$ is finite
and computable.  \end{lem} 

\begin{proof}
The root is obviously computable, and for every branch, one can compute the list of children nodes of the bottom-most node of the branch. Indeed these are finite and bounded.
The tree $\tree$ cannot have an infinite branch. If there was an infinite branch, then as a result of Higman's Lemma \cum Dickson's Lemma (and the Pigeonhole principle) there would be two nodes $x\neq y$, where $x$ is an ancestor of $y$, $x.q = y.q$, and for all $(i,j) \in I$, $x.\alpha_{ij} \subseq y.\alpha_{ij}$. Therefore, $y$ is saturated and it does not have children, contradicting the fact that $x$ and $y$ are in an infinite branch of $\tree$. Since all the branches are finite and the children of any node are finite, by  K{\H o}nig's Lemma, $\tree$ is finite, and computable.
\end{proof}
\begin{lem}\label{lem:finite-nonempt}
  If $\tree$ has a final node, $R \cap G \neq \emptyset$.
\end{lem}
\begin{proof}\newcommand{\run}{\rho}
  If a leaf $x$ is final, consider all the $x.n$ ancestors of $x$:
  $x_0, \dotsc, x_{n-1}$, such that $x_i.n = i$ for every $i \in
  [n-1]$. Consider the run $\run : [0..x.n] \to Q$ defined as
  $\run(x.n) = x.q$ and $\run(i) = x_i.q$ for $i < x.n$. It is easy to
  see that $\run$ is an accepting run of $\+A$ on $x.w_1 \otimes
  \dotsc \otimes x.w_k$ and therefore that $((x.w_1)_\Sigma, \dotsc,
  (x.w_k)_\Sigma) \in R$. On the other hand, for every $(i,j) \in I$,
  $(x.w_i)_\Sigma \subseq (x.w_j)_\Sigma$ since $\alpha_{ij} =
  \e$. Hence, $((x.w_1)_\Sigma, \dotsc, (x.w_k)_\Sigma) \in G$
  and thus $R \cap G \neq \emptyset$.  \end{proof}
  
\begin{lem}\label{lem:allsat-empt}
  If all the leaves of $\tree$ are saturated, $R \cap G = \emptyset$.
\end{lem}

\begin{proof} \newcommand{\run}{\rho} By means of contradiction
suppose that there is $w = w_1 \otimes \dotsb \otimes w_k \in
(\sigmabot^k)^*$ such $w \in \lang{\+A}$ through an accepting run
$\run : [0..n] \to Q$, and for every $(i,j) \in I$, $(w_i)_\Sigma
\subseq (w_j)_\Sigma$. Let $|w| = n$ be of minimal size.

By construction of $\tree$, the following claims follow.
\begin{claim}
  There is a maximal branch $x_0, \dotsc, x_m$ in $\tree$ such that
  $x_\ell.n = \ell$, $x_\ell.w_j = w_j[1..\ell]$, $x_\ell.q =
  \run(\ell)$ for every $\ell \in [0..m]$ and $j \in [k]$.
  \end{claim}

\begin{claim}\label{cl:subseq-alphaij}
For every $\ell \in [0..m]$ and $(i,j) \in I$,
  \begin{align}
    x_\ell.\alpha_{ij} \cdot (w_i[\ell+1 ..])_\Sigma
    &\subseq  (w_j[\ell+1 ..])_\Sigma ,\label{cl:subseq-alphaij:1}
\\
    (w_i[1.. \ell-|x_\ell.\alpha_{ij}|])_\Sigma &\subseq
    (w_j[1.. \ell])_\Sigma .
\label{cl:subseq-alphaij:2}
  \end{align}
\end{claim}

Since we assume that all the leaves of $\tree$ are saturated, in
particular $x_m$ is saturated and there must be some $m' < m$ such
that $x_m$ and $x_{m'}$ verify the saturation conditions.

Consider the following word.  \[ w' = w[1..m'] \cdot w[m+1..]  \] The
run $\run$ trimmed with the positions $[m'+1 .. m]$ is still an
accepting run on $w'$ (since $\run(m')=\run(m)$), and therefore $((\wproj_1(w'))_\Sigma, \dotsc,
(\wproj_k(w'))_\Sigma) \in R$.

For an arbitrary $(i,j) \in I$, we show that $(\wproj_i(w'))_\Sigma
\subseq (\wproj_j(w'))_\Sigma$. First, note that by
\eqref{cl:subseq-alphaij:2} we have that %
\begin{align*}
  (\wproj_i(w')[1..m' - |x_{m'}.\alpha_{ij}|])_\Sigma &=(w_i[1..m' -
  |x_{m'}.\alpha_{ij}|])_\Sigma \\ &\subseq (w_j[1..m'])_\Sigma
  \tag{by \eqref{cl:subseq-alphaij:2}}\\
  &=(\wproj_j(w')[1..m'])_\Sigma .  \end{align*}
Since $x_{m'}$ and $x_{m}$ verify the saturation conditions,
$x_{m'}.\alpha_{ij} \subseq x_m.\alpha_{ij}$. Therefore,
\begin{align*}
  (\wproj_i(w')[m' - |x_{m'}.\alpha_{ij}|+1 ..])_\Sigma &=
  (\wproj_i(w')[m' - |x_{m'}.\alpha_{ij}|+1 .. m'])_\Sigma \cdot
  (\wproj_i(w')[m' +1..])_\Sigma\\ &=x_{m'}.\alpha_{ij} \cdot (w_i[m
  +1..])_\Sigma\\ &\subseq x_m.\alpha_{ij} \cdot (w_i[m +1..])_\Sigma
  \tag{since $x_{m'}.\alpha_{ij} \subseq x_m.\alpha_{ij}$}\\ &\subseq
  (w_j[m+1 ..])_\Sigma \tag{by
  \eqref{cl:subseq-alphaij:1}}\\
  &=(\wproj_j(w')[m'+1..])_\Sigma \end{align*} Hence, we showed that
  there are some $\ell, \ell'$ such that
  $(\wproj_i(w')[1..\ell])_\Sigma \subseq
  (\wproj_j(w')[1..\ell'])_\Sigma$ and
  $(\wproj_i(w')[\ell+1..])_\Sigma \subseq
  (\wproj_j(w')[\ell'+1..])_\Sigma$, for $\ell = m' -
  |x_{m'}.\alpha_{ij}|$ and $\ell' = m'$. Thus, $(\wproj_i(w'))_\Sigma
  \subseq (\wproj_j(w'))_\Sigma$.

This means that $((\wproj_1(w'))_\Sigma, \dotsc,
(\wproj_k(w'))_\Sigma) \in G$ and thus $((\wproj_1(w'))_\Sigma,
\dotsc, (\wproj_k(w'))_\Sigma) \in R \cap G$. But this cannot be since
$|w'| < |w|$ and $w$ is of minimal length. The contradiction arises
from the assumption that $R \cap G \neq \emptyset$.  Then, $R \cap G =
\emptyset$.  \end{proof}

Hence, by Lemmas~\ref{lem:kary-tree-finite}, \ref{lem:finite-nonempt}
and \ref{lem:allsat-empt}, $R \cap G \neq \emptyset$ if and only if
$\tree$ has a final node, which is computable.  
\end{proof}

\smallskip 

\begin{cor}
The query evaluation problem for \ecrpq($\subseq$) queries is
decidable. 
\end{cor}

Of course the complexity is extremely high as we already know from 
Corollary \ref{ecrpq-subseq}.

\smallskip 

Note that while the intersection problem of $\subseq$ with rational
relations is decidable, as is  $\GENINT_{\subseq}(\REG)$, we lose the
decidability of $\GENINT_{\subseq}(\RAT)$ even in the simplest cases
that go beyond the intersection problem (that is, for ternary
relations in $\RAT$ and any $I$ that does not force two words to be
the same).

\begin{prop}
\label{no-ternary-lemma}
The problem $\INTPI \RAT {\mbox{$\subseq$}} I$ is undecidable even over
ternary relations when $I$ is one of the following: 
  \begin{enumerate}[\em(1)]
  \item\label{lem:k-ary:undec:1}
    $\set{(1,2), (2,3)}$, 
  \item\label{lem:k-ary:undec:2}
    $\set{(1,2), (1,3)}$, or
  \item\label{lem:k-ary:undec:3}
    $\set{(1,2), (3,2)}$.
  \end{enumerate}
\end{prop}

\begin{proof}
The three proofs use a reduction from the PCP problem. Recall that this
is defined as follows. The input are two equally long lists $u_1,u_2,\dots,u_n$
and $v_1,v_2,\dots,v_n$ of strings over alphabet $\Sigma$. The PCP
problems asks 
whether there exists a solution for this input, that
is, a sequence of indices $i_1,i_2,\dots,i_k$ such that $1 \leq i_j
\leq n$ ($1 \leq j \leq k$) and $u_{i_1} u_{i_2} \cdots u_{i_k} =
v_{i_1} v_{i_2} \cdots v_{i_k}$.

\smallskip 

\noindent 
\eqref{lem:k-ary:undec:1} $\set{(1,2), (2,3)}$: 
  The proof goes by reduction from an arbitrary PCP instance given by
lists $u_1,\dots,u_n$ and $v_1,
  \dots,v_n$ of strings over alphabet $\Sigma$. The following relation \[ R = \set{(u_{i_1}
  \dotsb u_{i_m},v_{i_1} \dotsb v_{i_m},u_{i_1} \dotsb u_{i_m}) \mid
  m\in \N \text{ and } i_1, \dotsc, i_m \in [n]} \] is rational and $R
  \cap \set{(x,y,z) \mid x \subseq y \subseq z}$ is non-empty if and
  only if the instance has a solution.

\smallskip 
\noindent
\eqref{lem:k-ary:undec:2} $\set{(1,2), (1,3)}$: 
The proof again goes by reduction from an arbitrary PCP instance given by
lists $u_1,\dots,u_n$ and $v_1,
  \dots,v_n$ of strings over alphabet $\Sigma$. For simplicity, and without any loss of generality, we
  assume that $|u_i|, |v_i| \leq 1$ for every $i$. Let $\hat \Sigma =
  \set{ \hat a \mid a \in \Sigma}$, and for every $w = a_1 \dotsb
  a_\ell \in \Sigma^*$, let $\hat w = \hat a_1 \dotsb \hat
  a_\ell$. Consider \begin{align*}
  R &= \{ (x,y,z) \mid m\in \N \text{, } i_1, \dotsc, i_m \in [n]
  \text{, } w_1, w'_1, \dotsc, w_{m+1}, w'_{m+1}\in \Sigma^*\text{,
  }\\ &\hspace{14ex} x = u_{i_1} \hat v_{i_1} u_{i_2} \hat v_{i_2}
    \dotsb u_{i_m} \hat v_{i_m}, \\
&\hspace{14ex} y = w'_1 \hat u_{i_1} w'_2 \dotsb w'_{m}
    \hat u_{i_m} w'_{m+1}, \\ &\hspace{14ex} z = \hat w_1 v_{i_1} \hat
    w_2 \dotsb \hat w_{m} v_{i_m} \hat w_{m+1} \} \end{align*} which
    is a rational relation. Note that there is some $(x,y,z) \in R$
    with $x \subseq y$ if and only if there is some $v_{i_1} \dotsb
    v_{i_m} \subseq u_{i_1} \dotsb u_{i_m}$. Similarly for $x \subseq
    z$. Therefore, there is $(x,y,z) \in R$ with $x \subseq y$, $x
    \subseq z$ if and only if $v_{i_1} \dotsb v_{i_m} = u_{i_1} \dotsb
    u_{i_m}$ for some choice of $i_1, \dotsc, i_m$.

\smallskip \noindent \eqref{lem:k-ary:undec:3} $\set{(1,2), (3,2)}$:
This is similar to \eqref{lem:k-ary:undec:2}, but this time we
consider the following rational relation.  \begin{align*}
  R &= \{ (x,y,z) \mid m\in \N \text{, } i_1, \dotsc, i_m \in [n]
  \text{, } w_1, w'_1, \dotsc, w_{m+1}, w'_{m+1}\in \Sigma^*\text{,
  }\\ &\hspace{14ex} y = u_{i_1} \hat v_{i_1} u_{i_2} \hat v_{i_2}
    \dotsb u_{i_m} \hat v_{i_m}, \\
&\hspace{14ex} x = w'_1 \hat u_{i_1} w'_2 \dotsb w'_{m}
    \hat u_{i_m} w'_{m+1}, \\
&\hspace{14ex} z = \hat w_1 v_{i_1} \hat w_2 
    \dotsb \hat w_m  v_{i_m} \hat w_{m+1}
 \}
\end{align*}
Analogously as before, there is $(x,y,z) \in R$ with $x \subseq y$, $z \subseq y$ if and only if the PCP instance has a solution.
\end{proof}

\OMIT{
We saw that the decidability of the intersection problem extended from
$\subseq$ to the class of $\SCR$ relations. Does the same hold for the
generalized intersection problem? The answer is negative.

\begin{prop}
\label{genint-scr-prop}
There exists a binary relation $S\in\SCR$ such that
$\GENINT_{S}(\REG)$ is undecidable. In particular, the query
evaluation problem for \ecrpq($S$) is undecidable.
\end{prop}
}
%end omit

\subsection{Generalized intersection problem for recognizable
  relations} 

We now consider the problem of answering \crpq s with rational
relations $S$, or, equivalently, the problem $\GENINT_S(\REC)$. Recall
that an instance of such a problem consists of an $m$-ary
recognizable relation $R$ and a set $I\subseteq [m]^2$. The question
is whether $R \cap_I S\neq \emptyset$, i.e., whether there exists a
tuple $(w_1,\ldots,w_m)\in R$ so that $(w_i,w_j)\in S$ whenever
$(i,j)\in I$. It turns out that the decidability of this problem
hinges on the graph-theoretic properties of $I$. In fact we shall
present a \emph{dichotomy result}, classifying problems $\GENINT_S(\REC)$
into \PSPACE-complete and undecidable depending on the structure of
$I$. 

Before stating the result, we need to decide how to represent a
recognizable relation $R$. Recall that an $m$-ary $R\in\REC$ is a
union of relations of the form $L_1\times\ldots\times L_m$, where each
$L_i$ is a regular language. Hence, as the representation of $R$ we
take the set of all such $L_i$s involved, and as the measure of its
complexity, the total size of NFAs defining the $L_i$s. 

With a set $I\subseteq [m]^2$ we associate an {\em undirected} graph $G_I$
whose nodes are $1,\ldots,m$ and whose edges are $\{i,j\}$ such that
either $(i,j)\in I$ or $(j,i)\in I$. We call an instance of
$\INTPI \REC S I$ {\em acyclic} if $G_I$ is an acyclic graph. 

Now we can state the dichotomy result.

\begin{thm}
\label{acyclic-thm}
\hfill
\begin{iteMize}{$\bullet$} 
\itemsep=0pt
\item Let $S$ be a binary rational relation. Then acyclic instances of
  $\GENINT_S(\REC)$ are decidable in \pspace. Moreover, there is a fixed
  binary relation $S_0$ such that the problem $\INTPI \REC {S_0} I$
  is \pspace-complete. 
\item For every $I$ such that $G_I$ is not acyclic, there exists a
 binary rational relation $S$ such that the problem 
$\INTPI \REC S I$ is undecidable.
\end{iteMize}
\end{thm}

\newcommand{\NN}{{\cal N}}
\newcommand{\M}{{\cal M}}  

\begin{proof} 
For \pspace-hardness we can do an easy reduction from
nonemptiness of the intersection of $m$ given NFA's, which is known to be
\pspace-complete \cite{kozen77}. 
Given $m$ NFAs $\A_1,\dots,\A_m$, define the (acyclic) relation $I = \{(i,i+1)
\mid 1 \leq i < m\}$. Then $\bigcap_i \LL(\A_i)$ 
is nonempty if and only if $\prod_i \LL(\A_i) \cap_I S_0 \neq \emptyset$, where 
$S_0$ is the regular relation $\{(w,w)\ | \ w \in \Sigma^*\}$.

\medskip 

For the upper bound, we use the following idea: First we show how to construct, in
exponential time, the following for 
each $m$-ary recognizable relation $R$, binary
rational relation $S$ and acyclic $I \subseteq [m]^2$: An
$m$-tape automaton $\A(R,S,I)$ that accepts precisely those $\bar w =
(w_1,\dots,w_m) \in (\Sigma^*)^m$ such that $\bar w \in R$ and
$(w_i,w_j) \in S$, for each $(i,j) \in I$. Intuitively, $\A(R,S,I)$
represents the ``synchronization" of the transducer that accepts $R$
with a copy of the 2-tape automaton that recognizes $S$ over each projection
defined by the pairs in $I$. Such synchronization is possible since
$I$ is acyclic. Hence, in order to solve $\GENINT_{S}(\REC)$ we only
need to check $\A(R,S,I)$ for nonemptiness. The latter can be done
in \pspace\ by the standard ``on-the-fly" reachability analysis.
We proceed with the details of the construction below. 

Recall that rational relations are the ones defined by $n$-tape
automata. We start by formally defining 
the class of $n$-tape automata that
we use in this proof. An $n$-tape automaton, $n > 0$, is a tuple $\A =
(Q,\Sigma,Q_0,\delta,F)$, where $Q$ is a finite set of control states,
$\Sigma$ is a finite alphabet, $Q_0 \subseteq Q$ is the set of initial
states, $\delta : Q \times (\Sigma \cup \{\e\})^n \to 2^{Q
\times ([n] \cup \{[n]\})}$ is the transition function with
$\e$ a symbol not appearing in $\Sigma$, and $F \subseteq Q$
is the set of final states. Intuitively, the transition function
specifies how $\A$ moves in a situation when it is in state $q$
reading symbol $\bar a \in \Sigma^n$: If $(q',j) \in \delta(q,\bar
a)$, where $j \in [n]$, then $\A$ is allowed to enter state $q'$ and
move its $j$-th head one position to the right of its tape. If
$(q',[n]) \in \delta(q,\bar a)$ then $\A$ is allowed to enter state
$q'$ and move each one of its heads one position to the right of its
tape.  
%Notice that $R$ is forced to be nonempty, as for technical
%reasons we want to disallow $\A$ to perform transitions on the empty
%word.  The transition function $\delta$ also considers the fresh
%symbol $\e$ as a way to deal with empty words.

Given a tuple $\bar w = (w_1,\dots,w_n) \in (\Sigma^*)^n$ such that
$w_i$ is of length $p_i \geq 0$, for each $1 \leq i \leq n$, a {\em
run} of $\A$ over $\bar w$ is a sequence $q_0 \, P_0 \, q_1 \, P_1 \, \cdots
\, q_{k-1} \, P_{k-1} \, q_{k}$, for $k \geq 0$, such that: 
\begin{enumerate}[(1)]
\item $q_i \in Q$, for each $0 \leq i \leq k$, \item $q_0 \in Q_0$,
\item $P_i$ is a tuple in $([p_1] \cup \{0\}) \times \cdots \times
([p_n]
  \cup \{0\})$, for each $0 \leq i \leq k-1$
  (intuitively, the $P_i$'s represent the positions of the $n$ heads
  of $\A$ at each stage of the run. In particular, the $j$-th
component of $P_i$ represents the position of the $j$-th head of $\A$
in stage $i$ of the run),
\item 
$P_0 = (b_1,\dots,b_n)$, where $b_i := 0$ if $w_i$ is the empty word $\e$
  (that is, $p_i = 0$) and $b_i := 1$ otherwise 
(that is, the run starts by initializing each one
  of the $n$
  heads to be in the initial position of its tape, if possible),  

\item $P_{k-1} = (p_1,\dots,p_n)$, that is, the run ends when each
head scans the last position of its head, and 

%\item
%$P_k := (p,\dots,p)$, 
%\item 
%$q_k \in F$, and 
\item 
for each $0 \leq i \leq k-1$, if $P_i = (r_1,\dots,r_n)$ and 
$$\big(\,(\pi_1(\bar w))[r_1],\, \ldots \, ,(\pi_n(\bar w))[r_n]\,\big) \ = \ 
(a_1,\ldots,a_n),$$ 
where we assume by definition that $w[0] = \e$, 
then 
$\delta(q_i,(a_1,\dots,a_n))$ contains a pair of the form
$(q_{i+1},j)$ such that:  
\begin{enumerate} 
\item 
if $i < k-1$ then $j \in [n]$  
and $P_{i+1}$ is the tuple $(r_1,\dots,r_{j-1}, r_{j} +
1, r_{j+1},\dots,r_n)$. In such case we say that
$(q_{i+1},P_{i+1})$ is a {\em valid transition from $(q_i,P_i)$
over $\bar w$ in the $j$-th head}, and 
\item 
if $i = k - 1$ then $j = [n]$. This is a technical condition
that ensures that each head of $\A$ should leave its tape 
after the last transition in the run is performed. 
\end{enumerate} 

\noindent That is, each run is forced to respect the transition
function $\delta$ when the $n$-tape automaton $\A$ is in 
state $q$ reading the symbols in
the corresponding positions of its $n$ heads. Further, the positions
of the $n$ heads are updated in the run also according to what is
allowed by $\delta$. Notice that each transition in a run moves a
single head, except for the last one that moves all of them at the
same time. 
\end{enumerate}  
The run is {\em accepting} if $q_k \in F$
(that is, $\A$ enters an accepting state after each one of its heads
scans the last position of its own tape). 

Each $n$-tape automaton $\A$ defines the language $L(\A) \subseteq
(\Sigma^*)^n$ of all those $\bar w =$ $(w_1,\dots,$ $w_n) \in (\Sigma^*)^n$
such that there is an accepting run of $\A$ over $\bar w$.  It can be
proved with standard techniques that languages defined by $n$-ary
rational relations are precisely those defined by $n$-tape
automata. Notice that there is an alternative, more general model of
$n$-tape automata that allows each transition to move an arbitrary
number of heads. It is easy to see that this model is equivalent in
expressive power to the one we present here, as transitions that move
an arbitrary number of heads can easily be encoded by a a series of
single-head transitions. We have decided to use this more restricted
version of $n$-tape automata here, as it will allow us simplifying
some of the technical details in our proof.

\medskip 

Now we continue with the proof that the problem $\GENINT_S(\REC)$ can
be solved in \pspace\ if $I$ is acyclic (that is, it defines an
acyclic undirected graph).  The main technical tool for proving this
is the following lemma:

\begin{lem} \label{lemma:m-tape-aut-acyc} Let $R$ be an $m$-ary
relation in $\REC$, $S$ a binary rational relation, and $I$ a subset
of $[m] \times [m]$ that defines an acyclic undirected graph. It is possible to
construct, in exponential time, an $m$-tape automaton $\A(R,S,I)$ such
that the language defined by $\A(R,S,I)$ is precisely the set of words
$\bar w = (w_1,\dots,w_m) \in (\Sigma^*)^m$ such that $\bar w \in R$
and $(w_i,w_j) \in S$ for all $(i,j) \in I$.  \end{lem}

We start by proving the lemma. The intuitive idea is that $\A(R,S,I)$
is an $m$-tape automaton that at the same time recognizes $R$
and represents the ``synchronization'' of the $|I|$ copies of the 
2-tape automaton $S$ over the projections
corresponding to the pairs in $I$. Since $I$ is acyclic, such
synchronization is possible. 

%\medskip 

\medskip

Assume that $|I| = \ell$. Let $t_1,\dots,t_\l$ be an arbitrary
enumeration of the pairs in $I$.  Also, assume that the recognizable
relation $R$ is given as $$\bigcup_i \NN_{i_1} \times \cdots \times
\NN_{i_m},$$ where each $\NN_{i_j}$ is an NFA over $\Sigma$ (without
transitions on the empty word).  
%For each
%$1 \leq j \leq m$, we construct in polynomial time an NFA $\NN_{i_j}$
%over $\Sigma$ that recognizes $L_{i_j}$. 
Assume that the set of states
of $\NN_{i_j}$ is $U_{i_j}$, its set of initial states is $U^0_{i_j}$
and its set of final states is $U^F_{i_j}$.
%We assume that all $U_{i_j}$'s are pairwise
%disjoint. 
Further, assume that the 2-tape transducer $S$ is given by the tuple
$(Q_S,\Sigma,Q^0_S,\delta_S,Q_S^F)$, 
where $Q_S$ is the set of states, the set of initial states is $Q^0_S$, the set of
final states is $Q^F_S$, and $\delta_S : Q_S \times (\Sigma \cup
\{\e\}) \times (\Sigma \cup
\{\e\}) \to 2^{Q \times (\{1,2\} \cup \{\{1,2\}\})}$ is the transition function.  
We take $|I| = \ell$ disjoint copies $S_1,\dots,S_\l$ of $S$, such
that $S_i$, for each $1 \leq i \leq \l$, is the tuple
$(Q_{S_i},\Sigma,Q^0_{S_i},\delta_{S_i},Q^F_{S_i})$. Without loss of
generality we assume that if $t_i = (j,j') \in [m] \times [m]$ then
$\delta_{S_i}$ is a function from $Q_{S_i} \times (\Sigma \cup
\{\e\}) \times (\Sigma \cup
\{\e\})$ into $2^{Q \times (\{j,j'\} \cup \{\{j,j'\}\})}$. We can do this because $I$ is
acyclic, and hence $j \neq j'$.   
%$Q_{S_1},\dots,Q_{S_\ell}$ of $Q_S$. In the same way, we take $\ell$ disjoint copies
%$Q^0_{S_1},\dots,Q^0_{S_\ell}$ of $Q^0_S$, and $\ell$ disjoint copies
%$Q^F_{S_1},\dots,Q^F_{S_\ell}$ of $Q_S^F$.  
%We assume that $Q_S$ is disjoint from
%each $U_{i_j}$. 

The $m$-tape automaton $\A(R,S,I)$ is defined as the tuple
$(Q,\Sigma,Q_0,\delta,F)$, where:
\begin{enumerate}[(1)]
\item The set of states $Q$ is $$\bigcup_i \big( U_{i_1} \times
\cdots \times U_{i_m} \times Q_{S_1} \times \cdots \times
Q_{S_{\ell}} \big).$$ 
\item The initial states in $Q_0$ are precisely those in  
$$\bigcup_i \big( U^0_{i_1} \times
\cdots \times U^0_{i_m} \times Q^0_{S_1} \times \cdots \times
Q^0_{S_{\ell}} \big).$$ 
\item The final states in $F$ are precisely those in $$\bigcup_i \big( U^F_{i_1} \times
\cdots \times U^F_{i_m} \times Q^F_{S_1} \times \cdots \times
Q^F_{S_{\ell}} \big).$$ 

\item The transition function $\delta : Q \times (\Sigma \cup
  \{\e\})^m \to 2^{Q \times ([m] \cup \{[m]\})}$ is defined as follows on
  state $\bar{q} \in Q$ and symbol $\bar a \in (\Sigma \cup
  \{\e\})^m$.  Assume that $\bar q =
  (u_{i_1},\dots,u_{i_m},q_1,\dots,q_\ell)$, where $u_{i_j} \in
  U_{i_j}$ for each $1 \leq j \leq m$, and $q_{j} \in Q_{S_j}$ for
  each $1 \leq j \leq \l$. Further, assume that $\bar a =
  (a_1,\dots,a_m)$, where $a_j \in (\Sigma \cup \{\e\})$ for
  each $1 \leq j \leq m$. Then $\delta(\bar q,\bar a)$ consists of all
  pairs of the form
  $\big((u'_{i_1},\dots,u'_{i_m},q'_1,\dots,q'_\ell), \,j\,\big)$, for
$j
\in [m]$, such
  that:

\begin{enumerate}

%\item $R \neq \emptyset$; 

\item $u'_{i_k} = u_{i_k}$ for each $k \in [m] \setminus \{j\}$, and there is a
transition in $\NN_{i_j}$ from $u_{i_j}$ into $u'_{i_j}$ labeled $a_j$; and 
 
\item
for each $1 \leq k \leq \l$, if $t_k$ is the pair $(k_1,k_2) \in [m]
\times [m]$ then the following holds: (1) If $j \not\in \{k_1,k_2\}$ 
then $q_k = q'_k$, and (2) if 
$j \in \{k_1,k_2\}$ then $(q'_k,j)$ belongs
to $\delta_{S_k}(q_k,(a_{k_1},a_{k_2}))$,  
%where $\delta_{S_j}$ is a copy of the
%transition function of $S$, but defined over $Q_{S_j}$ instead that
%over $Q_S$. 
  
\end{enumerate} 
plus all
  pairs of the form
  $\big((u'_{i_1},\dots,u'_{i_m},q'_1,\dots,q'_\ell), \,[m]\,\big)$ such
  that:

\begin{enumerate}

%\item $R \neq \emptyset$; 

\item for each $1 \leq k \leq m$ there is a
transition in $\NN_{i_k}$ from $u_{i_k}$ into $u'_{i_k}$ labeled $a_k$; and 
 
\item
for each $1 \leq k \leq \l$, if $t_k$ is the pair $(k_1,k_2) \in [m]
\times [m]$ then $(q'_k,\{\{k_1,k_2\}\})$ belongs
to $\delta_{S_k}(q_k,(a_{k_1},a_{k_2}))$.  
%where $\delta_{S_j}$ is a copy of the
%transition function of $S$, but defined over $Q_{S_j}$ instead that
%over $Q_S$. 
  
\end{enumerate} 

\noindent Intuitively, $\delta$ defines possible transitions of
$\A(R,S,I)$ that respect the transition function of each one of the
copies of $S$ over its respective projection. Further, while scanning
its tapes the automaton $\A(R,S,I)$ also checks that there is an $i$
such that for each $1 \leq j \leq m$ the $j$-th tape contains a word
in the language defined by $\NN_{i_j}$.

\end{enumerate}
Clearly, $\A(R,S,I)$ can be constructed in exponential time from $R$,
$S$ and $I$. Notice, however, that states of $\A(R,S,I)$ are of
polynomial size. 

\medskip 

We prove next that for every $\bar w = (w_1,\dots,w_m) \in
(\Sigma^*)^m$ it is the case that $\bar w$ is accepted by $\A(R,S,I)$
if and only if $\bar w$ belongs to the language of $R$ and $(w_i,w_j)
\in S$, for each $(i,j) \in I$. 

\smallskip

\noindent
$\Longrightarrow$) Assume first that $\bar w = (w_1,\dots,w_m) \in (\Sigma^*)^m$ is
accepted by $\A(R,S,I)$. It is easy to see from the way $\A(R,S,I)$ is
defined that, for some $i$, the projection of the accepting run of
$\A(R,S,I)$ on each $1 \leq j \leq m$ defines an accepting run of
$\NN_{i_j}$ over $w_j$. Further, for each $(j,k) \in I$ it is the case
that the projection of the accepting run of $\A(R,S,I)$ on $(j,k)$
defines an accepting run of $S$ over $(w_j,w_k)$. We conclude that
$\bar w$ belongs to the language of $R$ and $(w_j,w_k) \in S$, for
each $(j,k) \in I$.

\smallskip 

\noindent
$\Longleftarrow$) Assume, on the other hand, that $\bar w = (w_1,\dots,w_m) \in
(\Sigma^*)^m$ belongs to the language of $R$ and $(w_i,w_j) \in S$,
for each $(i,j) \in I$. Further, assume that the length of $w_i$ is
$p_i \geq 0$, for each $1\leq i \leq m$. 
We prove next that $\bar w$ is accepted by
$\A(R,S,I)$. 

Since $\bar w \in R$ it must be the case that $\bar w$ is accepted by
$\NN_{i_1} \times \cdots \times \NN_{i_m}$, for some $i$. Let us assume
that $$\rho_{i_j} \ := \ u_{i_j,0} \, (1) \, u_{i_j,1} \, (2) \,
\cdots \, u_{i_j,p_j-1} \, (p_j) \, 
 \, u_{i_j,p_j} $$ is an accepting run of the 1-tape automaton
 $\NN_{i_j}$
 over $w_j$, for each $1 \leq j \leq m$.  Since for every  
$t_j$ ($1
 \leq j \leq \ell$) of the form $(k,k') \in [m]
 \times [m]$ it is the case that $(w_k,w_{k'}) \in S$, there is an
 accepting run $$\lambda_j \ := \ q_{j,0} \, P_{j,0} \,
 q_{j,1} \, P_{j,1} \, \cdots \, q_{j,r_j} \, P_{j,r_j} \, q_{j,r_j +
  1}$$ of $S_j$ over $(w_k,w_{k'})$. We then inductively define a
sequence $$\bar q_0 \, P_0 \, \bar q_1\, P_1 \, \cdots \, $$ where
each $\bar q_j$ is a state of $Q$ and each $P_j$ is a tuple in $([p_1]
\cup \{0\}) \times \cdots \times ([p_m] \cup \{0\})$, as follows:

\begin{enumerate}[(1)]

\item $\bar q_0 :=
  (u_{i_1,0},\dots,u_{i_m,0},q_{1,0},\dots,q_{\ell,0})$. 

\item $P_0 = (b_1,\dots,b_m)$, where $b_i := 0$ if $w_i$ is the empty
  word and $b_i := 1$ otherwise.

\item Let $j \geq 0$. Assume that $\bar q_j =
  (u_{i_1},\dots,u_{i_m},q_{1},\dots,q_{\ell})$, where each $u_{i_k}$
  is a state in $\NN_{i_k}$ and each $q_k$ is a state in $S_k$, and
  that $P_j = (r_1,\dots,r_m) \in ([p_1]
\cup \{0\}) \times \cdots \times ([p_m] \cup \{0\})$. 

If for every $1 \leq k \leq m$ it is the case that $r_k = p_k$ then
the sequence stops. Otherwise it proceeds as follows. 

If for some $1 \leq k \leq m$ it is the case that $u_{i_k} (r_k)$ is
not a subword of the accepting run $\rho_{i_k}$,\footnote{Notice that
$\rho_{i_k}$ is a word in the language defined by $(U_{i_k} \cdot
[p_k])^* \cdot U_{i_k}$, and hence it is completely well-defined whether a
word in $U_{i_k} \cdot [p_k]$ is or not a subword of $\rho_{i_k}$.} or
that for some $1 \leq k \leq \ell$ such that $t_k = (k_1,k_2) \in [m]
\times [m]$ it is the case that $q_k (r_{k_1},r_{k_2})$ is not a
subword of the accepting run $\lambda_k$,\footnote{This is
well-defined for essentially the same reasons given in the previous
footnote.}  then the sequence simply fails.

Otherwise check whether there is a $1 \leq k \leq m$ such that the following
holds: 

\begin{enumerate}

\item $r_k \neq p_k$. 

\item 
For each pair $t_{k_1} \in I$ of the form $(k,k') \in [m] \times [m]$ it is the case
that if $q'_{k_1}
(r'_{k},r'_{k'})$ is the subword in $Q_{S_{k_1}} \cdot ([p_{k}]
\times [p_{k'}])$ that immediately follows $q_{k_1} (r_{k},r_{k'})$ in
the run $\lambda_{k_1}$,\footnote{Notice, since $\A(R,S,I)$ does not allow
empty transitions, that $q'_{k_1} (r'_{k},r'_{k'})$ is well-defined
since the subword $q_{k_1} (r_{k},r_{k'})$ appears exactly once in
the run $\lambda_{k_1}$ and, further, $q_{k_1} (r_{k},r_{k'})$ is
followed in $\lambda_{k_1}$ by a subword in $Q_{S_{k_1}} \cdot ([p_{k}] \times [p_{k'}])$
because $r_k \neq p_k$.} then  
$r'_k = r_k + 1$, and $r'_{k'} = r_{k'}$. 

\item For each pair $t_{k_1} \in I$ of the form $(k',k) \in [m] \times [m]$ it is the case
that if $q'_{k_1}
(r'_{k'},r'_{k})$ is the subword in $Q_{S_{k_1}} \cdot ([p_{k'}]
\times [p_{k}])$ that immediately follows $q_{k_1} (r_{k'},r_{k})$ in
the run $\lambda_{k_1}$, then  
$r'_k = r_k + 1$, and $r'_{k'} = r_{k'}$. 

\end{enumerate} 
Intuitively, this states that we can move the $k$-th head of
$\A(R,S,I)$ and preserve the transitions on each run of the form
$\lambda_{k_1}$ such that $S_{k_1}$ is a copy of $S$ that has one of its
components reading tape $k$. 

If no such $k$ exists the sequence fails. Otherwise pick the least $1
\leq k \leq m$ that satisfies the conditions above,   
and continue the sequence by defining the
   pair $(\bar q_{j+1},P_{j+1})$ as
  $$\big(\,(u_{i_1},\cdots,u_{i_{k-1}},u'_{i_{k}},u_{i_{k+1}},\cdots,u_{i_m},
  q'_1,\cdots,q'_\l), \,
  (r_1,\cdots,r_{k-1},r_{k} + 1,r_{k+1},\cdots,r_m)\,\big),$$ 
where the following holds:

\begin{enumerate}

\item $u'_{i_{k}} (r_{k}+1)$ is the subword 
in $U_{i_{k}} \cdot [p_{k}]$ that 
immediately follows
$u_{i_{k}} (r_{k})$ in $\rho_{i_{k}}$. 

\item For each pair $t_{k_1} \in I$ of the form
$(k,k') \in [m] \times [m]$, it is the case that $q'_{k_1}$ satisfies
that $q'_{k_1}
(r_{k} + 1,r_{k'})$ is the subword in $Q_{S_{k_1}} \cdot ([p_{k}]
\times [p_{k'}])$ that immediately follows $q_{k_1} (r_{k},r_{k'})$ in
the run $\lambda_{k_1}$. 

\item For each pair $t_{k_1} \in I$ of the form
$(k',k) \in [m] \times [m]$, it is the case that $q'_{k_1}$ satisfies
that $q'_{k_1}
(r_{k'},r_{k} + 1)$ is the subword in $Q_{S_{k_1}} \cdot ([p_{k'}]
\times [p_{k}])$ that immediately follows $q_{k_1} (r_{k'},r_{k})$ in
the run $\lambda_{k_1}$.     

\item For each pair $t_{k_1} \in I$ of the form $(k',k'') \in [m]
\times [m]$ such that $k' \neq k$ and $k'' \neq k$, it is the case
that $q'_{k_1} = q_{k_1}$. 

\end{enumerate}
In this
case we say that $(\bar q_{j+1},P_{j+1})$ is {\em obtained from $(\bar
q_{j},P_{j})$ by performing a transition on the $k$-th head}.

\end{enumerate} 

\medskip 

\noindent We first prove by induction the following crucial property of the
sequence $\bar q_0 P_0 \bar q_1 P_1  \cdots$:
The sequence does not fail at any stage $j \geq 0$. Clearly, the
sequence does not fail in stage 0 given by pair $(\bar
q_0,P_0)$. Assume now by induction that the sequence has not failed
until stage $j \geq 0$ given by pair $(\bar q_j,P_j)$, and, further,
that the sequence does not stop in stage $j$. We prove next that the
sequence does not fail in stage $j+1$.  

If the sequence stops in stage $j+1$ it clearly does not fail. Assume
then that the sequence does not stop in stage $(j+1)$.  Also, assume
that $q_j = (u_{i_1},\dots,u_{i_m},q_{1},\dots,q_{\ell})$, where each
$u_{i_k}$
  is a state in $\NN_{i_k}$ and each $q_k$ is a state in
  $S_k$. Further, assume that $P_j = (r_1,\dots,r_m) \in ([p_1] \cup
  \{0\}) \times \cdots \times ([p_m] \cup \{0\})$. Since the sequence
  did not stop in stage $j$ it must be the case that for every $1 \leq
k \leq m$  the sequence $u_{i_k} (r_k)$ is
a subword of the accepting run $\rho_{i_k}$, and 
that for every $1 \leq k \leq \ell$ such that $t_k = (k_1,k_2) \in [m]
\times [m]$ the sequence $q_k (r_{k_1},r_{k_2})$ is a
subword of the accepting run $\lambda_k$. 

Assume that $(\bar q_{j+1},P_{j+1})$ is obtained from $(\bar q_j,P_j)$
by performing a transition on the $k$-th head, for $1 \leq k \leq m$. 
Then the
   pair $(\bar q_{j+1},P_{j+1})$ is of the form:  
  $$\big(\,(u'_{i_1},\cdots,u'_{i_{k}},\cdots,u'_{i_m},
  q'_1,\cdots,q'_\l), \,
  (r'_1,\cdots,r'_{k},\cdots,r'_m)\,\big),$$ 
where the following holds:

\begin{enumerate}[(1)]

\item $u'_{i_{k'}} = u_{i_{k'}}$, for each $k' \in [m] \setminus \{k\}$, 

\item $u'_{i_{k}} (r_{k}+1)$ is the subword 
in $U_{i_{k}} \cdot [p_{k}]$ that 
immediately follows
$u_{i_{k}} (r_{k})$ in $\rho_{i_{k}}$, 

\item 
$r'_{k'} = r_{k'}$, for each $k' \in [m] \setminus \{k\}$,

\item $r'_{k} = r_k + 1$, 

\item for each pair $t_{k_1} \in I$ of the form
$(k,k') \in [m] \times [m]$, it is the case that $q'_{k_1}$ satisfies
that $q'_{k_1}
(r_{k} + 1,r_{k'})$ is the subword in $Q_{S_{k_1}} \cdot ([p_{k}]
\times [p_{k'}])$ that immediately follows $q_{k_1} (r_{k},r_{k'})$ in
the run $\lambda_{k_1}$, 

\item for each pair $t_{k_1} \in I$ of the form
$(k',k) \in [m] \times [m]$, it is the case that $q'_{k_1}$ satisfies
that $q'_{k_1}
(r_{k'},r_{k} + 1)$ is the subword in $Q_{S_{k_1}} \cdot ([p_{k'}]
\times [p_{k}])$ that immediately follows $q_{k_1} (r_{k'},r_{k})$ in
the run $\lambda_{k_1}$, and      

\item for each pair $t_{k_1} \in I$ of the form $(k',k'') \in [m]
\times [m]$ such that $k' \neq k$ and $k'' \neq k$, it is the case
that $q'_{k_1} = q_{k_1}$. 

\end{enumerate}

\noindent Then, by inductive hypothesis, it is the case that for every $k' \in
[m] \setminus \{k\}$ the sequence $u'_{i_{k'}} (r'_{k'})$ is a subword
of the accepting run $\rho_{i_{k'}}$. For the same reason, for every
$1 \leq k' \leq \l$ such that $t_{k'} = (k_1,k_2) \in [m] \times [m]$,
$k_1 \neq k$ and $k_2 \neq k$, it is the case that $q'_{k'}
(r'_{k_1},r'_{k_2})$ is a subword of the accepting run $\lambda_{k'}$.
Further, simply by definition $u'_{i_k} (r'_k)$ is a subword of the
accepting run $\rho_{i_{k}}$. Also, by definition, for each pair
$t_{k_1} \in I$ of the form $(k',k) \in [m] \times [m]$, it is the
case that $q'_{k_1} (r'_{k'},r'_{k})$ is a subword of the accepting
run $\lambda_{k_1}$, and, similarly, for each pair $t_{k_1} \in I$ of
the form $(k,k') \in [m] \times [m]$, it is the case that $q'_{k_1}
(r'_{k},r'_{k'})$ is a subword of the accepting run
$\lambda_{k_1}$. Hence, in order to prove that the sequence does not
fail in stage $j+1$ it is enough to show that there is an $1 \leq h
\leq m$ such that some pair of the form $(\bar q,P)$, where $\bar q
\in Q$ and $P \in ([p_1] \cup \{0\}) \times \cdots \times ([p_m] \cup
\{0\})$, can be obtained from $(\bar q_{j+1},P_{j+1})$ by performing a
transition on the $h$-th head.

Since the sequence does not stop in stage $j+1$, the set $\cal{H}$ $= \{
1 \leq h' \leq m \mid r'_{h'} \neq p_{h'}\}$ must be nonempty. Let 
$h_1$ be the least element in $\cal H$. Since the underlying undirected
graph of $I$ is acyclic, the connected component of $I$ to
which $h_1$ belongs is a tree $T$. Without loss of generality we
assume that $T$ is rooted at $h_1$. 
  
We start by trying to prove that there is pair of the form $(\bar
q,P)$, where $\bar q \in Q$ and $P \in ([p_1] \cup \{0\}) \times
\cdots \times ([p_m] \cup \{0\})$, that can be obtained from $(\bar
q_{j+1},P_{j+1})$ by performing a transition on the $h_1$-th head. 
If this is the case we are done and the proof finishes. Assume
otherwise. Then we can assume without loss of generality that 
there is a pair of the form $t_{k'} \in I$ of the form
$(h_1,h_2) \in [m] \times [m]$ such that the subword in $Q_{S_{k'}} \cdot ([p_{h_1}]
\times [p_{h_2}])$ that immediately follows $q'_{k'} (r'_{h_1},r'_{h_2})$ in
the run $\lambda_{k'}$ is of the form $q''_{k'} (r'_{h_1},r'_{h_2} +
1)$. (That is, the run $\lambda_{k'}$ continues from $q'_{k'}
(r'_{h_1},r'_{h_2})$ by moving its second head). The other possibility
is that there is a pair of the form $t_{k''} \in I$ of the form
$(h_2,h_1) \in [m] \times [m]$ such that the subword in $Q_{S_{k''}} \cdot ([p_{h_2}]
\times [p_{h_1}])$ that immediately follows $q'_{k''} (r'_{h_2},r'_{h_1})$ in
the run $\lambda_{k''}$ is of the form $q''_{k''} (r'_{h_2} +
1,r'_{h_1})$. But this case is completely symmetric to the previous
one.  

We then continue by trying to show that there is pair of the form $(\bar
q,P)$, where $\bar q \in Q$ and $P \in ([p_1] \cup \{0\}) \times
\cdots \times ([p_m] \cup \{0\})$, that can be obtained from $(\bar
q_{j+1},P_{j+1})$ by performing a transition on the $h_2$-th head. If
this is the case then we are ready and the proof finishes. Assume
otherwise. Then again we can assume without loss 
of generality that there is a pair of the form $t_{k''} \in I$ of the form
$(h_2,h_3) \in [m] \times [m]$ such that the subword in $Q_{S_{k''}} \cdot ([p_{h_2}]
\times [p_{h_3}])$ that immediately follows $q'_{k''} (r'_{h_2},r'_{h_3})$ in
the run $\lambda_{k''}$ is of the form $q''_{k''} (r'_{h_2},r'_{h_3} +
1)$. (That is, the run $\lambda_{k''}$ continues from $q'_{k''}
(r'_{h_2},r'_{h_3})$ by moving its second head). 

Since $T$ is acyclic and finite, if we iteratively continue in this way
from $h_2$ we will either have to find some $h \in \cal H$ such that
there is pair of the form $(\bar q,P)$, where $\bar q \in Q$ and $P
\in ([p_1] \cup \{0\}) \times \cdots \times ([p_m] \cup \{0\})$, that
can be obtained from $(\bar q_{j+1},P_{j+1})$ by performing a
transition on the $h$-th head, or we will have to stop in some $h \in
\cal H$ that is a leaf in $T$. But clearly for this $h$ it must be
possible to show that there is pair of the form $(\bar q,P)$, where
$\bar q \in Q$ and $P \in ([p_1] \cup \{0\}) \times \cdots \times
([p_m] \cup \{0\})$, that can be obtained from $(\bar
q_{j+1},P_{j+1})$ by performing a transition on the $h$-th head. This
shows that the sequence does not fail in stage $j+1$. 

\medskip 

We now continue with the proof of the first part of the theorem. Since
the sequence does not fail, and from stage $j$ into stage $j+1$ the
position of at least one head moves to the right of its tape, the
sequence must stop in some stage $j \geq 0$ with associated pair
$(\bar q_j,P_j)$. Then $P_j = (p_1,\dots,p_m)$. Assume that $\bar q_j
= (u_{i_1},\dots,u_{i_m},q_{1},\dots,q_{\ell})$, where each $u_{i_k}$
is a state in $\NN_{i_k}$ and each $q_k$ is a state in $S_k$. Then,
from the properties of the sequence, it must be the case that $u_{i_k}
(p_k)$ appears as a subword in the accepting run $\rho_{i_k}$, for
each $1 \leq k \leq m$, and for each $1 \leq k \leq \ell$ such that
$t_k = (k_1,k_2) \in [m] \times [m]$ it is the case that $q_k
(p_{k_1},p_{k_2})$ appears as a subword in the accepting run
$\lambda_k$. Hence $u_{i_k} = u_{i_k, p_k - 1}$ and $q_k = q_{k,r_k}$.

It easily follows from the definition of the sequence $(\bar q_0,P_0)
(\bar q_1,P_1) \cdots$ and the transition function $\delta$ of
$\A(R,S,I)$, that the following holds for each $k< j$: If $(\bar
q_{k+1},P_{k+1})$ is obtained from $(\bar q_k,P_k)$ by performing a
transition on the $k'$-the head, $1 \leq k' \leq m$, then $(\bar
q_{k+1},P_{k+1})$ is a valid transition 
from $(\bar q_k,P_{k})$ over $\bar w$ in the $k'$-th head. 
Further, assume that 
$$\bar a \ = \ \big(\,(\pi_1(\bar
  w))[p_1],\, \ldots \, ,(\pi_n(\bar w))[p_n]\,\big),$$ then
  $\delta(\bar q_j,\bar a)$ contains a pair of the form $(\bar
  q_{j+1},\{[m]\})$, where: $$\bar q_{j+1} \ := \
  \big(\,u_{i_1,p_1},\cdots,u_{i_m,p_m},q_{1,r_1+1},\cdots,q_{\l,r_\l+1}
  \, \big).$$ Clearly, $\bar q_{j+1} \in F$ (that is, $\bar q_{j+1}$
  is a final state of $\A(R,S,I)$) and we conclude that $\bar q_0 P_0
  \bar q_1 P_1 \cdots
  \bar q_j P_j \bar q_{j+1}$ is an accepting run of $\A(R,S,I)$
  over $\bar w$, which was to be proved.

\medskip 

We now explain how Theorem \ref{acyclic-thm} follows from Lemma \ref{lemma:m-tape-aut-acyc}. The lemma tells us that
in order to solve acyclic instances of $\GENINT_S(\REC)$ we can
construct, from the $m$-ary recognizable relation $R$, the binary
rational relation $S$ and the acyclic $I \subseteq [m] \times [m]$,
the $m$-tape automaton $\A(R,S,I)$, and then check $\A(R,S,I)$ for
nonemptiness. The latter can be done in polynomial time in the size of
$\A(R,S,I)$ by performing a simple reachability analysis in the
states of $\A(R,S,I)$. This gives us a simple exponential time bound
for the complexity of solving acyclic instances of
$\GENINT_S(\REC)$. However, as we mentioned before, each state in
$\A(R,S,I)$ is of polynomial size.  Thus, checking whether $\A(R,S,I)$
is nonempty can be done in nondeterministic \pspace\ by using a
standard ``on-the-fly'' construction of $\A(R,S,I)$ as follows:
Whenever the reachability algorithm for checking emptiness
of $\A(R,S,I)$ wants to move from a state $r_1$ of $\A(R,S,I)$ to a
state $r_2$, it guesses $r_2$ and checks whether there is a transition
from $r_1$ to $r_2$. Once this is done, the algorithm can discard
$r_1$ and follow from $r_2$. Thus, at each step, the algorithm needs
to keep track of at most two states, each one of polynomial size. From
Savitch's theorem, we know that \pspace\ equals
  nondeterministic \pspace. This shows that acyclic instances of
$\GENINT_S(\REC)$ can be solved in \pspace. 

\bigskip

The proof of the second part of the theorem is by an easy reduction
from the PCP problem (e.g. in the style of the proof of the second
part of Theorem~\ref{scr-dichotomy}).  
\end{proof}

\subsection{\crpq s with rational relations}

The acyclicity condition
 gives us a robust class of queries, with an easy
syntactic definition, that can be extended with {\em arbitrary}
rational relations. 
Note that acyclicity is a very standard restriction imposed on
database queries to achieve better behavior, often with respect to
complexity; it is in general known to be easy to enforce
syntactically, and to yield benefits from both the semantics and query
evaluation point of view. This is the approach we follow here.

Recall that \crpq($S$) queries are those of the
form 
\[
\phi(\bar x)  \ =\  \exists \bar y\ \Big( \bigwedge_{i=1}^m (u_i
\Ltoo{\chi_i:L_i} u_i') \ \  \wedge \ \  \bigwedge_{(i,j)\in I}
S(\chi_i,\chi_j)\Big),
\] 
see (\ref{crpqs-eq}) in Sec.\ref{gl:sec}. 
We call such a query {\em acyclic} if $G_I$, the underlying undirected
graph of $I$, is acyclic. %Now we have

\begin{thm}
\label{crpqs-thm}
The query evaluation problem for acyclic \crpq($S$) queries is decidable for
every binary rational relation $S$. Its combined complexity
is \pspace-complete, and data complexity is \nlog-complete.
\end{thm}
\begin{proof}
We provide a nondeterministic \pspace\ algorithm that solves the query
evaluation problem when we assume the query to be part of the input
(i.e. combined complexity). Then the result will follow from Savitch's
theorem, that states that \pspace\ equals nondeterministic \pspace.

Given a graph $G$, a tuple $\bar a$ of nodes, and acyclic \crpq($S$)
query of the form
$$\phi(\bar x) \ =\ \exists \bar y\ \Big( \bigwedge_{i=1}^m (u_i
\Ltoo{\rho_i:L_i} u_i') \ \ \wedge \ \ \bigwedge_{(i,j)\in I}
S(\rho_i,\rho_j)\Big),$$ the algorithm starts by guessing a polynomial
size assignment $\bar b$ for the existentially quantified variables of
$\phi(\bar x)$, that is, the variables in $\bar y$. It then checks
that $G \models \psi(\bar a,\bar b)$, assuming that $\psi(\bar x,\bar
y)$ is the \crpq($S$) formula $$\Big(\bigwedge_{i=1}^m (u_i
\Ltoo{\rho_i:L_i} u_i') \ \ \wedge \ \ \bigwedge_{(i,j)\in I}
S(\rho_i,\rho_j)\Big).$$ If this is the case the algorithm accepts and
declares that $G \models \phi(\bar a)$. Otherwise it rejects and
declares that $G \not\models \phi(\bar a)$.

By using essentially the same techniques as in the proof of Lemma
\ref{crpqs-lemma-one}, one can show that there is a polynomial time
translation that, given $G$ and $\psi(\bar a,\bar b)$, constructs an
acyclic instance of $\GENINT_S(\REC)$ such that the answer to this
instance is `yes' iff $G\models\psi(\bar a,\bar b)$. From Theorem
\ref{acyclic-thm} we know that acyclic instances of $\GENINT_S(\REC)$ can be
solved in \pspace, and hence that the algorithm described above can be
performed in nondeterministic \pspace.

\medskip 

With respect to the data complexity, we start with the following
observation. Acyclic instances of $\GENINT_S(\REC)$ can be solved
in \nlog\ for $m$-ary relations in $\REC$, if we assume $m$ to be
fixed. The proof of this fact mimicks the proof of the \pspace\ upper
bound in Theorem \ref{acyclic-thm}, but this time we assume the arity
of $R$ to be fixed. In such case $\A(R,S,I)$ is of polynomial size,
and each one of its states is of logarithmic size. We can easily check
$\A(R,S,I)$ for nonemptiness in \nlog\ in this case, by 
performing a standard
``on-the-fly'' reachability analysis.

We provide an \nlog\ algorithm that solves the query evaluation
problem when we assume the query to be fixed (i.e. data complexity). 
Consider a fixed acyclic
\crpq($S$) query of the form
$$\phi(\bar x) \ =\ \exists \bar y\ \Big( \bigwedge_{i=1}^m (u_i
\Ltoo{\rho_i:L_i} u_i') \ \ \wedge \ \ \bigwedge_{(i,j)\in I}
S(\rho_i,\rho_j)\Big).$$ Given a graph $G$ and tuple $\bar a$ of
nodes, the algorithm constructs (using the proof of Lemma \ref{crpqs-lemma-one}) in
deterministic logarithmic space an acyclic instance of
$\GENINT_S(\REC)$, given by recognizable relation $R$ of {\em fixed}
arity $m$ (this follows from the fact that $\phi(\bar x)$ is fixed),
and fixed $I \subseteq [m] \times [m]$, such that the answer to this
instance is `yes' iff $G \models \phi(\bar a)$.  Since the arity of
$R$ is fixed, our previous observation tells us that we can solve the
instance of $\GENINT_S(\REC)$ given by $R$ and $I$
in \nlog. But \nlog\ reductions compose, and hence the data
complexity of the query evaluation problem for \crpq($S$) queries is
also \nlog.
\end{proof}

Thus, we get not only the possibility of extending \crpq s with
rational relations but also a good complexity of query evaluation. The
\nlog-data complexity matches that of RPQs, \crpq s, and \ecrpq s
\cite{CM90,CMW87,pods10}, and the combined complexity matches that of
first-order logic, or \ecrpq s without extra relations. 

The next natural question is whether we can recover decidability for
weaker syntactic conditions by putting restrictions on a class of
relations $S$. The answer to this is positive if we consider {\em
  directed} acyclicity of $I$, rather than acyclicity of the underlying
undirected graph of $I$. Then we get 
decidability for the class of 
$\SCR$ relations. In fact, we have a
dichotomy similar to that of Theorem \ref{acyclic-thm}.

\begin{thm}
\label{scr-dichotomy}
\hfill
\begin{iteMize}{$\bullet$}\itemsep=0pt
\item Let $S$ be a relation from $\SCR$. Then $\INTPI \REC S I$
  is decidable in \nexp\ if $I$ is a directed acyclic graph.
\item There is a relation $I$ with a directed cycle and $S\in\SCR$
  such 
that 
$\INTPI \REC S I$ is undecidable.
\end{iteMize}
\end{thm}
\newcommand{\SCRpo}{{\sf SCR}_{\prec}}

\begin{proof}
We start by proving the first item. In order to do that, we first prove
a small model property for the size of the witnesses of the instances
in $\INTPI \REC S I$, when $S$ is a relation in $\SCR$ and $I$ is a DAG. Let
$R$ be an $m$-ary recognizable relation, $m > 0$, and $I \subseteq [m]
\times [m]$ that defines a DAG. Assume that
both $R$ and $S$ are over $\Sigma$. Then the following holds: Assume
$R \cap_I S \neq \emptyset$. There is $\bar w = (w_1,\dots,w_m) \in
(\Sigma^*)^m$ of at most exponential size that is accepted by $R$ and
such that $(w_i,w_j) \in S$, for each $(i,j) \in I$. We prove this
small model property by applying usual cutting techniques.

Assume that $R$ is given as $$\bigcup_i \NN_{i_1} \times \cdots \times
\NN_{i_m},$$ where each $\NN_{i_j}$ is an NFA over $\Sigma$. Further,
assume that $S$ is given as one of the 2-tape NFAs used in the
\pspace\ upper bound of Theorem \ref{acyclic-thm}. That is, $S$
defined by the tuple $(Q_S,\Sigma,Q^0_S,\delta_S,Q_S^F)$, where $Q_S$
is the set of states, the set of initial states is $Q^0_S$, the set of
final states is $Q^F_S$, and $\delta_S : Q_S \times (\Sigma \cup
\{\e\}) \times (\Sigma \cup \{\e\}) \to 2^{Q \times (\{1,2\} \cup
\{\{1,2\}\})}$ is the transition function.  Assume also that there is
$\bar u = (u_1,\dots,u_m) \in (\Sigma^*)^m$ that is accepted by $R$
such that $(u_i,u_j) \in S$, for each $(i,j) \in I$. Then $\bar u$ is
accepted by $\NN_{i_1} \times \cdots \times \NN_{i_m}$, for some $i$.

Since $I$ is a DAG it has a topological order on $[m]$. 
We assume without loss of generality that such
 topological order is precisely the linear order on $[m]$. 
We prove
 the following invariant on $1 \leq \l \leq m$: There exists $\bar w =
 (w_1,\dots,w_m) \in (\Sigma^*)^m$ such that (1) $\bar w$ is accepted
 by $R$, (2) $(w_j,w_k) \in S$, for each $(j,k) \in I$, and 
(3) each $w_{\l'}$ with $\l' \leq \l$ is of at
most exponential size. Clearly this proves our small model property on
$\l = m$. The proof is by induction.

The basis case is $\l = 1$. We start from $\bar u$ and ``cut" its
first component in order to satisfy the invariant.  By
using standard pumping techniques it is possible to show that there is
a subsequence $w_1$ of $u_1$ of size at most $O(|\NN_{i_1}|)$ that is accepted
by $\NN_{i_1}$. Clearly the tuple $(w_1,u_2,\dots,u_m)$ belongs to
$R$. Further, for each pair of the form $(1,j)$ in $I$ it is the case
that $(w_1,u_j) \in S$. This is the case because $(u_1,u_j) \in S$,
$u_1 \subseq w_j$ and $S \in \SCR$. Notice that we do
not need to consider pairs of the form $(j,1)$ since we are assuming
that the linear order on $[m]$ is a topological order of $I$. This
implies that $(w_1,u_2,\dots,u_m)$ satisfies our invariant on $\l =
1$.

Assume now that the invariant holds for $\l < m$. Then there exists $\bar w =
 (w_1,\dots,w_m) \in (\Sigma^*)^m$ such that (1) $\bar w$ is accepted
 by $R$, (2) $(w_j,w_k) \in S$, for each $(j,k) \in I$, and 
(3) each $w_{\l'}$ with $\l' \leq \l$ is of at
 most exponential size.  We proceed to ``cut" $w_{\l+1}$ while
 preserving the invariant. Let $I(\l+1)$ be $\{ 1 \leq j \leq \l \mid
 (j,\l+1) \in I\}$. Let $\rho_j$ be an accepting
 run of $S$ over
 $(w_j,w_{\l+1})$, for each $j \in I(\l+1)$. Further, let $\cal{P}$ be
 the set of all positions $1 \leq k \leq |w_{\l+1}|$ such that for
 some $j \in I(\l+1)$ the accepting run $\rho_j$ contains a subword of
 the form $q \, (k',k) \, q' \, (k'+1,k)$, where $q,q' \in Q_S$ and $1 \leq k'
 \leq |w_j|$. That is, $\cal{P}$ defines the set of positions over
 $w_{\l+1}$, in which the accepting run $\rho_j$ of $S$ over
 $(w_j,w_{\l+1})$, for some $j \in I(\l+1)$, makes a move on the head
 positioned over $w_j$.
 Intuitively, these are the positions of $w_{\l+1}$ that should not be
 ``cut" in order to maintain the invariant. Notice that the size of
 $\cal{P}$ is
 bounded by $s := \Sigma_{1 \leq \l' \leq \l} |w_{\l'}|$, and hence from the
 inductive hypothesis the size of $\cal{P}$ is exponentially bounded.

By using standard pumping techniques it is possible to show that there
is a subsequence $w'_{\l+1}$ of $w_{\l+1}$ of size at most
$|\NN_{i_{\l+1}}| \cdot |\cal{P}| \cdot 
|I(\l+1) \cdot |Q_S| \cdot |\Sigma| + 2$, such that
$w'_{\l+1}$ is accepted by $\NN_{i_{\l+1}}$ and $(w_j,w'_{\l+1})$ is
accepted by $S$, for each $j \in I(\l+1)$. 
Assume this is not the case, and 
that the shortest subsequence $w'_{\l+1}$ of $w_{\l+1}$ that satisfies this condition 
is of length strictly bigger than $|\NN_{i_{\l+1}}| \cdot |\cal{P}| \cdot 
|I(\l+1)| \cdot |Q_S| \cdot |\Sigma| + 2$. Then there exist two positions $1 \leq i <
j \leq |w_{\l+1}|$ such that (i) $k \not \in \cal{P}$, for each $i
\leq k \leq j$, (ii) the labels of $i$ and $j$ in $w_{\l+1}$ coincide,
(iii) the run $\rho_s$ assigns the same state to both $i$ and $j$, for
each $s \in I(\l+1)$, and (iv) some accepting run of $\cal{N}_{\l+1}$
assigns the same state to both $i$ and $j$. Let $w''_{\l+1}$ be the
subsequence of $w'_{\l+1}$ that is obtained by cutting all positions
$i \leq k \leq j-1$. Clearly, $w''_{\l+1}$ is shorter than $w'_{\l+1}$
and is accepted by $\cal{N}_{\l+1}$. Further, $(w_s,w''_{\l+1})$ is
accepted by $S$, for every $s \in I(\l+1)$. This is because 
$(w_s,w'_{\l+1})$ is invariant with respect to the accepting run
$\rho_s$, for each $s \in I(\l+1)$, as the cutting does not include
elements in $\cal{P}$ (that is, we only cut elements in which $\rho_s$
does not need to synchronize with the head positioned over $w_s$) and
$\rho_s$ assigns the same state to both $i$ and $j$, which have, in addition,
the same label. This is a contradiction. 

We claim that $\bar w' =
(w_1,\dots,w_\l,w'_{\l+1},w_{\l+2},\cdots,w_m) \in (\Sigma^*)^m$
satisfies the invariant. Clearly, $\bar w'$ is accepted by $R$ since
$w'_{\l+1}$ is accepted by $\NN_{i_{\l+1}}$ and, by inductive
hypothesis, $w_j$ is accepted by $\NN_{i_j}$, for each $j \in [m]
\setminus \{\l+1\}$. Further, simply by definition it is the case that
$(w_j,w'_{\l+1}) \in S$, for each $j \in I(\l+1)$. Moreover,
$(w'_{\l+1},w_j) \in S$, for each $(\l+1,j) \in I$, simply because
$w'_{\l+1} \subseq w_{\l+1}$ and $S \in \SCR$. The remaining pairs
in $I$ are satisfied by induction hypothesis. Finally, $w'_{\l+1}$ is
of size at most $O(|\NN_{i_{\l+1}}| \cdot |\cal{P}| \cdot |I(\l+1)|
\cdot |Q_S| \cdot |\Sigma|)$, and hence, by inductive hypothesis, it
is of size at most exponential. By inductive hypothesis, each
$w_{\l'}$ with $\l' \leq \l$ is of size at most exponential.

\smallskip 

It is now simple to prove the first part of the theorem using the
small model property. In fact, in order to check whether $R \cap_I S
\neq \emptyset$, for $S \in \SCR$, we only need to guess an
exponential size witness $\bar w$, and then check in polynomial time
that it satisfies $R$ and each projection in $I$ satisfies $S$. This
algorithm clearly works in nondeterministic exponential time.

\medskip 

Now we prove the second item. We reduce from the PCP problem. Assume
that the input to PCP are two equally long lists $a_1,a_2,\dots,a_n$
and $b_1,b_2,\dots,b_n$ of strings over alphabet $\Sigma$. Recall that
we want to decide whether there exists a solution for this input, that
is, a sequence of indices $i_1,i_2,\dots,i_k$ such that $1 \leq i_j
\leq n$ ($1 \leq j \leq k$) and $a_{i_1} a_{i_2} \cdots a_{i_k} =
b_{i_1} b_{i_2} \cdots b_{i_k}$.

Assume without loss of generality that $\Sigma$ is disjoint from
$\mathbb{N}$. Corresponding to every input $a_1,a_2,\dots,a_n$ and
$b_1,b_2,\dots,b_n$ of PCP over alphabet $\Sigma$, we define the
following: 

\begin{iteMize}{$\bullet$}

\item An alphabet $\Sigma(n) := \Sigma \cup
\{1,2,\dots,n\}$;  

\item a regular language $R_{a,n} := (\bigcup_{1 \leq i \leq n} a_i
\cdot i)^*$; 

\item a regular language $R_{b,n} := (\bigcup_{1 \leq j \leq n} b_j
\cdot j)^*$. 

\end{iteMize}

Consider a ternary recognizable relation $R$ over alphabet $\Sigma(n)
\cup \{\star,\dagger\}$, where $\star$ and $\dagger$ are symbols not
appearing in $\Sigma(n)$, defined as $$\big(\star \cdot \Sigma^*\big)
\, \times \, \big(\dagger \cdot R_{a,n}\big) \, \times \, \big(\dagger
\cdot R_{b,n}\big).$$ Further, consider a binary relation $S$ over
$(\Sigma(n) \cup \{\star,\dagger\})^*$ defined as the union of the
following sets:

\begin{enumerate}[(1)]

\item $\{(w,w') \in (\dagger \cdot (\Sigma(n))^*) \times (\dagger \cdot
(\Sigma(n))^*) \, \mid \, \text{$w_{\{1,\dots,n\}} \subseq w'_{\{1,\dots,n\}}$}\}$. 

%\item $\{(w,w') \in (\star \cdot \Sigma^*) \times (\star \cdot
%\Sigma^*) \, \mid \, \text{$w_{\Sigma} \subseq w'_\Sigma$\}}$.

\item $\{(w,w') \in (\dagger \cdot (\Sigma(n))^*) \times (\star \cdot
\Sigma^*) \, \mid \, \text{$w_{\Sigma} \subseq w'_\Sigma$\}}$.

\item $\{(w,w') \in (\star \cdot \Sigma^*) \times (\dagger \cdot
(\Sigma(n))^*) \, \mid \, \text{$w_{\Sigma} \subseq w'_\Sigma$\}}$.

\end{enumerate} 
The intuition is that $S$ takes care that indices in 
the sequences are consistent. 
It is easy to see that $S$ is a rational
relation, which implies that $S_\subseq$ is in $\SCR$.

From input $a_1,\dots,a_n$ and $b_1,\dots,b_n$ to the PCP problem, we
construct an instance of $\GENINT_{S_\subseq}(\REC)$ defined by the recognizable
relation $R$ and $$I \ = \
\{(1,2),(2,1),(1,3),(3,1),(2,3),(3,2)\}.$$ We claim that
$R \cap_I S \neq \emptyset$ if and only if the PCP instance given by
lists $a_1,\dots,a_n$ and $b_1,\dots,b_n$ has a solution. 

\medskip 

Assume first that $R \cap_I S \neq \emptyset$. Hence there are words
$w_1 \in (\star \cdot \Sigma^*)$, $w_2 \in (\dagger \cdot R_{a,n})$
and $w_3 \in (\dagger \cdot R_{b,n})$, such that $(w_i,w_j)$ belongs
to $S_{\subseq}$, for each $(i,j) \in I$. Since $(2,3) \in I$, it must
be the case that $(w_2,w_3)$ belongs to $S_\subseq$. Thus, since
the first symbol of both $w_2$ and $w_3$ is $\dagger$, it must be the
case that $(w_2)_{\{1,\dots,n\}} \subseq (w_3)_{\{1,\dots,n\}}$. For
the same reasons, and given that $(3,2) \in I$, it must be the case
that $(w_3)_{\{1,\dots,n\}} \subseq (w_2)_{\{1,\dots,n\}}$. We
conclude that $(w_2)_{\{1,\dots,n\}} = (w_3)_{\{1,\dots,n\}}$. 

Since $(1,2) \in I$, it must be the case that $(w_1,w_2)$ belongs to
$S_\subseq$. Thus, since the first symbol of $w_1$ is $\star$ and the
first symbol of $w_2$ is $\dagger$, it must be the case that
$(w_1)_{\Sigma} \subseq (w_2)_{\Sigma}$. For the same reasons, and
given that $(2,1) \in I$, it must be the case that $(w_2)_{\Sigma}
\subseq (w_1)_{\Sigma}$. We conclude that $(w_1)_{\Sigma} =
(w_2)_{\Sigma}$.

Mimicking the same argument, but this time using the fact that
$\{(1,3),(3,1)\} \subseteq I$, we conclude that $(w_1)_{\Sigma} =
  (w_3)_{\Sigma}$. 
But then $(w_2)_{\Sigma} = (w_3)_{\Sigma}$ (because
  $(w_1)_{\Sigma} = (w_2)_{\Sigma}$). 

Assume $(w_2)_{\{1,\dots,n\}} = (w_3)_{\{1,\dots,n\}} = i_1 i_2
\cdots i_n$, where each $i_j \in [n]$. Then from the fact
that $(w_2)_{\Sigma} = (w_3)_{\Sigma}$ we conclude that $a_{i_1}
a_{i_2} \cdots a_{i_n} = b_{i_1} b_{i_2} \cdots b_{i_n}$, and hence
that the instance of the PCP problem given by $a_1,\dots,a_n$ and
$b_1,\dots,b_n$ has a solution.

\medskip 

The other direction, that is, that the fact that the instance of the
PCP problem given by $a_1,\dots,a_n$ and $b_1,\dots,b_n$ has a
solution implies that $R \cap_I S \neq \emptyset$, can be proved using
the same arguments. 
\end{proof} 

\smallskip 

In particular, if we have a \crpq($S$) query of the form 
$$\exists \bar y\ \Big( \bigwedge_{i=1}^m (u_i
\Ltoo{\chi_i:L_i} u_i') \ \  \wedge \ \  \bigwedge_{(i,j)\in I}
S(\chi_i,\chi_j)\Big),$$ where
$I$ is acyclic (as a directed graph) and $S\in\SCR$, then query
evaluation has \nexp\ combined complexity. 

The proof of this result is quite different from the upper bound proof
of Theorem \ref{acyclic-thm}, since the set of witnesses for the
generalized intersection problem is no longer guaranteed to be rational
without the undirected acyclicity condition. Instead, here we establish the
finite-model property, which implies the result. 

\medskip 

Also, as a corollary to the
proof of Theorem \ref{scr-dichotomy}, we get the following result:

\begin{prop}
\label{po-scr-prop}
Let $S\in\SCR$ be a partial order. Then $\GENINT_S(\REC)$ is decidable
in \nexp. 
%In particular, \crpq($S$) queries can be evaluated with
%\nexp\ combined complexity. 
\end{prop}
\begin{proof}  
As in the previous proof, we start
by proving a small model property for the size of the witnesses of the
instances in $\GENINT_S(\REC)$, for $S$ a partial order in $\SCR$. Let
$R$ be an $m$-ary recognizable relation, $m > 0$, and $I \subseteq [m]
\times [m]$. Assume that both $R$ and $S$ are over $\Sigma$. Then the
following holds: Assume $R \cap_I S \neq \emptyset$. There is $\bar w
= (w_1,\dots,w_m) \in (\Sigma^*)^m$ of at most exponential size that
is accepted by $R$ and such that $(w_i,w_j) \in S$, for each $(i,j)
\in I$. We prove this small model property by applying usual cutting
techniques.

Assume that $R$ is given as $$\bigcup_i \NN_{i_1} \times \cdots \times
\NN_{i_m},$$ where each $\NN_{i_j}$ is an NFA over $\Sigma$. Further,
assume that $S$ is given as the 2-tape transducer $S$ defined by the
tuple $(Q_S,\Sigma,Q^0_S,\delta_S,Q_S^F)$, where $Q_S$ is the set of
states, the set of initial states is $Q^0_S$, the set of final states
is $Q^F_S$, and $\delta_S : Q_S \times (\Sigma \cup \{\e\})
\times (\Sigma \cup \{\e\}) \to 2^{Q \times (\{1,2\} \cup
\{\{1,2\}\})}$ is the transition function.
Assume also that there is $\bar u = (u_1,\dots,u_m) \in (\Sigma^*)^m$
that is accepted by $R$ and such that $(u_i,u_j) \in S$, for each
$(i,j) \in I$. Then $\bar u$ is accepted by $\NN_{i_1} \times \cdots
\times \NN_{i_m}$, for some $i$.

Let $I^+$ be the transitive closure of $I$. Notice, since $S$ defines
a partial order over $\Sigma^*$, that $(u_j,u_k) \in S$, for each
$(j,k) \in I^+$. Further, for every pair $(j,k) \in [m] \times [m]$
such that $\{(j,k),(k,j)\} \subseteq I^+$ we must have that $u_j =
u_k$. We need to maintain such equality when applying our cutting
techniques over $\bar u$. In order to do that we define an equivalence
relation $\mathcal{E}_I$ over $[m]$ as follows: $$\mathcal{E}_I \ := \
\{(j,k) \in [m] \times [m] \mid j = k \text{ or } \{(j,k),(k,j)\}
\subseteq I^+\}.$$ Hence $\mathcal{E}_I$ contains all pairs $(j,k) \in
[m] \times [m]$ such that $I$ implies $u_j = u_k$. Take the quotient
$[m]/\mathcal{E}_I$, and consider the restriction
$I([m]/\mathcal{E}_I)$ of $I$ over $[m]/\mathcal{E}_I$, defined in the
expected way: $([j]_{\mathcal{E}_I},[k]_{\mathcal{E}_I}) \in
I([m]/\mathcal{E}_I)$ if and only if $(j',k') \in I$, for some $j'\in
[j]_{\mathcal{E}_I}$ and $k' \in [k]_{\mathcal{E}_I}$. Notice that
$I([m]/\mathcal{E}_I)$ defines a DAG over $[m]/\mathcal{E}_I$.

Consider now a new input to $\GENINT_S(\REC)$, given this time by
$I([m]/\mathcal{E}_I) \subseteq ([m]/\mathcal{E}_I) \times
([m]/\mathcal{E}_I)$, and the recognizable relation $R'$ defined as
$$\prod_{[j]_{\mathcal{E}_I} \in [m]/\mathcal{E}_I}
\M_{i}^{[j]_{\mathcal{E}_I}},$$ where $\M_{i}^{[j]_{\mathcal{E}_I}} =
\bigcap_{k \in [j]_{\mathcal{E}_I}} \NN_{i_k}$.  Notice that this new
input may be of exponential size in the size of $R$.  

Assume that
$[m]/\mathcal{E}_I$ consists of $p \leq m$ equivalence classes and,
without loss of generality, that these correspond to the first $p$
indices of $[m]$. Hence each product in $R'$ is of the form $\prod
\M_{i_1} \times \cdots \times \M_{i_p}$, where $\M_{i_j}$ is defined
as the intersection of all NFAs in the equivalence class
$[j]_{\mathcal{E}_I}$. Also, $I([m]/\mathcal{E}_I)$ is the restriction
of $I$ to $[p] \times [p]$.  Then it must be the case that
 $(u_1,\dots,u_{p}) \in (\Sigma^*)^p$ belongs to $R'$ and 
 $(u_j,u_k) \in S$, for each
 $(j,k) \in I([m]/\mathcal{E}_I)$. Further, from
 every witness to the fact that $R' \cap_{I([m]/\mathcal{E}_I)} S \neq
 \emptyset$ we can construct in polynomial time a witness to the fact
 that $R \cap_I S \neq \emptyset$. Hence, in order to prove our small
 model property it will be enough to prove the following: There is
 $\bar w = (w_1,\dots,w_p)
 \in (\Sigma^*)^p$ of at most exponential size (in $R$) that is accepted by
 $R'$ and such that $(w_j,w_k) \in S$, for each $(j,k) \in
 I([m]/\mathcal{E}_I)$.

The latter can be done by mimicking the inductive proof of the first
part of Theorem \ref{scr-dichotomy}. We only have to deal now with the
issue that some of the NFAs that define $R'$ may be exponential in the
size of $R$. However, by following the inductive proof one observes
that this is not a problem, and that the same exponential bound holds
in this case.  

\smallskip 

It is now simple to prove the first part of the theorem using the
small model property. In fact, in order to check whether $R \cap_I S
\neq \emptyset$, for $S$ a partial order in $\SCR$, we only need to
guess an exponential size witness $\bar w$, and then check in
exponential time that it satisfies $R$ and each projection in $I$
satisfies $S$. This algorithm clearly works in nondeterministic
exponential time.  
\end{proof}

By applying similar techniques to those in the proof of 
Theorem \ref{crpqs-thm} we obtain the following.  

\begin{cor}
\label{crpq-subsec-cor} If $S\in\SCR$ is a partial order, then 
\crpq($S$) queries can be evaluated with
\nexp\ combined complexity. In particular, 
 \crpq($\subseq$) queries have 
\nexp\ combined complexity.
\end{cor}

We do not have at this point a matching lower bound for the complexity
\crpq($\subseq$) queries. Notice that an easy \pspace\ lower bound
follows by a reduction from the intersection problem for NFAs, as the 
one presented in the proof of Theorem \ref{acyclic-thm}.  

\smallskip 

The last question is whether these results can be extended to other
relations considered here, such as subword and suffix. We do not know
the result for subword (which appears to be hard), but we do have a
matching complexity bound for the suffix relation.

\begin{prop}
\label{crpq-suff-prop}
The problem
$\GENINT_{\suff}(\REC)$ is decidable in \nexp. In particular, \crpq($\suff$)
queries can be evaluated with \nexp\ combined complexity. 
\end{prop}
\begin{proof}
We only prove that $\GENINT_{\suff}(\REC)$ is decidable
in \nexp. The fact that \crpq($\suff$)
queries can be evaluated with \nexp\ combined complexity follows
easily from this by applying the same techniques as in the proof of
Theorem \ref{crpqs-thm}.   

We start by proving a small model property for the size of the
witnesses of the instances in $\GENINT_{\suff}(\REC)$. Let $R$ be an
$m$-ary recognizable relation, $m > 0$, and $I \subseteq [m] \times
[m]$. Assume that both $R$ and
$\suff$ are over $\Sigma$. Then the following holds: Assume it is the
case that  
$R \cap_I \{\suff\}
\neq \emptyset$. There is $\bar w = (w_1,\dots,w_m) \in (\Sigma^*)^m$
of at most exponential size that is accepted by $R$ and such that
$w_i \suff w_j$, for each $(i,j) \in I$. We prove this small model
property by applying cutting techniques.

Assume that $R$ is given as $$\bigcup_i \NN_{i_1} \times \cdots \times
\NN_{i_m},$$ where each $\NN_{i_j}$ is an NFA over $\Sigma$. We assume,
without loss of generality, that $I$ defines a DAG over $[m] \times
[m]$. In fact, assume otherwise; that is, $I$ does not define a DAG over $[m]
\times [m]$. Since $\suff$ defines a partial order over $\Sigma^*$, we
can always reduce in polynomial time the instance of
$\GENINT_{\suff}(\REC)$ given by $R$ and $I$ to an ``equivalent''
instance of $\GENINT_{\suff}(\REC)$ given by recognizable relation
$R'$ of arity $m' \leq m$ and $I' \subseteq [m'] \times [m']$ such
that $I'$ defines a DAG. We already showed how to do this for an
arbitrary partial order over $\Sigma^*$ in the proof of Proposition 
\ref{po-scr-prop}, so we prefer not to repeat the argument here, and simply
assume that $I$ defines a DAG over $[m] \times [m]$. Since
$I$ defines a DAG it has a topological order over $[m]$. We assume
without loss of generality that such topological order is precisely
the linear order on $[m]$.

Assume then that there is $\bar u = (u_1,\dots,u_m) \in (\Sigma^*)^m$
that is accepted by $R$ and such that $u_i \suff u_j$, for each $(i,j)
\in I$. Then $\bar u$ is accepted by $\NN_{i_1} \times \cdots \times
\NN_{i_m}$, for some $i$. Assume that the length of $u_j$ is $p_j \geq
0$, for each $1 \leq j \leq m$. Our goal is to ``cut'' $\bar u$ in
order to obtain an exponential size witness to the fact that $R \cap_I
\{\suff\} \neq \emptyset$. 

We recursively define the set $\M_k$ of {\em marked}
positions in string $u_k$, $1 \leq k \leq m$, as follows:

\begin{iteMize}{$\bullet$}

\item No position in $u_1$ is marked. 

\item For each $1 < k \leq m$ the set $\mathcal{M}_k$ of marked positions in
  $u_k$ is defined as the union of the marked positions in $u_k$ {\em with
  respect to $j$}, for each $j < k$ such that $(j,k) \in I$, where the
  latter is defined as follows.  Assume that $\mathcal{M}_j$ is the set of marked
  positions in $u_j$.  Then the set $\mathcal{M}_k$ 
of positions $1 \leq \l \leq p_k$
  that are marked in $u_k$ with respect to $j$ is $\{r + p_k - p_j
  \mid \text{$r = 1$ or $r \in \mathcal{M}_j$}\}$. (Notice that $p_k - p_j \geq 0$ since $u_j
  \suff u_k$, and hence $1 \leq r + p_k - p_j \leq p_k$ for each $r
  \in \M_j$ and for $r = 1$).

\end{iteMize} Intuitively, $\M_k$ consists of those positions $1 \leq
\l \leq p_k$ such that for some $j < k$ with $(j,k) \in I^+$, where
$I^+$ is the transitive closure of $I$, it is the case that that $u_k
= u_k[1,\l-1] \cdot u_j$. Or, in other words, the fact that $u_j \suff
u_k$ starts to be ``witnessed'' at position $\l$ of $u_k$. We assume
the $\M_k$'s to be linearly ordered by the restriction of the linear
order $1 < 2 < \cdots < m$ to $\M_k$.  By a simple inductive argument
it is possible to prove that the size of $\M_k$ is polynomially
bounded in $m$, for each $1 \leq k \leq m$.

\newcommand{\Lho}[1]{\stackrel{#1}{\rightharpoonup}}
\newcommand{\LHo}[1]{\stackrel{#1}{\rightleftharpoons}}

Since $u_j \suff u_k$, for each $(j,k) \in I$, this implies that the
labels in some positions of $u_j$ are preserved in the respective
positions of $u_k$ that witness the fact that $u_j \suff u_k$.  The
important thing to notice is that, since we are dealing with $\suff$,
the following holds: For each position $p$ that is ``copied'' from
$u_j$ into $u_k$ in order to satisfy $u_j \suff u_k$, the distance
from $p$ to the last element of $u_j$ equals the distance from the
copy of $p$ in $u_k$ to the last position of $u_k$. That is, distances
to the last element of the string are preserved when copying positions
(and labels) in order to satisfy $I$. 
We need to take
care of this information when ``cutting'' $\bar u$ in order to obtain
an exponential size witness for the fact that $R \cap_I \{\suff\} \neq
\emptyset$. In order to do this we define for each 
$0 \leq r \leq \max {\{p_k \mid 1 \leq k \leq m\}}$, a binary relation $\Lho{r}$ on
$\{u_1,\dots,u_m\}$ such that $u_j \Lho{r} u_{k}$ if $p_j - r > 0$ 
and $(j,k) \in I$. This implies that position $p_j - r$ of
$u_j$ is ``copied'' as position $p_{k} - r$ of $u_{k}$ in order to
satisfy the fact that $u_j \suff u_{k}$. 

But in order to consistently ``cut'' $\bar u$, we need to preserve the
suffix relation both with respect to forward and backward edges of the
graph defined by $I$.  In order to do that we define $\LHo{r}$ as
$(\Lho{r} \cup \, (\Lho{r})^{-1})$.  Further, since $\suff$ is a
partial order over $\Sigma^*$, and hence it defines a transitive
relation, it is important for us also to consider the transitive
closure $(\LHo{r})^+$ of the binary relation $\LHo{r}$. Intuitively,
$u_j (\LHo{r})^+ u_k$, for $1 \leq j,k \leq m$, if position $p_j - r$
of $u_j$ has to be ``copied'' into position $p_k - r$ of $u_k$ in
order for $\bar u$ to satisfy the pairs in $I$ with respect to
$\suff$.

Let $t := |\NN_{i_1}| \cdot |\NN_{i_2}| \cdots |\NN_{i_m}|$ and $s :=
(\sum_{1 \leq k \leq m} |\M_k|) + 1$. We claim the following: 
There is
$\bar w = (w_1,\dots,w_m) \in (\Sigma^*)^m$ such that: (1) $\bar w$ is
accepted by $R$, (2) $w_i \suff w_j$, for each $(i,j) \in I$, and (3)
for each $1 \leq k \leq m$
the number of positions in $w_k$ between any two consecutive positions in
$\M_k$ is bounded by $s \cdot t \cdot 2^m \cdot |\Sigma|^m$. 
This clearly implies our small model property.

Assume that $\bar u$ does not satisfy this. Then there exists $1 \leq
j \leq m$ and two consecutive positions $p$ and $p'$ in $\M_j$, such
that the number of positions in $u_j$ between $p$ and $p'$ is bigger
than $s \cdot t \cdot 2^m \cdot |\Sigma|^m$. But this implies that
there are two positions $p_j - r$ and $p_j - r'$ ($r > r'$) 
between $p$ and $p'$ in $u_j$ such that the following hold:

\begin{enumerate}[(1)]

\item $\{1 \leq k \leq m \mid u_j (\LHo{r})^+ u_k\} = \{1 \leq k \leq
  m \mid u_j (\LHo{r'})^+ u_k\}$. Intuitively, this says that the set
  of strings in which position $p_j - r$ of $u_j$ is ``copied''
  coincides with the set of strings in which position $p_j - r'$ of
  $u_j$ is ``copied''. 

\item For each $k$ such that $u_j (\LHo{r})^+ u_k$ it is the case that
  neither $p_k - r$ nor $p_{k} - r'$ is a marked position in
  $\M_k$, and there is no marked position in $\M_k$ in between $p_k -
  r$ and $p_{k} - r'$ in $u_k$.

\item The state assigned by the accepting run of $\NN_{i_j}$ over $u_j$
  to
  position $p_j - r$ of $u_j$ is the same than the one assigned to
  position $p_j - r'$. 

\item The state assigned by the accepting run of $\NN_{i_k}$ over $u_k$
  to
  the ``copy'' $p_k - r$ of position $p_j - r$ over $u_k$, for each 
$k$ such that $u_j (\LHo{r})^+ u_k$, is the 
 same than the one assigned to the ``copy'' $p_k - r'$ of position $p_j -
 r'$ over $u_k$.

\item The symbol in
  position $p_j - r$ of $u_j$ is the same as the symbol in 
  position $p_j - r'$ of $u_j$.   

\item For each $k$ such that $u_j (\LHo{r})^+ u_k$ it is the case that
  the symbol in position $p_k - r$ of $u_k$ is the same as the
  symbol in position $p_k - r'$ of $u_k$. 

\end{enumerate}
Intuitively, this states that if we ``cut'' the string $u_j$ from
position $p_j - r + 1$ to $p_j - r'$, and string $u_k$ from position $p_k
- r + 1$ to $p_k - r'$, for each $k$ such that $u_j (\LHo{r})^+ u_k$, then
the resulting $\bar u' = (u'_1,\dots,u'_m) \in (\Sigma)^m$ satisfies
the following: (1) $\bar u'$ is accepted by $R$, and (2) for each
$(j,k) \in I$ it is the case that $u'_j \suff u'_k$. We formally
prove this below. Notice for the time being that this implies our
small model property. Indeed, if we recursively apply this procedure
to $\bar u$ we will end up with $\bar w = (w_1,\dots,w_m) \in
(\Sigma^*)^m$ such that: (1) $\bar w$ is accepted by $R$, (2) $w_j
\suff w_k$, for each $(j,k) \in I$, and (3) for each $1 \leq k \leq m$
the number of positions in $w_k$ between any two consecutive positions
in $\M_k$ is bounded by $s \cdot t \cdot 2^m \cdot |\Sigma|^m$.

Let $\bar u' = (u'_1,\dots,u'_m) \in (\Sigma)^m$ be the result of
applying once the cutting procedure described above to $\bar u =
(u_1,\dots,u_m)$, starting from string $\bar u_j$ by cutting positions
from $p_j - r + 1$ to $p_j - r'$ ($r > r'$). It is not hard to see that $\bar u'$ is
accepted by $R$, since each $u_k$ has been cut in a way that is
invariant with respect to the accepting run of $\NN_{i_k}$ over
$u_k$. Assume that $(\l,k) \in I$. We need to prove that $u'_\l \suff
u'_k$. If $u_\l = u'_\l$ and $u_k = u'_k$ then $u'_\l \suff u'_k$ by
assumption. Assume then that at least one of $u_\l$ and $u_k$ has been
cut. Suppose first that $u_\l$ has been cut from position $p_\l - r + 1$
to position $p_\l - r'$ in order to obtain $u'_\l$. Then $u_j
(\LHo{r})^+ u_\l$ and $u_j (\LHo{r'})^+ u_\l$. Clearly, it is also the
case that $u_\l \LHo{r} u_k$ and $u_\l \LHo{r'} u_k$, which implies
that $u_j (\LHo{r})^+ u_k$ and $u_j (\LHo{r'})^+ u_k$. Thus,
$u_k$ is also cut from position $p_k - r + 1$ to $p_k - r'$ in order
to obtain $u'_k$, and hence $u'_\l \suff u'_k$. Suppose, on the
other hand, that $u_\l$ has not been cut but $u_k$ has been cut from
position $p_k - r + 1$ to position $p_k - r'$ in order to
obtain $u'_k$. We consider three cases:

\begin{enumerate}[(1)]

\item $r' > p_j - 1$. Then clearly $u'_k \suff u'_j$.  

\item $r' \leq p_j - 1$ and $r > p_j - 1$. This cannot be the case
  since then either $p_k - r'$ is a marked position in $\M_k$ (when
  $r' = p_j - 1$), or $p_k - r$ and $p_k - r'$ have a marked position
  in $\M_k$ in between (namely, $p_k - p_j + 1$). Any of these 
  contradicts the fact that a cutting of $u_k$ could be applied from
  position $p_k - r$ to position $p_k - r'$ in order to obtain $u'_k$. 

\item $r' < p_j - 1$ and $r \geq  p_j - 1$. Similar to the previous
  one. 

\item $r < p_j - 1$. But then clearly $u_\l \LHo{r} u_k$ and $u_\l
  \LHo{r'} u_k$, which implies that $u_j (\LHo{r})^+ u_\l$ and $u_j
  (\LHo{r'})^+ u_\l$. This implies that $u_\l$ should have also been
  cut from position $p_\l - r$ to position $p_\l - r'$ in order to
  obtain $u'_\l$, which is a contradiction.

\end{enumerate} 

\medskip

\noindent We can finally prove the theorem using the small model property. In
fact, in order to check whether $R \cap_I \{\suff\} \neq \emptyset$ we
only need to guess an exponential size witness $\bar w$, and then
check in polynomial time that it satisfies $R$ and each projection in
$I$ satisfies $\suff$. This algorithm clearly works in
nondeterministic exponential time.
\end{proof}

%%% Local Variables: 
%%% mode: latex
%%% TeX-master: "lrr"
%%% End: 

\section{Conclusions}
\label{concl:sec}

\begin{figure*}
\begin{center}
{%\small
\begin{tabular}{|c|c|c|c|} \hline
& $R\in\REC$ &  $R\in\REG$ &  $R\in\RAT$ \\ \hline 
$\INT R {\mbox{$\subw$}}$ &  & undecidable
& undecidable \\ \cline{1-1}\cline{3-4} 
$\INT R {\mbox{$\suff$}}$ & \ptime\ (cf.~\cite{berstel}) & {undecidable}
& undecidable \\ \cline{1-1}\cline{3-4} 
$\INT R {\mbox{$\subseq$}}$ &  & decidable,
NMR &  decidable,
NMR \cite{CS-fsttcs07} \\ \hline\hline
$\INTPI R {\mbox{$\subw$}} I$ & ? & undecidable
&  \\ \cline{1-3}
$\INTPI R {\mbox{$\suff$}} I$ & \nexp & undecidable
& undecidable \\ \cline{1-3}
$\INTPI R {\mbox{$\subseq$}} {I_{}}$ & \nexp & decidable,
NMR &   \\ \hline
\end{tabular}

\vspace{7mm}

\begin{tabular}{|c|c|c|c|c|} \hline
& $S\ =\ \subseq$ &  $S\ = \ \suff$ & 
$S\ =\ \subw$ &  $S$ arbitrary in $\RAT$ \\ \hline 
\ecrpq($S$) &  decidable, NMR & undecidable
& undecidable & undecidable \\ \hline
\crpq($S$) & \nexp & \nexp & ? & undecidable \\ \hline
acyclic \crpq($S$) & \pspace & \pspace & \pspace & \pspace \\ \hline
\end{tabular}
}
\caption{Complexity of the intersection and generalized intersection
problems, and combined complexity of graph queries for subword
($\subw$), suffix ($\suff$), and 
subsequence ($\subseq$) relations. NMR stands for non-multiply-recursive lower bound.}
\label{summary-fig}
\end{center}
\end{figure*}

Motivated by problems arising in studying logics on graphs
(as well as some verification problems), we studied the intersection
problem for rational relations with recognizable and regular relations
over words. We have looked at rational relations such as subword
$\subw$, suffix $\suff$, and subsequence $\subseq$, which are often
needed in graph 
querying tasks. The main results on the complexity of the intersection
and generalized intersection problems, as well as the combined
complexity of evaluating different classes of logical queries over
graphs are summarized in Fig.~\ref{summary-fig}. Several results
generalizing those (e.g., to the class of $\SCR$ relations) were also
shown. Two problems related to the interaction of the subword relation
with recognizable relations remain open and
appear to be hard.

\OMIT{
While problems $\INT R S$ and $\INT R S$ are known to
be decidable in polynomial time for every rational $S$ when $R$ ranges
over $\REC$ \cite{berstel}, and $\INT R {\mbox{$\subseq$}}$ is known
to be decidable in non-multiply-recursive time when $R$ ranges over
$\RAT$ \cite{CS-fsttcs07}, we completed the picture by showing that
$\INT R {\mbox{$\subw$}}$ is undecidable when $R$ ranges over $\REG$
or $\RAT$, and $\INT R {\mbox{$\subseq$}}$ is non-multiply-recursive
even when $R$ ranges over $\REG$. Moreover, these results extend to
classes of relations $S$ that share structural properties with $\subw$
and $\subseq$. 

We then looked at the generalized version of the intersection
problem. The undecidability results for subword carry over to it; for
$\GENINT_{\subseq}(\R)$ results differ depending on the class $\R$: if
it is $\REG$, the problem is decidable, if it is $\RAT$, the problem
is undecidable. 

In terms of graph logics, this tells us that \ecrpq($\subw$) is
undecidable, and \ecrpq($\subseq$) is decidable, although
non-multiply-recursive. We then looked at languages \crpq($S$), and
showed that queries in which $S$-comparisons are acyclic, can always
be answered with \nlog\ data and \pspace\ combined
complexity. Non-acyclic queries here can be undecidable. 

}

{From} the practical point of view, as  
rational-relation comparisons are demanded by many
applications of graph data, our results essentially say that such
comparisons should not be used together with regular-relation
comparisons, and that they need to form acyclic patterns (easily
enforced syntactically) for efficient evaluation.

So far we dealt with the classical setting of graph data
\cite{AG-survey,CGLV00,CGLV00b,CM90,CMW87} in which the model of data
is that of a graph with labels from a finite alphabet. In
both graph data and verification problems it is often necessary to
deal with the extended case of infinite alphabets (say, with graphs
holding data values describing its nodes), and languages that query
both topology and data have been proposed recently
\cite{1-in-3-bitch,icdt12}. A natural question is to extend the
positive results shown here to such a setting. 

%%% Local Variables: 
%%% mode: latex
%%% TeX-master: "lrr"
%%% End: 

\small{

}  

%\end{document}

%\appendix
%\onecolumn
%\normalsize
%\noindent
%{\Huge Appendix: Proofs}
%\input{proofs4}
%\input{app}

\end{document}